\documentclass[12pt]{article}

\usepackage{amssymb,amsmath,amsfonts,amssymb}
\usepackage{graphics,graphicx,color}
\usepackage{empheq}
\usepackage{dsfont}

\def\@abssec#1{\vspace{.05in}\footnotesize \parindent .2in
{\bf #1. }\ignorespaces}

\graphicspath{{/EPSF/}{Figures/}}
\DeclareGraphicsExtensions{.eps}

\newtheorem{theorem}{Theorem}[section]
\newtheorem{lemma}[theorem]{Lemma}
\newtheorem{proposition}[theorem]{Proposition}
\newtheorem{corollary}[theorem]{Corollary}

\newtheorem{remark}[theorem]{Remark}

\def \Rm {{\mathbb R}}
\def \Nm {\mathbb N}
\def \Cm {\mathbb C}
\def \Zm {\mathbb Z}
\def \Sm {\mathbb S}
\def \Tm {\mathbb T}
\def \Xm {\mathbb X}
\newcommand{\eps}{\varepsilon}

\newcommand{\dsum}{\displaystyle\sum}
\newcommand{\dint}{\displaystyle\int}

\newcommand{\pdr}[2]{\dfrac{\partial{#1}}{\partial{#2}}}

\newcommand{\aver}[1]{\langle {#1} \rangle}

\newcommand{\mI}{\mathcal I}

\newcommand{\mP}{\mathcal P}

\newcommand{\fa}{{\mathfrak a}}

\newcommand{\fh}{{\mathfrak h}}

\newcommand{\fv}{{\mathfrak v}}
\newcommand{\fT}{{\mathfrak T}}

\renewcommand{\Pr}{{\mathcal P}}
\newcommand{\Qr}{{\mathcal Q}}
\newcommand{\Tc}{{\mathcal T}}
\newcommand{\Un}{{\mathcal U}}
\newcommand{\Pin}{\mathsf{P}}
\newcommand{\Qin}{\mathsf{Q}}
\newcommand{\Uin}{\mathsf{U}}
\newcommand{\Win}{\mathsf{W}}

\newcommand{\Ch}[1]{{\rm Ch}_{\frac12{#1}}}

\newcommand{\cout}[1]{}

\newcommand{\sgn}[1]{\,{\rm sign}(#1)}

\newcommand{\Tr}{{\rm Tr}}
\newcommand{\tr}{{\rm tr}}

\newcommand{\Ind}[1]{{\rm Index}\ (#1)}

 \renewcommand{\arraystretch}{1.5}

\title{Continuous bulk and interface description of topological insulators}

\author{Guillaume Bal \thanks{University of Chicago; {\tt guillaumebal@uchicago.edu}}}

\begin{document}
 
\maketitle


\begin{abstract}
	We analyze continuous partial differential models of topological insulators in the form of systems of Dirac equations. We describe the bulk and interface topological properties of the materials by means of indices of Fredholm operators constructed from the Dirac operators by spectral calculus. We show the stability of these topological invariants with respect to perturbations by a large class of spatial heterogeneities. These models offer a quantitative tool to analyze the interplay between topology and spatial fluctuations in topological phases of matter. The theory is first presented for two-dimensional materials, which display asymmetric (chiral) transport along interfaces. It is then generalized to arbitrary dimensions with the additional assumption of chiral symmetry in odd spatial dimensions.
\end{abstract}
 

\renewcommand{\thefootnote}{\fnsymbol{footnote}}
\renewcommand{\thefootnote}{\arabic{footnote}}

\renewcommand{\arraystretch}{1.1}





\section{Introduction}

Topological insulators are materials whose phase transitions are characterized by topological invariants rather than by symmetries and their spontaneous breaking. These phases display fundamental properties that are stable with respect to continuous changes in the material parameters so long as the topological invariant is defined. Two dimensional examples include the quantum Hall effect and the quantum spin Hall effect, which display unusual transport properties of electronic edge states at the interface between insulators. Several higher-dimensional topological insulators have also been analyzed theoretically and observed experimentally. For references on topological insulators, see e.g.,  \cite{bernevig2013topological,Bernevig1757,chiu2016classification,Fruchart2013779,PhysRevB.76.045302,PhysRevLett.98.106803,PhysRevLett.61.2015,RevModPhys.82.3045,PhysRevLett.95.146802,RevModPhys.83.1057}. Similar effects may be observed in photonic and mechanical structures as well \cite{fang2012realizing,Fleury2014,hafezi2011robust,hafezi2013imaging,PhysRevLett.100.013904,khanikaev2013photonic,lu2014topological,nash2015topological,PhysRevA.78.033834,rechtsman2013photonic,wang2009observation}.

In this paper, we model the topology of the insulators with Dirac Hamiltonian representatives. These low-energy descriptions in the vicinity of Dirac points are ubiquitous in the physical literature and accurately display the topology of more general models; see \cite[Section III.C]{chiu2016classification}. The models take the form of systems of Dirac equations in the spatial domain $\Rm^d$. As such, they are amenable to a large class of spatially-varying (random) perturbations, which do not need to satisfy any invariance by (discrete) translations and thus provide a reasonably simple and quantitative tool to assess the effect of topology on the material's properties.

Constant-coefficient Dirac operators characterize the bulk properties of topological insulators. The practically interesting transport properties of such materials appear when two bulk materials with different topologies are brought next to each other. Transport along the resulting interface often displays unusual asymmetric properties, which are moreover topological in nature and therefore stable with respect to large classes of perturbations. We introduce interface Hamiltonians as continuous transitions between two bulk Hamiltonians. 

Our objective is to describe the topology of the above bulk and interface Hamiltonians, to understand their correspondence, and to obtain their stability with respect to perturbations. In doing so, we obtain explicit models in which the influence of random perturbations can be analyzed quantitatively. Such models were used in \cite{B-EdgeStates-2018} to obtain a quantitative description of transport, localization, and back-scattering of modes along a one-dimensional interface between two-dimensional topological insulators. 

The Dirac operators are symmetric unbounded partial differential operators with no clear a priori topological characterization. In the absence of spatial variations, or when the latter satisfy some periodicity assumptions, the Dirac operators can be represented by a Floquet-Bloch transformation that involves a family of Hamiltonians parametrized on a compact Brillouin zone. The eigenspaces associated to the Hamiltonians then take the form of vector bundles over the Brillouin zone, whose topology may be non-trivial. This is the standard way of assigning topological invariants to materials; see e.g \cite{chiu2016classification,RevModPhys.82.3045,RevModPhys.83.1057}.

In this commutative setting, because composition of operators reduces to commutative multiplication on the Brillouin zone, the effect of random perturbations is rendered cumbersome by the hypothesis of invariance by (discrete) translations. The non-commutative setting therefore involves replacing the Fourier multipliers on the Brillouin zone by non-commutative operators on the physical variables. The resulting operators then typically act on functions of the lattice $\Zm^d$, or as in the present paper, the continuum $\Rm^d$, and are amenable to a much larger class of spatial perturbations. This effect was one of the driving forces behind the development of the non-commutative setting \cite{avron1994,bellissard1986k,bellissard1994noncommutative,bourne2018chern,prodan2016bulk}. Following these references, we assign topological invariants to the Dirac operators by means of Fredholm operators constructed by spectral calculus. First obtained for scalar operators, \cite{avron1994,bellissard1986k,bellissard1994noncommutative,kellendonk2004quantization}, the non-trivial construction of Fredholm operators in higher-dimensional and vectorial cases is also understood in many settings \cite{prodan2016bulk}. We adapt the constructions of the aforementioned references to continuous systems of Dirac operators.

In the analysis of the quantum Hall effect, randomness and the resulting electron localization are an essential effect in the understanding of (several) conductivity plateaus \cite{avron1994,bellissard1986k,bellissard1994noncommutative,elgart2005equality} in the sense that randomness influences the topological index of the material. Such random fluctuations are typically required to have statistics that are spatially homogeneous; see, however \cite{elgart2005equality}. In contrast, several engineered topological insulators may first be represented in a `clean', unperturbed state, which is then perturbed by random fluctuations that are quite arbitrary so long as they do not break the topological invariant. This is the setting considered here. We show that an index may be assigned to systems of Dirac operators and provide a means to calculate such an index. An interesting observation, already made for instance in \cite{Bernevig1757,lu2010massive,RevModPhys.83.1057}, is that the Dirac operators, which typically appear as first-order systems, need to be regularized by a second-order term, or their spectral calculus properly modified, in order for the topological invariants to be defined. 

Once the invariants are defined, both for bulk and for interface Hamiltonians, and shown to be in appropriate correspondence, we also obtain that they are indeed stable with respect to arbitrary perturbations that are relatively compact in an appropriate sense. This was the route followed in \cite{B-EdgeStates-2018} to obtain topology-consistent quantitative descriptions of back-scattering effects. 

A brief recall on the derivation of the Dirac systems of equations is proposed in section \ref{sec:model}. The analysis of bulk and interface Hamiltonians for two-dimensional materials is presented in section \ref{sec:2d}. We generalize to the vectorial setting the results obtained in \cite{avron1994} for a scalar operator with magnetic field. The topology acquired by the constant magnetic field in the quantum Hall effect is replaced here by the `twisting' nature of the vectorial Dirac operator, which is the effect observed in the Haldane model \cite{PhysRevLett.61.2015}. The derivation is related to properties of Fredholm modules, indices of pairs of projections, and trace-class operators, which are summarized in Appendix \ref{sec:FH}. The inclusion of random effects is handled by means of a standard Helffer-Sj\"ostrand formulation of the spectral calculus and is presented in detail in section \ref{sec:stab}. Generalizations to higher dimensions are given in section \ref{sec:nd}, with conventions on Clifford algebras and representations of Dirac systems of operators postponed to Appendix \ref{sec:CA}.

Let us also mention some related works on continuous models for the microscopic description of the topology of Schr\"odinger (and other) operators with perturbed periodic potential \cite{Fefferman2016}, or the recent derivation of topological invariants for materials with aperiodic microscopic description \cite{bourne2017non}.

\section{Modeling of Dirac operators}
\label{sec:model}

We briefly introduce the low-energy systems of Dirac equations and refer, e.g., to \cite{bernevig2013topological,chiu2016classification,Fruchart2013779} for motivations and details. We consider two-dimensional materials. Materials with discrete translational invariance are represented by a continuous family of Hermitian Hamiltonians $H(k)$ with $k$ belonging to a compact Brillouin zone (e.g., the two-dimensional torus). For each $k$, such Hamiltonians are self-adjoint operators with compact resolvent and sheets of eigenvalues $k\mapsto\lambda_j(k)$ for $j\geq1$. At specific points $K$, two such sheets may touch, and will generically do so by touching conically. This means, assuming the energy equal to $E=0$ at such Dirac points $K$ and linearizing $\tilde k=K+k$ for $|k|\ll1$ in the vicinity if $K$, that the two sheets of the Hamiltonian are linear in $k$ and the latter, reduced to the vicinity of $K$, takes the form
\[
  \hat H = Ak \cdot \sigma,
\]
where $k=(k_1,k_2)$, $A$ is an invertible $2\times 2$ matrix describing the geometry of the cones, and $\sigma=(\sigma_1,\sigma_2)$ are the first two Pauli matrices. The Hamiltonian acts on $2-$spinors corresponding to the modes `above' and `below' the energy of intersection $E=0$.

The above operator does not describe an insulator since all energies in the vicinity of $0$ are allowed to propagate. The insulation comes from the variation of a mass term $m$ in a Hamiltonian, which in the linearization described above, takes the form
\[
  \hat H[m] = Ak \cdot \sigma + m \sigma_3,
\]
where $\sigma_3$ is the third Pauli matrix. These matrices satisfy the commutativity relations $\sigma_i\sigma_j+\sigma_i\sigma_i=2\delta_{ij} I_2$ for $1\leq i,j\leq 3$. With this, we observe that $\hat H^2=(|Ak|^2+m^2) I_2$ so that the resulting Hamiltonian has a spectral gap $(-|m|,|m|)$.

As $m$ varies continuously from a non-vanishing value to another non-vanishing value with a different sign, the material moves from insulating to metallic when $m=0$ to insulating again. This is the basic topological phase transition in topological materials and this will be justified by assigning indices of Fredholm operators to the above Hamiltonians.

The topology of a Hamiltonian, whether in the commutative or non-commutative settings, typically involves projection onto the negative part of the spectrum of the Hamiltonian. This is how Chern numbers are calculated \cite{chiu2016classification,Fruchart2013779}. For Hamiltonians described over a non-compact domain $k\in\Rm^2$, it turns out that the resulting projection is not sufficiently regular to obtain Fredholm operators. One then has two (quite similar) choices: either regularize the operator so that the corresponding projection leads to a Fredholm operator; this is what we will do in section \ref{sec:2d}; or regularize (modify) the spectral projection so that we are also led to Fredholm operators; this is what we will do in section \ref{sec:nd} in arbitrary dimensions. The regularized operator is of the form
\[
	\hat H_\eta[m] = Ak\cdot\sigma + m_\eta \sigma_3,\qquad m_\eta = m-\eta|A k|^2,
\]
for some $\eta\not=0$. We will see that only the sign of $\eta$ matters topologically (or the sign of the determinant of the invertible matrix $\eta$ in a more general regularization of the form $-k^t\eta k$ instead of $-\eta|Ak|^2$).  Dirac operators with quadratic corrections have been used in a variety of contexts; see, e.g., \cite{PhysRevB.83.125109,RevModPhys.83.1057,shen2011topological} as well as \cite[Chapters 8\&16]{bernevig2013topological}.

The above Hamiltonians are the Fourier multipliers of Dirac operators, which in the physical variables, have the form
\[
  H_\eta[m] = \dfrac1i A\nabla\cdot\sigma + (m+\eta\Delta_A)\sigma_3,
\]
where $\nabla$ is the standard gradient operator and $\Delta_A=(A\nabla)\cdot (A\nabla)$ is the standard Laplacian when $A=I$. This is our elementary model of {\em bulk Hamiltonian}. 

The above Hamiltonian describes transport for low energies $|E|\ll1$ and (equivalently) wave-numbers $k$ in the vicinity of the Dirac point $K$. It turns out that most materials are characterized by a larger (than one, typically even) number of Dirac points $K_j$ for $1\leq j\leq J$ for a given energy ($E=0$). Each Hamiltonian $H_{j,\eta}[m]$ comes with a structure $A_j$ and a mass term $m_j$ (and the same regularization $\eta\not=0$ to simplify). The Hamiltonians may also act on larger spinors when additional internal degrees of freedom (such as spin) are taken into account. Accounting for all Dirac points (corresponding to the energy $E=0$) and internal degrees of freedom (and relabeling $J$ above accordingly), the general form of the (unperturbed) Hamiltonian we consider here can be written as a direct sum of operators of elementary Hamiltonians to yield
\[
  H_\eta[m]  = \oplus_{j=1}^J H_{j,\eta}[m_j],\qquad H_{j,\eta}[m_j] = \dfrac1i A_j\nabla\cdot\sigma + (m_j+\eta\Delta_A)\sigma_3.
\]
This is the most general model of {\em unperturbed bulk Hamiltonian} considered in this paper.

\medskip

For energies $E$ close to $0$ inside the spectral gap, the material is insulating and no transport may occur. The interesting practical features of topological insulators appear when two bulk materials with two different topologies are in contact. Then, modes may develop along the common interface and display unusual asymmetric (chiral) transport properties. In this paper, as in, e.g., \cite{bernevig2013topological}, we model such transitions continuously assuming that $m=m(x)$ continuously changes from $m_+$ as one component of $x$ goes to $+\infty$ to $m_-$ as that same component goes to $-\infty$. In the two-dimensional setting parametrized by $(x,y)$, that component is $x$ itself and interface modes then appear in the vicinity of $x=0$, where $m$ vanishes, and propagate along the interface parametrized by $y$. 
The elementary {\em interface Hamiltonian} then takes the form
\[
	H_\eta[m(x)] =  \dfrac1i A\nabla\cdot\sigma + (m(x)+\eta\Delta_A)\sigma_3.
\]

In section \ref{sec:2d}, we assign topological invariants to the Hamiltonians $H_\eta[m]$ and $H_\eta[m(x)]$ following the spectral calculus presented in \cite{avron1994,bellissard1986k,prodan2016bulk}. The explicit calculation of the bulk indices closely follows that in \cite{avron1994}, which we generalize to the vectorial setting after appropriate regularization $(\eta\not=0)$. The interface Hamiltonian is analyzed by means of the winding number of an appropriate unitary operator modeling transport along the edge. Once these indices are defined, they are shown to be stable with respect to a large class of heterogeneous perturbations in section \ref{sec:stab}. The bulk and interface topological indices are then related to a (quantized) physical quantity showing the transport asymmetry even in the presence of spatial heterogeneities. 

Extensions to higher dimensions are then given in section \ref{sec:nd}. We first concentrate on the two-dimensional setting.


\section{Two-dimensional bulk and interface Hamiltonians}
\label{sec:2d}

We consider the elementary $2\times2$ Dirac Hamiltonians described in the preceding section with the simplifying assumption that $A=I_2$. We will remove such an assumption at the end of the section. The unperturbed two dimensional Dirac Hamiltonian is thus given by
\begin{equation}
 \label{eq:Hh} 
   H_\eta[m] =\dfrac 1i \nabla\cdot\sigma + (m+\eta\Delta) \sigma_3
\end{equation}
where $\sigma=(\sigma_1,\sigma_2)$, $m\not=0$ and $\eta\in\Rm$ is a sufficiently small regularizing parameter. We represent position on $\Rm^2$ as $x=(x_1,x_2)$.

We use the notation $H_\eta[m(x)]$ for the Hamiltonian with $m$ above replaced by $m(x_1)$, a function that  converges to $m_\pm$ sufficiently rapidly as $x_1\to\pm\infty$ for some $|m_\pm|\geq m_0>0$.  The important topological property of such wall domains (transition from a bulk state characterized by $m_-$ to one characterized by $m_+$) is captured by the quantity
\[
   \eps = \dfrac12 (\sgn{m_+}-\sgn{m_-}),
\]
which is non-trivial only when the asymptotes have different signs.

The bulk Hamiltonian in the Fourier domain is given by its symbol
\begin{equation}
\label{eq:HF}
 \hat H_\eta[m] = k\cdot \sigma + (m-\eta|k|^2)\sigma_3
\end{equation}
for $k=(k_1,k_2)$. We also denote the partial Fourier transform by
\begin{equation}
\label{eq:HFm}
\hat H_\eta[m(x)] = \dfrac1i \partial_x \sigma_1 + k_2\sigma_2 + (m(x)+\eta\partial^2_x- \eta k_2^2)\sigma_3.
\end{equation}

For the bulk index, we define $\Un(x)=\dfrac{x_1+ix_2}{|x_1+ix_2|}$. We verify that
\begin{equation}\label{eq:U2d}
  \sgn{x\cdot\sigma} = \dfrac{x\cdot\sigma}{|x\cdot\sigma|} = \left(\begin{matrix} 0&\Un^*(x)\\ \Un(x)&0\end{matrix}\right).
\end{equation}
The role of $\Un(x)$ is to be used in the construction of a Fredholm operator whose index reflects the topological properties of the Hamiltonian $H_\eta[m]$; see Appendix \ref{sec:FH} for the role of Fredholm modules in the theories presented in this paper.

\subsection{Bulk Index}
\label{sec:bulkindex2d}
For $|\eta|$ small enough\footnote{More precisely, for $\eta>0$ and $2\eta m<1$, we observe that the minimum of $k^2\mapsto E^2(k^2):=k^2+(m-\eta k^2)^2$ is attained at $k^2=0$ and equals $m^2$. For $2\eta m>1$, the minimum is given by $(4\eta m-1)/(4\eta^2)>m/(2\eta)$. For $\eta<0$, we have $E^2\geq m^2$.}, we verify that $(-m,m)$ remains a spectral gap for $H_\eta[m]$, while the gap shrinks for large values of $|\eta|$. As is then standard in the analysis of topological insulators \cite{avron1994,bellissard1994noncommutative,Fruchart2013779,prodan2016bulk}, we define
\begin{equation}
\label{eq:P}
P := P_-(H_\eta[m]) = \chi(H_\eta[m]<0)
\end{equation}
the projection onto the negative part of the spectrum of $H_\eta[m]$, with $\chi(E<0)$ the indicatrix function equal to $1$ for $E<0$ and equal to $0$ otherwise. The operator $P$ is a bounded operator (of norm $1$) on $L^2(\Rm^2)\otimes\Cm^2$ defined by spectral calculus \cite{davies_1995}. The main result of the section is the:
\begin{theorem}
\label{thm:bulk2d}
  Let $P$ and $\Un$ be defined as above and $0<|\eta|<\frac{1}{2m}$.  Then $P\Un P_{|{\rm Ran }P}$ is a Fredholm operator and we have
  \begin{equation}
\label{eq:indexbulk2d}
  -\Ind{P\Un P_{|{\rm Ran} P} } = - \Ind { P\Un P + I-P } = \dfrac12\big( \sgn{m} + \sgn{\eta}\big).
\end{equation}
\end{theorem}
The rest of the section is devoted to the proof of the above theorem. 

When $m$ and $\eta$ are positive, say, then the above topological invariant, given by the index of a Fredholm operator, takes a non-vanishing value here equal to $1$. It describes the topology of the material associated to the Hamiltonian \eqref{eq:Hh}.
Note that the above right-hand-side is integer-valued only for $\eta\not=0$. Only after proper regularization $\eta\not=0$ (which essentially allows one to obtain a meaningful one-point compactification of the open domain $\Rm^2$) is the operator $P\Un P_{|{\rm Ran}P}$ Fredholm. After regularization, the latter operator has a sufficiently rapidly decaying kernel that the theory for trace-class operators described in \cite{avron1994} can be extended to the present vectorial setting.

Let us define the projector
\begin{equation}
\label{eq:Q}
  Q := \Un P \Un^* \qquad\mbox{ so that } \qquad P-Q= \Un [\Un^*,P] = [P,\Un]\Un^*,
\end{equation}
with $[A,B]=AB-BA$ the usual commutator. Then we have the:
\begin{proposition}\label{prop:1index2d}
  The operator $(P-Q)^3=\Un[P,\Un^*][P,\Un][P,\Un^*]$ is trace-class so that $P\Un P_{|{\rm Ran }P}$ is Fredholm and
  \begin{equation}
\label{eq:indexQ}
   -\Ind{P\Un P_{|{\rm Ran }P} } = \Tr (P-Q)^3 = \Tr \ \Un[P,\Un^*][P,\Un][P,\Un^*].
\end{equation}
\end{proposition}
\begin{proof}
  We verify that $(P-Q)^3=\Un[P,\Un^*][P,\Un][P,\Un^*]$ using \eqref{eq:Q}. The result then follows from Proposition \ref{prop:trace} in the Appendix provided that we can show that $(P-Q)^3$ is trace-class.
  
 Let $K=P-Q=({K_{ij}})_{1\leq i,j\leq 2}$, where $K_{ij}$ are bounded operators on $L^2(\Rm^2)$ with integral kernels $k_{ij}(x,y)$. It is sufficient to prove that each $k_{ij}\in {\cal I}_3$. Here, ${\cal I}_n$ is the Schatten class of compact operators whose definition is recalled in Appendix \ref{sec:FH}. For $k(x,y)$ the kernel of any operator $K_{ij}$ or their adjoints, we wish to show that $\|k\|_{q,p}<\infty$ with $q=\frac32$ and $p=3$ by applying Lemma \ref{lem:russo}.
 Let $p(x,y)\equiv p(x-y)$ be the kernel of $P$. In the Fourier domain, $P$ acts as a multiplier given by:
\[
 \hat P(k) = \frac I2 -\frac12 \frac{k_1\sigma_1+k_2\sigma_2 + (m-\eta|k|^2)\sigma_3}{(|k|^2+(m-\eta|k|^2)^2)^{\frac12}} .
\]
This is the symbol of a multiplier 
and we verify, using the notation $\aver{x}=(1+|x|^2)^{\frac12}$, that 
\[
  |\partial^\alpha \hat P(k)| \leq C_\alpha \aver{k}^{-|\alpha|},\qquad \partial^\alpha=\partial_1^{\alpha_1}\partial_2^{\alpha_2},\quad |\alpha|=\alpha_1+\alpha_2.
\]
This decay is a direct consequence of the existence of the spectral gap for $H$ in the vicinity of $k=0$ so that $\hat P(k)$ is also a smooth function. As a consequence, 
\[
  x^\alpha p(x) = {\cal F}^{-1} ((i\partial)^\alpha \hat P )(x),
\]
(with ${\cal F}^{-1}$ the inverse Fourier transform)
involves an absolutely convergent integral when $|\alpha|>d=2$. The same reasoning shows that $p(x)$ is smooth away from $x=0$ and decays faster than $|x|^{-\beta}$ for any $\beta>0$. It remains to address the behavior of $p$ at $x=0$. For this, we decompose
\[
  \hat P(k) = \hat p_0(k) + \hat p_{-1}(k) + \hat p_{-2}(k) + \hat r(k),
\]
where $\hat p_j(\lambda k)=\lambda^{-j}\hat p_j(k)$ for $\lambda>0$ and where $\hat r(k)$ generates an absolutely convergent integral with continuous inverse Fourier transform $r(x)$. The contribution $\hat p_0(k)$ is constant and generates a term in $p(x)$ proportional to $\delta(x)$. This term disappears in the kernel of $(P-Q)^j$. The contribution $\hat p_{-1}(k)$ provides a term
\[
p_{-1}(x) =  {\cal F}^{-1} [\frac{k_1\sigma_1+k_2\sigma_2}{\eta|k|^2}](x) = C\frac{1}{\eta} \dfrac{x\cdot\sigma}{|x|^2},
\]
for some constant $C$. Therefore, $|p(x)|\lesssim |x|^{-1}$ in the vicinity of $0$.  The contribution from $p_{-2}$ is smoother and involves at most a logarithmic singularity in $(x-y)$ which is handled similarly to the one above. All other contributions are jointly continuous in $(x,y)$.

Looking at the kernel $k(x,y)$ of any of the components of $(P-Q)$, we thus obtain the estimate:
\begin{equation}\label{eq:bdk}
  |k(x,y)| \lesssim \Big|1-\frac{\Un(x)}{\Un(y)} \Big| \dfrac{1}{|x-y|^\alpha \aver{x-y}^{\beta-\alpha}} \lesssim \min (1,\frac{|x-y|}{|y|}) \dfrac{1}{|x-y|^\alpha \aver{x-y}^{\beta-\alpha}} ,
\end{equation}
with $\alpha=1$ and $\beta$ as large as we want and certainly larger than $d=2$. Here, we use the same bound coming from the unitary $\Un$ as in \cite{avron1994}. We can now use Lemma \ref{lem:estip} below, which generalizes results in \cite{avron1994}, with $n=d=2$, $p=3$ and $q=\frac32$ to obtain that all components of $(P-Q)$ belong to ${\cal I}_3$. This shows that $(P-Q)^3$ is trace-class and an application of Proposition \ref{prop:trace} concludes the proof of the result.
\end{proof}
\begin{lemma}
\label{lem:estip}
Let $p>n$ and $\frac 1p+\frac 1q=1$.
Let $k(x,y)$ be the kernel of an operator $K$ defined on functions in $\Rm^n$ with bound
\begin{equation}\label{eq:bdk2}
|k(x+y,y)|  \leq C \min (1,\frac{|x|}{|y|}) \dfrac{1}{|x|^\alpha \aver{x}^{\beta-\alpha}} \qquad \mbox{ for } \quad \alpha<n<\beta.
\end{equation}
Then,
\[
   \Big( \dint_{\Rm^n} \big( \dint_{\Rm^n} |k(x,y)|^q dx\big)^{\frac pq} dy\Big)^{\frac 1p} \leq C,
\]
so that $K$ is in the Schatten class ${\cal I}_p(\Rm^n)$ and $K^p$ (for $p$ an integer chosen as $p=n+1$ here) is trace-class.
\end{lemma}

\begin{proof}
The method follows that in \cite{avron1994}. Let $F(y) = \int |k(x+y,y)|^q dx$. By translation, we wish to show that $F\in L^{p-1}(\Rm^n)$ for $p-1=\frac pq$. We separate the integrals $|x|<|y|$ and $|x|>|y|$. 
The first one is
\[
  \dint_{|x|<|y|} \dfrac{1}{|x|^{\alpha q}\aver{x}^{\gamma q}} \dfrac{|x|^q}{|y|^q} dx = c_n \dint_0^{|y|} \dfrac{1}{|x|^{\alpha q}\aver{|x|}^{\gamma q}} \dfrac{|x|^{q+n-1}}{|y|^q} d|x|,
\] 
using $\gamma=\beta-\alpha$. The second integral is 
\[
  \dint_{|x|>|y|} \dfrac{1}{|x|^{\alpha q}\aver{x}^{\gamma q}} dx = c_n \dint_{|y|}^\infty  \dfrac{1}{|x|^{\alpha q}\aver{|x|}^{\gamma q}} |x|^{n-1}d|x|.
\]
We separate the cases $|y|<1$ and $|y|>1$. For $|y|<1$, the first integral gives a contribution, using $\aver{|x|}\geq1$, of the form $|y|^{-q+(1-\alpha)q+n}=|y|^{n-\alpha q}$ provided that $|x|^{(1-\alpha)q+n-1}$ is integrable at $0$, i.e., provided that $(1-\alpha)q+n-1>-1$, or $\alpha<1+\frac nq=n+\frac{p-n}p$, which holds. The second integral is split into $|y|<|x|<1$ and $1<|x|$, with a second contribution that is integrable when $n-1-\beta q<-1$, i.e., $n<\beta q$, which holds. The first contribution is the integral between $|y|$ and $1$ of $|x|^{n-1-\alpha q}$, which provides a contribution $|y|^{n-\alpha q}$ again.

Now consider $|y|>1$. The first integral provides a contribution $|y|^{-q}$ times the integral on $(0,1)$ of $|x|^{(1-\alpha)q+n-1}$, which is finite again when $\alpha<1+\frac nq$, which holds. So, we obtain a contribution of order $|y|^{-q}$. The second contribution  involves the integral $\int_{|y|}^\infty |x|^{n-1-\beta q} d|x|$, which is of order $|y|^{n-\beta q}$ with no difficulty of integration at $+\infty$ since $n<\beta q$.
 
 Summarizing the above bounds, we find that
 \[
   F(y) \lesssim (|y|^{-q}\wedge 1) + |y|^{n-\alpha q}\aver{|y|}^{-\gamma q}.
 \]
 These two contributions are in $L^{p-1}(\Rm^n)$ provided that $q(p-1)=p>n$ and provided that $(n-\beta q)(p-1)+n-1<-1$ and $(n-\alpha q)(p-1)+n-1>-1$ or equivalently $\beta>n$ and $\alpha<n$, which hold.
\end{proof} 

The next step is to show that such a trace can be computed as an integral. 
\begin{proposition}\label{prop:2index2d}
  Let $\kappa(x,z)$ be the kernel of the trace-class operator $(P-Q)^3$. Then we have
  \begin{equation}
\label{eq:indexTr1}
   -\Ind{P\Un P_{|{\rm Ran P}} } = \Tr (P-Q)^3 =\dint_{\Rm^2} \tr \ \kappa(x,x)dx.
\end{equation}
\end{proposition}
This result is a direct application of Lemma \ref{lem:estip} and  Lemma \ref{lem:russo} to obtain that $(P-Q)^3$ is trace-class and then of corollary \ref{cor:trace} or Lemma \ref{lem:trace} (since $\kappa(x,y)$ is continuous except where $x=0$ and $y=0$) to obtain the trace as an integral.

Our next step involves a (known and non-trivial) geometric identity to simplify the expression of the integral.
\begin{proposition}\label{prop:3index2d}
  Let $p(x,z) \equiv p(x-z)$ be the kernel of the spatially homogeneous operator $P$. Under the hypotheses of the above proposition, we have
  \begin{equation}
\label{eq:indexTr2}
   \Tr (P-Q)^3 =  -2\pi i \dint_{\Rm^4}  \tr\ p(-x)p(x-z)p(z) x\wedge (x-z) dx dz.
\end{equation}
\end{proposition}
\begin{proof}
The trace of $(P-Q)^3$ is given explicity by
\[
	{\cal T} := \dint_{\Rm^{2\times3}} p(x_1,x_2)(1-\frac{\Un(x_1)}{\Un(x_2)}) p(x_2,x_3)(1-\frac{\Un(x_2)}{\Un(x_3)}) p(x_3,x_1)(1-\frac{\Un(x_3)}{\Un(x_1)}) dx_1 dx_2 dx_3.
\]
Changing variables $x=x_1$ and $y_j=x_{j+1}-x$ for $j=1,2$, and using the translational invariance of $P$, we obtain
\begin{equation}\label{eq:traceconvol}
 {\cal T} = \dint_{\Rm^{2\times2}} p(-y_1)p(y_1-y_2) p(y_2) (2\pi i) y_1\wedge y_2 dy_1 dy_2
\end{equation}
using the geometric identity in \cite[Lemma 4.4]{avron1994}; see also \cite{prodan2016bulk} and Appendix \ref{sec:connes}:
\[
 \dint_{\Rm^2} (1-\frac{\Un(x)}{\Un(y_1+x)})(1-\frac{\Un(y_1+x)}{\Un(y_2+x)}) (1-\frac{\Un(y_2+x)}{\Un(x)}) dx = 2\pi i (-y_1)\wedge (-y_2).
\]
Here, $x\wedge z=x_1z_2-x_2z_1=(x-z)\wedge z$ is the (two-dimensional) volume of the parallelogram with vertices $0$, $x$, and $z$ (and $x+z$). This gives the result.
\end{proof}

The preceding result involves convolutions in the spatial domain that may be estimated in the Fourier domain.
Let $\hat P(k)$ be the Fourier transform of the kernel of the translation-invariant operator $P$. We use the following convention for Fourier transforms
\begin{equation}\label{eq:FT}
  \hat Q(k) = \dint_{\Rm^d} e^{-ik\cdot x} Q(x)dx,\qquad Q(x) = \dfrac{1}{(2\pi)^d}\dint_{\Rm^d} e^{ik\cdot x} \hat Q(k) dk. 
\end{equation}
We define the Chern number of the projector as 
\begin{equation}\label{eq:chern}
 \Ch{2}[\hat P] := \dfrac i{2\pi} \dint_{\Rm^2} \tr\ \hat P [\partial_1 \hat P,\partial_2 \hat P] dk = \dfrac{i}{2\pi} \dint_{\Rm^2} \tr\ \hat P d\hat P \wedge d\hat P.	
\end{equation}
We verify that the last two terms are equal. We follow the sign convention in \cite{avron1989chern} and \cite[(2.3)]{prodan2016bulk} although the opposite sign convention often appears in the physical literature. Chern numbers are defined as integrals of even dimensional forms (Chern characters) and are therefore (possibly) non-trivial only in even dimensions. The (even) Chern number is defined as ${\rm Ch}_d$ in (real) space dimension $d$ in \cite{prodan2016bulk}, while it is typically denoted by ${\rm Ch}_{\frac d2}$ in the literature, where $\frac d2$ is the complex dimension. We use the above somewhat hybrid notation to represent the Chern number in even spatial dimension $d=2$.
\begin{proposition}\label{prop:4index2d}
  Let $\hat P(k)$ be the Fourier transform of the kernel of the translation-invariant operator $P$. Under the hypotheses of the above proposition, we have
  \begin{equation}
\label{eq:indexTr3}
   \Tr (P-Q)^3 =  \Ch{2}[\hat P] .
\end{equation}
\end{proposition}
\begin{proof}
	We verify that $(2\pi)^d\int_{\Rm^d} P(-x)Q(x) = \int_{\Rm^d} \hat P(k)\hat Q(k)dk$ and that the Fourier transform of  the convolution $Q_1*Q_2$ is given by $\hat Q_1 \hat Q_2$. Applying this to $P(y)=p(y)$ and $Q=Q_1Q_2$ with $Q_1(y)=Q_2(y)=yp(y)$ in \eqref{eq:traceconvol} using $y_1\wedge y_2=(y_1-y_2)\wedge y_2$, we obtain that ${\cal T}$ equals the first form in \eqref{eq:chern} and hence the result.
\end{proof}
The following proposition concludes the proof of the above theorem.
\begin{proposition}\label{prop:5index2d}
  Let $\hat P(k)$ be the Fourier transform of the kernel of the translation-invariant operator $P$. Under the hypotheses of the above proposition, we have
  \begin{equation}
\label{eq:indexTr4}
   \Ch{2}[\hat P] = \frac 12 \big( \sgn{m} + \sgn{\eta} \big).
\end{equation}
\end{proposition}
\begin{proof} This is a standard result; see, e.g., \cite{Fruchart2013779}. We propose a short derivation for completeness.
  We start from \eqref{eq:indexTr3} and recall that $\hat P=\frac12(I-\frac{\hat H}{|\hat H|})$. First,
  \[
  	-2\partial_j \hat P = \partial_j \dfrac{1}{|\hat H|} \ \hat H + \dfrac{1}{|\hat H|} \partial_j \hat H.
  \]
  Since $\tr\ A[B,C]=\tr B[C,A]$, we observe that $[\hat P,\partial_j\hat P]=\frac{1}{4|\hat H|^2}[\hat H,\partial_j \hat H]$ so that we need to compute
  \[
   \Ch{2}[\hat P] = \dfrac{i}{2\pi} \dint_{\Rm^2} \dfrac{-1}{8|\hat H|^3} \tr \ \hat H[\partial_1\hat H,\partial_2 \hat H] dk.
  \]
  With $\hat H=k\cdot\sigma+(m-\eta |k|^2)\sigma_3$, we observe that $\partial_j \hat H=\sigma_j-2\eta k_j\sigma_3$ so that
  \[
    [\partial_1 \hat H,\partial_2 \hat H] = [\sigma_1,\sigma_2]-2\eta k_1[\sigma_3,\sigma_2]-2\eta k_2[\sigma_1,\sigma_3],
  \]
  with therefore, using identities such as $[\sigma_1,\sigma_2]=2i\sigma_3$,
  \[
    \tr  \ \hat H[\partial_1\hat H,\partial_2 \hat H] = 4i\big( (m-\eta |k|^2) + 2\eta |k|^2\big) = 4i (m+\eta |k|^2).
  \]
  As a consequence,
  \[
  	\Ch{2}[\hat P] = \dfrac{1}{4\pi} \dint_{\Rm^2} \dfrac{m+\eta |k|^2}{(|k|^2+(m-\eta|k|^2)^2)^{\frac32}} dk =\dfrac{1}{2}\dint_0^\infty \dfrac{ (m+\eta r^2)r}{(r^2+(m-\eta r^2)^2)^{\frac32}}dr.
  \]
  We verify that the above integral gives the right-hand side in \eqref{eq:indexTr4}; see also \cite[Chapter 8]{bernevig2013topological} or \cite[(27)-(31)]{lu2010massive} and the proof of Theorem \ref{thm:bulkeven} for a more `topological' calculation of the Chern number.
\end{proof}

\begin{remark}\label{rem:half}
 The calculations of the preceding proposition also apply to the case $\eta=0$, and the result is a number $\Ch{2}[\hat P_{|\eta=0}]=\frac 12 \sgn{m}$, which is not an integer and therefore cannot be given any interpretation as a Fredholm index. This shows the necessity to regularize the projection $P$ in order to be able to construct a Fredholm operator and assign a topological invariant to the bulk Hamiltonian. Note that the kernel of $P$ (when $\eta=0$) is too singular for $(P-Q)^3$, or $(P-Q)^{2n+1}$ for any $n$, to be trace-class. 
 
 Geometrically, $\hat P$ is represented by the map from $\Rm^2$ to $\Sm^2$ given by 
 \[
 k\mapsto h(k) = \dfrac{(k_1,k_2,m-\eta|k|^2)}{|(k_1,k_2,m-\eta|k|^2)|} = 
    \dfrac{(k_1,k_2,m-\eta|k|^2)}{(|k|^2+(m-\eta|k|^2)^2)^{\frac12}}.
 \]
 Assume $m>0$. When $\eta=0$, such a map covers only the upper half sphere. When $\eta>0$, the whole sphere is covered by the map. When $\eta<0$, at most half of the sphere is covered as well. This standard topological argument \cite{Fruchart2013779} will be used more directly in the higher dimensional case, where we will show in more generality that $\deg{h}=\frac12(\sgn{m}+\sgn{\eta})$.

See also \cite{bal2019topological} for a notion of bulk-difference invariant that bypasses the regularization.
 \end{remark}

\subsection{Interface index}
\label{sec:interface2d}
In this section, we replace the coordinates $(x_1,x_2)$ by the more convenient $(x,y)$. The dual variable to the coordinate $y$ will be denoted by $\zeta$.

Let us now consider the interface Hamiltonian $H[m(x)]$ with $m(x)$ converging to $m_\pm$ as $x\to\pm\infty$ with $|m_{\pm}|\geq m_0$ for some $m_0>0$. Unlike its bulk counterparts (corresponding to $H[m_\pm]$), the interface Hamiltonian does not have a spectral gap and this is what generates its unusual transport properties along the interface $x=0$ parametrized by $y\in\Rm$. The topology of the interface Hamiltonian is captured by the integer
\begin{equation}\label{eq:eps}
 \eps := \dfrac12 \Big( \sgn{m_+}  - \sgn{m_-} \Big).
\end{equation}

We define $\Pr=\Pr(y):=\chi(y>0)$ the projection onto the positive part of the $y$ axis. The projection $\Pr$ plays a similar role in odd dimension $d-1=1$ to that of the unitary $\Un(x)$ in even dimension $d=2$; see Appendix \ref{sec:FH}.

We now define the following functional of $H[m(x)]$. Let $\chi_\delta(E)$ be a smooth function from $\Rm$ to $\Rm$ such that $\chi_\delta'(E)$ is supported in $(-\delta,\delta)$ while $\chi_\delta(-\delta)=0$ and $\chi_\delta(\delta)=1$. The parameters $\eta$ and  $\delta>0$ are chosen such that
\begin{equation}\label{eq:delta}
0<\delta<(1+2|\eta|m_0)^{-\frac12}m_0\leq m_0,\qquad  0\leq |\eta|<\frac1{2m_0}.
\end{equation}
We define the unitary $\Uin(E)=e^{i 2\pi \chi_\delta(E)}$ and by spectral calculus 
\begin{equation}\label{eq:UH}
\Uin:=\Uin(H[m(x)]) = e^{i2\pi \chi_\delta(H[m(x)])}.
\end{equation}
This a unitary operator on $L^2(\Rm^2)\otimes\Cm^2$. Then, we have:
\begin{theorem}
\label{thm:interface2d}
  Let $\Pr$ and $\Uin$ be defined as above and $\eps$ in \eqref{eq:eps}. Then $\Pr \Uin\Pr_{|{\rm Ran}\Pr}$ is a Fredholm operator and 
  \begin{equation}
  \label{eq:interface2d}
    -\Ind{\Pr \Uin \Pr_{|{\rm Ran \Pr}} } = -\Ind { \Pr \Uin \Pr + I-\Pr } = \eps.
\end{equation}
\end{theorem}

 The above index can be related to a physical quantity similar to a current as will be shown in section \ref{sec:current}.
To prove the above theorem, we first show that $[\Pr,\Uin]$ is Hilbert-Schmidt. This implies that $\Pr \Uin \Pr_{|{\rm Ran}\Pr}$ is a Fredholm operator; see Proposition \ref{prop:fedosov} in Appendix \ref{sec:FH}. 

Let ${\cal F}={\cal F}_{y\to\zeta}$ be the partial Fourier transform with respect to the second variable $y$. Upon permuting the basis $(\sigma_{1,2,3})\to(\sigma_{2,3,1})$, which does not change the commutation relations of the Pauli matrices, we represent $H[m(x)]$ as the direct sum
\begin{equation}
	\label{eq:directsum2d}
	H[m(x)] = {\cal F}^{-1} \dint_{\Rm}^{\oplus} \hat H(\zeta) d\zeta \ {\cal F},
\end{equation}
where
\[
  \hat H(\zeta) = \zeta\sigma_3 + \fa \sigma_-+\fa^*\sigma_+,\quad \sigma_\pm = \frac12(\sigma_1\pm i\sigma_2),\quad \fa = \partial_x + m(x) + \eta(\partial^2_x-\zeta^2).
\]
This is a convenient basis to represent $H$ since  $\partial_x$ and $m(x)$ appear in the same entries of the matrix $H$; see also \cite{B-EdgeStates-2018,Fruchart2013779}.
Note that $\fa^*$ and $\fa$ resemble the standard creation and annihilation operators of the quantum oscillator.
Define
\begin{equation}
	\label{eq:WH}
	\Win = \Win(H) = \Win(H[m(x)]) := \Uin-I,	
\end{equation}
such that $\Win(E)=\Uin(E)-1$ is compactly supported in $(-m_0,m_0)$. By functional calculus, we have
\[
	 \Win = {\cal F}^{-1} \dint_{\Rm}^{\oplus} \Win(\hat H(\zeta)) d\zeta \ {\cal F}.
\] 
To further analyze $\Uin$, we need to obtain a more explicit spectral decomposition of the operator $\Win$.
We write the above sum as an integral of operators concentrating in the vicinity of the interface $x=0$ as follows.  
\begin{lemma}
\label{lem:astara}
	Let $\eta$ and $\delta$ as in \eqref{eq:delta}. Then the part of the spectrum of $\fa^*\fa$ and $\fa \fa^*$ inside $[0,\delta^2)$ is discrete and composed of a finite number of eigenvalues $0<\lambda_j(\zeta)<\delta^2$ plus a one dimensional contribution at $\lambda=0$ when $|\epsilon|=1$. The corresponding eigenvalues of $\hat H(\zeta)$ are given by 
	\begin{equation}\label{eq:Ej}
		E_{j,\pm}(\zeta) = \pm (\lambda_j(\zeta) + \zeta^2)^{\frac12}.
	\end{equation}
	Associated to the eigenvalues are finite-rank projectors $\Pi_{j,\pm}(\zeta)$, which are sums of rank-one projectors of the form $\psi_{j,\pm}(\zeta,x)\psi_{j,\pm}(\zeta,x)^*$ that are rapidly decaying as $|x|\to\infty$. Moreover, the above eigenvalues and projectors are real-analytic functions of $\zeta\in (-\delta,\delta)$.
	
	For $m(x)$ smooth on $\Rm$ and equal to $m_\pm$ for $|x|$ large enough, there is a branch of simple eigenvalues $E(\zeta)=\epsilon\zeta$ associated to a non-trivial kernel of $\fa$ (when $\epsilon=1$) or $\fa^*$ (when $\eps=-1$) for all $|\zeta|<\delta$. The corresponding eigenvectors $\psi_0(x,\zeta)$ are rapidly (exponentially) decaying in $x$ uniformly in $\zeta$ and real-analytic in $\zeta$.\end{lemma}
\begin{proof}
We recall that $\fa=\partial_x + m(x) + \eta(\partial^2_x-\zeta^2)$ and $\fa^*=-\partial_x+m(x)+\eta(\partial^2_x-\zeta^2)$. The cases $\eta=0$ and $\eta\not=0$ are treated differently. Note also that 
\[
  \hat H^2(\zeta) = {\rm Diag} (\zeta^2 + \fa^*\fa,\zeta^2 + \fa\fa^*).
\]
The spectrum of $\hat H^2$ is therefore given by $\zeta^2$ plus the shared positive spectrum of $\fa^*\fa$ and $\fa\fa^*$ as well as the null spectrum of $\fa$ or $\fa^*$. Our objective is to show that while $\fa\fa^*$ may possess essential spectrum, the latter does not intersect $[0,\delta^2]$. We write
\[\begin{array}{rcl}
 \fa\fa^* &=&  -\partial^2_x + (m(x)+\eta(\partial^2_x-\zeta^2))^2 + m'(x) \\
   &=& -\partial^2_x + m^2(x)+ \eta^2(\partial^2_x-\zeta^2)^2 + \eta m\partial^2_x + \eta \partial^2_x m -2m(x)\eta\zeta^2  + m'(x)\\
   &=& -\partial^2_x + m^2(x) + \eta^2(\partial^2_x-\zeta^2)^2 +2\eta\partial_x m \partial_x +\eta m''(x) -2m(x)\eta\zeta^2  + m'(x)\\
   &=& m_0^2 -2\eta\zeta^2 m_0\sgn{x}  -\partial_x (1-2\eta m) \partial_x + \eta^2(\partial^2_x-\zeta^2)^2 + v(x),
 \end{array}
\]
 for some smooth, compactly supported function $v(x)$. This is a sum of positive operators provided $\eta$ is sufficiently small for $m$ bounded and $|\zeta|<\delta$. Since $v$ is a relatively compact perturbation (for $\eta^2\partial^4_x$ when $\eta\not=0$ and for $\partial^2_x$ when $\eta=0$), it can only generate discrete spectrum. The essential spectrum, if any, is bounded below by $m_0^2-2|\eta|\zeta^2m_0\geq m_0^2-2|\eta| \delta^2m_0\geq\delta^2$ by construction of $\delta$ in \eqref{eq:delta}. The spectrum of $\hat H$ below $\delta$ in absolute value is therefore necessarily discrete (point spectrum with finite dimensional eigenspaces) for the same reasons as above. Such a discrete spectrum is characterized by
 \[
   \zeta\psi_1 + \fa^*\psi_2 = E \psi_1,\qquad \fa\psi_1-\zeta\psi_2 = E\psi_2.
 \]
From the above expression of $\hat H^2$, $E^2\geq\zeta^2$.
 Let us assume that $E^2>\zeta^2$. Then we find that $\fa\fa^*\psi_2=(E^2-\zeta^2)\psi_2$ and $\psi_1=(\zeta-E)^{-1}\fa^*\psi_2$. For each $0<\lambda_j(\fa\fa^*)=\lambda_j(\fa^*\fa)$, we therefore have two eigenvalues of $\hat H$ given by \eqref{eq:Ej}. The eigenvectors are then given by $\psi=(\psi_1,\psi_2)^t$. Such expressions are real-analytic in $\zeta$ as $\fa$ has real-analytic coefficients in $\zeta$. Outside of a compact interval, $\fa^*\fa$ and $\fa\fa^*$ have constant coefficients and one readily verifies that any eigenvector associated with discrete spectrum has to decay exponentially as $|x|\to\infty$. We leave the details to the reader.
 
 It remains to analyze the kernels of $\fa$ and $\fa^*$. Let us now assume that $E^2=\zeta^2>0$. Then either, $E=\zeta$ and $\fa^*\psi_2=0$ while $2E\psi_2=\fa\psi_1$ or $E=-\zeta$ and $\fa\psi_1=0$ while $2E\psi_1=\fa^*\psi_2$. The latter case implies $\fa\fa^*\psi_2=0$, hence $\fa^*\psi_2=0$ and $\psi_1=0$. For the same reasons, the former case implies $a\psi_1=\psi_2=0$. We now distinguish three cases. We obtain that the kernel of $\fa^*$ is trivial while that of $\fa$ is not when $\eps=\frac12(\sgn{m_+}-\sgn{m_-})=1$, that the kernel of $\fa$ is trivial while that of $\fa^*$ is not when $\eps=-1$, and that both kernels are trivial when $\eps=0$. The above shows that $E(\zeta)=\eps\zeta$ when $|\eps|=1$.
 
We focus on the case $\eps=1$ so that $m_-<0<m_+$ and show that the kernel of $\fa^*$ is trivial while that of $\fa$ is not; the other cases are treated similarly. This holds independent of $\eta$ sufficiently small, including $\eta=0$. Note that this also solves the case $E=\zeta=0$.
 
Let $\eta\not=0$ and look for $\fa\psi=0$. For $|x|$ large, say $|x|>x_0$, we have $\fa= \partial_x + \eta(\partial^2_x-\zeta^2) + m_\pm$. For $x>x_0$, this is a second-order equation with solutions $e^{-\lambda x}$ given by $\eta\lambda^2-\lambda+m_+-\eta\zeta^2=0$. For $\eta>0$, there are two positive solutions, while for $\eta<0$, there is only one positive solution. For $x<-x_0$, we look for solutions $e^{\lambda x}$ given by $\eta\lambda^2+\lambda+m_--\eta\zeta^2=0$, with one positive solution when $\eta>0$ and two when $\eta<0$. 

On the interval $(-x_0,x_0)$, the second-order ordinary differential equation (ODE) also admits two independent solutions. Let $\eta>0$ so that $\psi=a_1e^{-\lambda_1 x}+a_2e^{-\lambda_2 x}$ for $x>x_0$. We solve the ODE on $(-x_0,x_0)$ and obtain an invertible (as one easily verifies) $2\times2$ transition matrix $B$ mapping $(a_1,a_2)$ to $(b_1,b_2)$ with $\psi$ for $x<-x_0$ given by $b_1e^{-\lambda_3 x} + b_2 e^{\lambda_4 x}$, with $\lambda_3,\lambda_4>0$. Any normalized vector requires $b_2=0$ and this imposes a linear constraint on $(a_1,a_2)$. With that linear constraint, we have constructed a unique, up to normalization, solution $\psi$ of $\fa\psi=0$. When $\eta<0$, a similar mechanism allows us to construct a unique normalized solution as well. When $\eta=0$, a simpler procedure, since $\fa$ is first-order, leads to the same conclusion.

One now repeats the procedure for $\fa^*$ and obtains that there is always a side $x>x_0$ or $x<-x_0$ such that no normalized solution exists. When $\eta=0$, we look for a solution of the form $\partial_x\psi=m_+\psi$ for $x>0$, which is not normalizable unless trivial.

The solutions constructed above converge exponentially to $0$ as $|x|\to\infty$, and do so uniformly in $\zeta$ for $|\zeta|<\delta$. 
This concludes the proof of the lemma. 
\end{proof}
Since $\Win(E)$ is supported in $(-\delta,\delta)$, we deduce from the preceding lemma that
\begin{equation}
	\label{eq:Wspectral}
	\Win = {\cal F}^{-1} \dint_{(-\delta,\delta)}^{\oplus} \dsum_{j,\pm} \Win(E_{j,\pm}(\zeta)) \Pi_{j,\pm}(\zeta) d\zeta \ {\cal F},
\end{equation}
where $\Pi_{j,\pm}$ are the above finite rank projectors for $\hat H$ associated to $E_{j,\pm}$.
\begin{proposition}
  Let $\Uin$ and $\Win$ be defined as above. Then $\Pr \Uin \Pr$ is a Fredholm operator on ${\rm Ran}\Pr$. More precisely, $[\Pr,\Win]=[\Pr,\Uin-I]=[\Pr,\Uin]$ is a Hilbert-Schmidt operator.
\end{proposition}
\begin{proof}
	That $[\Pr,\Win]$ is a Hilbert-Schmidt operator is proved for all of the terms in the above finite sum. Then $[\Pr,\Uin]$ is compact and $\Pr \Uin \Pr$ is Fredholm; see Appendix \ref{sec:FH}.
	
	Consider the above sum of terms, each of which may be written as
	\[
	  \Win_0={\cal F}^{-1} \dint_{(-\delta,\delta)}^{\oplus} \Win(E(\zeta)) \psi(x,\zeta)\psi^*(x,\zeta) d\zeta\ {\cal F},
	\]
	which is an operator on $L^2(\Rm^2)$ with Schwartz kernel of the form:
	\[
		w(x,y;x',y') = \dint_{\Rm} \Win(E(\zeta)) \psi(x,\zeta) \psi^*(x,\zeta) \dfrac{1}{2\pi} e^{i(y-y')\zeta} d\zeta.
	\]
	The kernel of the bounded operator $[\Pr,\Win_0]$ is thus given by
	\[
	  k(x,y;x',y') = (\chi(y)-\chi(y')) \dint_{\Rm} \Win(E(\zeta)) \psi(x,\zeta) \psi^*(x',\zeta) \dfrac{1}{2\pi} e^{i(y-y')\zeta} d\zeta.
	\]
	It remains to show that $\Tr k^*k$ is integrable in all four (one-dimensional) variables. This is expressed as an integral in $(\zeta,\zeta')$. The trace and contributions in the $(x,x')$ variables yield
	\[
		\psi_2(\zeta,\zeta') = \dint_{\Rm^2} \Tr\ \psi(x',\zeta')\psi^*(x,\zeta') \psi(x,\zeta)\psi^*(x',\zeta) dx dx' = \Big|\dint_{\Rm} \psi(x,\zeta)\cdot \psi^*(x,\zeta')dx\Big|^2.
	\]
	This is bounded by $1$ by the Cauchy-Schwarz inequality  and is real-analytic in $\zeta$ for $|\zeta|<\delta$. Integrating in $(\zeta,\zeta')$ therefore yields the term
	\[
	  0\leq w_2(y-y') = \dint_{\Rm^2} \Win(E(\zeta))\Win^*(E(\zeta')) \psi_2(\zeta,\zeta') \dfrac{1}{(2\pi)^2} e^{i(y-y')(\zeta-\zeta')} d\zeta d\zeta'. 
	\]
	The integral runs over $|\zeta|,|\zeta'|<\delta$ and $y^{2n}w_2(y)$ is bounded for any $n\geq0$ so that in particular $\int_{\Rm} |y|w_2(y)dy<\infty$. But this means that
	\[
	  \int_{\Rm^4} \Tr k^*k dx dx' dy dy' = \int_{\Rm^2} (\chi(y')-\chi(y))^2 w_2(y-y') dy dy' \leq \dint_{\Rm} |y| w_2(y) dy <\infty,
	\]
	so that $[\Pr,\Win_0]$ is Hilbert-Schmidt and by finite summation $[\Pr,\Win]=[\Pr,\Uin]$ as well.
\end{proof}

We next show that the index is not perturbed by local perturbations of $m(x)$. Let $m_t(x)$ be a family of functions equal to $m_\pm$ outside of a fixed interval $I$ and continuous in $t\in[0,1]$ in the 
square integrable sense on $I$ (this means that $\int_I(m_{t+h}(x)-m_t(x))^2 dx$ goes to $0$ with $h$). 
Let $\Uin_t=\Uin[H[m_t(x)]]$ the corresponding unitary operators. Then we have the following result.
\begin{proposition}
\label{prop:contm}
  Let $\Uin_t$ be defined as above. Then $\Pr \Uin_t\Pr$ is Fredholm for all $t\in[0,1]$ and the index of these operators on ${\rm Ran}\Pr$ is independent of $t$.
\end{proposition}
\begin{proof}
  The proof is based on the stability of $\Uin[H[m_t(x)]]$ with respect to perturbations. We write the spectral calculus in terms of resolvents and obtain from the above constraint on $m$ that for $h$ small,
  \[
  \| R_{t+h}(z) - R_t(z) \| = \| (z-H_t)^{-1} - (z-H_{t+h})^{-1} \| \leq C |h| |{\Im z}|^{-1},
  \]
for a constant $C$ independent of $t$.  We then apply Proposition \ref{prop:HSF2} to the compactly supported function $\Win=\Uin-I$ showing that $\Uin_{t+h}-\Uin_t$ is bounded in uniform norm by a constant times $|h|$. The family $\Pr \Uin_t\Pr$ is therefore a continuous family of Fredholm operators (they are indeed Fredholm since $[\Pr,\Uin_t]$ is Hilbert-Schmidt as obtained above). The proposition results from the continuity of the index for such families as recalled in Appendix \ref{sec:FH}.
\end{proof}
This result shows that the index is independent of $m(x)$ so long as its asymptotic limits at $\pm\infty$ are preserved (as a combination of the above argument and the stability of the index with respect to small perturbations in the uniform norm easily shows). It also allows us to restrict the calculation to the following simpler case to conclude the proof of Theorem \ref{thm:interface2d}:
\begin{proposition}\label{prop:interfacetrace}
      Let $m(x)$ be a continuous monotone function from $m_-$ to $m_+$, with $m(x)$ not equal to $m_\pm$ on a compact interval. Then $[\Pr,\Uin]=[\Pr,\Win]$ is trace-class and so is $\Pr-\Qr=[\Pr,\Uin]\Uin^*$ with $\Qr=\Uin\Pr\Uin^*$. Moreover, 
    \begin{equation}\label{eq:indextrace}
   -\Ind{\Pr \Uin \Pr_{|{\rm Ran \Pr}} } = \Tr ([\Pr,\Uin]\Uin^*) = \eps=\frac12(\sgn{m_+}-\sgn{m_-}).
   \end{equation}
\end{proposition}
\begin{proof}
We find that
\[
 \fa^*\fa = -\partial_x^2 + (m(x)+\eta(\partial_x^2-\zeta^2))^2 + \partial_x m(x).
\]
 We know that the first two terms are jointly bounded from below by a positive constant. For $m(x)$ monotonically increasing (consider $\fa\fa^*$ otherwise), $\partial_x m\geq0$ so that $\fa^*\fa$ has no spectrum inside $(-\delta,\delta)$. Therefore, there is only one branch in the spectral representation corresponding to $E(\zeta)=\eps\zeta$ so that
 \[
   \Win={\cal F}^{-1} \dint^{\oplus}_{|\zeta|<\delta} \Win(\eps\zeta) \Pi_{0}(\zeta) d\zeta\ {\cal F}.
 \]
 The operator $K=[\Pr,\Uin]$ therefore has kernel
 \[\begin{array}{rcl}
   k(x,y;x',y') &=&  (\chi(y)-\chi(y')) u(x,x',y-y'), \\ u(x,x',y-y')&=& \dint_{\Rm} \Win(\eps\zeta) \psi_0(x,\zeta)\psi_0^*(x',\zeta) \dfrac{1}{2\pi}e^{i(y-y')\zeta} d\zeta. \end{array}
 \]
 We want to show that the above operator $K$ is trace-class. The continuity of the kernel $k$ is not a sufficient condition to obtain such a result; see Appendix \ref{sec:FH}. We need to estimate the singular values of $K$ more directly. The trick \cite{GK-AMS-69,LAX-JW-02} is to use the smoothness and decaying properties of $u$. We need to replace $u(x,x',z)$ by compactly supported functions in $z$. 
 
 Let $0\leq\varphi(z)\leq1$ be a smooth function compactly supported in $(-\frac\pi2-\eta,\frac\pi2+\eta)$ for $0<\eta<1$ such that $\varphi(z)=1$ on $(-\frac\pi2+\eta,\frac\pi2-\eta)$ and such that $\varphi(z)+\varphi(z-\pi)=1$. Then $\sum_{k\in\Zm} \varphi(z-k\pi)=1$ and this generates a partition of unity. Let now $u_k(x,x',z)=u(x,x',z)\varphi(z-k\pi)$. By construction of $u$ and integration over a compact support in $\zeta$, as well as the smoothness properties of $\psi_0(x,\zeta)$, we have the following approximation property. Define
 \[
   u_k(x,x',z)=\dsum_n e^{-inz} \hat u_{k,n}(x,x')
 \]
 the Fourier coefficients of the function $u_k$ on its support $k\pi\pm(\frac\pi2+\eta)\subset ((k-1)\pi,(k+1)\pi)$. On that domain, after periodization, let
 \[
 	u_{kN}(z,x,x')=\dsum_{|n|\leq N} e^{-inz} \hat u_{k,n}(x,x').
 \]
 Then by approximation theory, 
 \[
   \|u_k-u_{kN}\|_{L^2((k-1)\pi,(k+1)\pi)} \leq C |u_k|_\alpha N^{-\alpha}
 \]
 where $|u_k|(x,x')$  is the $C^\alpha$ norm in the $z$ variable of $u_k$ on its support. Again by smoothness, we obtain that 
 \[
   |u_k|_\alpha \leq C\aver{k}^{-\beta} \aver{x}^{-\gamma}\aver{x'}^{-\gamma}
 \]
 for $\beta$ and $\gamma$ as large as necessary. 
 
 Let now $K_k$ be the operator with kernel $(\chi(y)-\chi(y'))u_k(x,x',y-y')$ and let $K_{kN}$ be the corresponding operator with $u_k$ replaced by $u_{kN}$. Then the above calculations show that 
 \[
   \|K_k^* K_k - K_k^* K_{kN}\| \leq \|K_k^*\| \|K_k-K_{kN}\| \leq C \aver{k}^{-\beta} N^{-\alpha}.
 \]
 The reason is that 
 \[
   \|K_k-K_{kN}\|^2 \leq \dint_{|y-y'-k\pi|<\pi} \hspace{-1cm} (\chi(y)-\chi(y'))^2  |u_k(y-y',x,x')-u_{kN}(y-y',x,x')|^2 dydy'dxdx',
 \]
which is bounded by $C \aver{k}^{-\beta} N^{-\alpha}$ for $u_k$ and $u_{kN}$ compactly supported.
 Note that $K_{kN}$ has rank at most $2N+1$. Since $K_k^*K_k$ is self-adjoint, its best finite rank approximation provides an estimate of $\lambda^2_{k,2N+1}$, where $\lambda_{k,N}$ are the singular values of $K_k$. Therefore, we have shown that $\lambda_{k,N}\leq C\aver{k}^{-\beta/2} N^{-\alpha/2}$.  For $\alpha>2$, this shows that $K_k$ is trace-class and that its ${\cal I}_1$ norm is bounded by $\aver{k}^{-\beta/2}$. Since $K=\sum K_k$, and the latter sum converges absolutely in ${\cal I}_1$, which is complete, we deduce that $K$ is trace-class as well.
 
Since ${\cal I}_1$ is an ideal, $\Tc=[\Pr,\Uin]\Uin^*=[\Pr,\Win](I+\Win^*)$ is also trace-class. The kernel of $K=[\Pr,\Win]$ is continuous except at $y=0$ and $y'=0$. Replacing the kernel $k$ by $\tilde k$ as in the work \cite{brislawn1988kernels} recalled in Lemma \ref{lem:trace}, we obtain that $\tilde k=k$ almost everywhere so that the trace of $K$ is given by the integral of its kernel along the diagonal.  Such an integral vanishes by explicit computation so that $\Tr \Tc = \Tr [\Pr,\Win]\Win^*$.

The Schwartz kernel of $\Win$ is: 
\[
  \hat w(y-y',x,x') = \dint_{\Rm} \Win(\eps\zeta)\psi_0(x,\zeta) \psi_0^*(x',\zeta) \dfrac{1}{2\pi} e^{i(y-y')\zeta} d\zeta.
\]
It is a smooth and rapidly decaying function in all variables. The kernel of $\Win^*$ is given by $\hat w^*(y'-y,x',x)$. The kernel of $\Tc$ is thus of the form
\[
  t(y,y',x,x') =   \dint_{\Rm^2} (\chi(y)-\chi(y'')) \hat w(y-y'',x,x'') \hat w^*(y'-y'',x',x'') dy'' dx''.
\]
From the smooth decay of $\hat w$, we deduce that the above kernel is jointly continuous in all variables except at $y=0$. Again, we observe that $t=\tilde t$ (as described in Appendix \ref{sec:FH}) at all points except $y=0$ so that the trace of $\Tc$ is given by the diagonal integral of $\tilde t$, which is the same as that of $t$. Therefore, $\Tr \Tc$ is given by 
\[
    \dint\tr \ t(y,y,x,x) dy dx =\dint (\chi(y)-\chi(y''))  \tr  \hat w(y-y'',x,x'') \hat w^*(y-y'',x,x'') dx''dy''dx dy.
\]
 With the change of variables $z=y-y''$ and integration in the remaining variable, we get
 \[
 	\Tr \Tc = \dint z \tr \hat w(z,x,x'')  \hat w^*(z,x,x'') dx dx'' dz.
 \]
 Moving back to the Fourier domain yields 
 \[
  \Tr \Tc = \dfrac{1}{2\pi i}\dint \tr \partial_\zeta [\Win(\eps\zeta) \psi_0(x,\zeta)\psi_0^*(x'',\zeta)] \Win^*(\eps\zeta) \psi_0(x'',\zeta)\psi_0^*(x,\zeta) dx dx'' d\zeta.
 \]
 One term is 
 \[
 	\dfrac{1}{2\pi i}  \dint_{\Rm} \Win'(\eps\zeta) \Win^*(\eps\zeta) d\zeta = \dfrac{\eps}{2\pi i}  \dint_{\Rm} \Win'(\zeta) \Win^*(\zeta) d\zeta,
 \]
 after taking traces and integrating in $x$ and $x''$. The other term involves 
 \[
 	\dint_{\Rm^2} \Tr [\partial_\zeta \psi_0(x,\zeta) \psi_0^*(x,\zeta) \psi_0(x',\zeta) \psi_0^*(x',\eta) + \psi_0(x,\zeta) \partial_\zeta\psi_0^*(x',\zeta) \psi_0(x',\zeta) \psi_0^*(x,\zeta)] dx dx'.
 \]
 Exchanging $x$ and $x'$ integrations in the latter term and taking traces and integration in $x'$, we get
 \[
   \dint_{\Rm} (\partial_\zeta \psi_0(x,\zeta)\psi_0^*(x,\zeta) + \psi_0(x,\zeta)\partial_\zeta  \psi_0^*(x,\zeta) ) dx = \partial_\zeta \dint |\psi_0(x,\zeta)|^2 dx = 0.
 \] 
 This shows that $\Tr \Tc$ is $\eps$ times the winding number of $\Win$. We verify that $\Win'\Win^*=2\pi i(\chi_\delta'-\Win')$. Since $\int \Win'=0$ and  $\int \chi_\delta'=1$, this concludes the proof.
\end{proof}

\subsection{Direct sum of elementary Hamiltonians}
\label{sec:directsum}

The above results obtained for elementary bulk and interface Hamiltonians easily generalize to the direct sums of such terms as they appear in the (low-energy) description of many topological insulators. We now summarize such extensions.

We first consider the elementary bulk Hamiltonian $H_\eta[m]=\frac1i A\nabla + \eta \Delta_A$ with $A$ an invertible matrix. All propositions in the preceding section hold until the final calculation in Proposition \ref{prop:5index2d}. With $l=Ak$, we observe that $\hat H(k)=l\cdot\sigma+(m-\eta|l|^2)\sigma_3$. We then verify that $\partial_j \hat H(k)=[A^t (\sigma-2\eta l\sigma_3)]_j$, from which we deduce that $[\partial_1 \hat H,\partial_2 H]$ equals the term we would obtain with $A=I$ multiplied by the determinant ${\rm det}(A^t)$. Now, changing variables $dl=|{\rm det}(A)| dk$, we deduce that the Chern number for $\hat P$ is simply that for $A=I$ multiplied by $|{\rm det}(A)|^{-1}{\rm det}(A^t)=\sgn{{\rm det}(A)}$.

Thus, for $H_\eta[m]=\frac1i A\nabla + \eta \Delta_A$ and $P$ the projection as defined in \eqref{eq:P}, we obtain as in Theorem \ref{thm:bulk2d} that $P\Un P_{|{\rm Ran }P}$ is Fredholm and that its index is given by 
\[
  -\Ind{P\Un P_{|{\rm Ran }P} } = I[A,m,\eta] :=  \dfrac12 \big( \sgn{m}+\sgn{\eta})\sgn{{\rm det}(A)}.
\] 
As a consequence, for a general unperturbed Hamiltonian given by direct sum
\[
  H_\eta[m] = \oplus_{j=1}^J H_{\eta,j}[m_j] ,\qquad H_{\eta,j}[m]=\frac1i A_j\nabla + (m_j+\eta \Delta_{A_j})\sigma_3,
\]
we obtain that for  $P$ the projection as defined in \eqref{eq:P}, then $P\Un P_{|{\rm Ran }P}$ is Fredholm and its index is given by 
\[
  -\Ind{P\Un P_{|{\rm Ran }P} } = \dsum_{j=1}^J I[A_j,m_j,\eta] = \dsum_{j=1}^J  \dfrac12 \big( \sgn{m_j}+\sgn{\eta})\sgn{{\rm det}(A_j)}.
\] 

\medskip

Let us now move to the edge Hamiltonians. To preserve the commutative structure of the $\sigma$ matrices and the orthogonality between the interface variable $y$ and the distance to the interface variable $x$, we assume that $A={\rm Diag}(a_x,a_y)$ with $a_xa_y\not=0$. Retracing the derivation in section \ref{sec:interface2d}, we observe that for $\eta$ sufficiently small, the results obtained with $a_x=a_y=1$ are such that $\eps$ should be multiplied by $\sgn{a_x}$ in the derivation of the decay away from the interface $x=0$ and $\zeta$ should then be multiplied by $a_y$ so that eventually, the invariant is given by the winding number of $\Win(\sgn{a_x}\eps \sgn{a_y}\zeta)$, and hence by $\sgn{a_xa_y}=\sgn{{\rm Det}(A)}$ times the winding number when $a_x=a_y=1$.

Therefore, for $H[m(x)]=\dfrac1i A\nabla\cdot\sigma + (m(x)-\eta \Delta_A)\sigma_3$, and $\Uin=e^{i2\pi\chi_\delta(H_\eta[m(x)])}$ the unitary given in \eqref{eq:UH}, we find that $\mP \Uin\mP_{|{\rm Ran}\mP}$ is Fredholm and its index is given by $\frac12(\sgn{m_+}-\sgn{m_-})\sgn{{\rm Det}(A)}$. More generally, define
\[
  H_\eta[m(x)] = \oplus_{j=1}^J H_{\eta,j}[m_j(x)],\qquad H_{\eta,j}[m_j] = \frac1i A_j \nabla + (m_j(x)+\eta\Delta_A)\sigma_3
\]
with $m_j(x)$ equal to $m_{j\pm}$ as $x\to\pm\infty$ outside of a compact interval and $A_j$ invertible diagonal matrices. Let $\Uin$ be defined as in \eqref{eq:UH}. Then $\mP \Uin\mP_{|{\rm Ran}\mP}$ is Fredholm and its index is given by 
\[
-\Ind{\mP \Uin\mP_{|{\rm Ran}\mP}} = N_+-N_-,\qquad N_\pm = \dsum_{j=1}^J I[A_j,m_{j\pm},\eta] .
\]
This clearly displays the (bulk-boundary or more appropriately bulk-interface) correspondence between the invariant of the (unperturbed) interface Hamiltonian and the difference of invariants of the (unperturbed) bulk Hamiltonians, which is independent of the regularization parameter $\eta$. The objective of the next section is to generalize this correspondence to the setting where the Hamiltonians are perturbed by spatial fluctuations.

\section{Stability under perturbation, Bulk-Interface correspondence and interface current}
\label{sec:stab}

The indices calculated in the preceding section hold for unperturbed Hamiltonians, although the proof of invariance of the index under changes in the profile $m(x)$ did involve a perturbation of the Hamiltonian $H[m(x)]$. Here, we wish to consider more general perturbations modeled by a bounded operator $V$ (to simplify) so that 
\[
   H_V = H_\eta[m] + V,
\]
for the bulk Hamiltonian and 
\[
	H_V = H_\eta[m(x)] + V,
\]
for the interface Hamiltonian.  Both $H_\eta[m]$ and $H_\eta[m(x)]$ are given as the direct sum of $J$ elementary bulk and interface Hamiltonians, respectively. The objective of this section is to show the invariance of the bulk and interface indices, once properly defined, under a large class of perturbations $V$. 

Perturbations of operators defined by spectral calculus have been analyzed in many settings. A versatile methodology to propagate perturbations though the spectral calculus is offered by the Helffer-Sj\"ostrand formula, which we now recall following \cite[Chapter 2]{davies_1995}, to which we refer for details and background.  

Let $f(x)$ be a smooth, bounded function on $\Rm$ converging to $0$ at infinity.  Then the functional calculus states that
\begin{equation}
\label{eq:HSformula}
	f(H) = -\frac1\pi \dint_{\Cm} \pdr{\tilde f}{\bar z}(z) (z-H)^{-1} dx dy,
\end{equation}
where $z=x+iy$ and where $\tilde f(z)$ is an extension of $f(x)$ such that 
\begin{equation}
\label{eq:HSftilde}
		 \pdr{\tilde f}{\bar z}(z) = \dfrac12 \dsum_{r=0}^n f^{(r)}(x) \dfrac{(iy)^r}{r!}(\sigma_x+i\sigma_y) + \frac12 f^{(n+1)}(x) (iy)^n \dfrac{\sigma(z)}{n!}.
\end{equation}
Here, $\sigma_{t}=\partial_{t}\sigma$ for $t=x,y$ and using the notation $\aver{x}=(1+x^2)^\frac12$, $\sigma(z)$ is defined as
\[
  \sigma(z) = \tau\big(\dfrac{y}{\aver{x}}\big),\ \ \tau\in C^\infty_0(\Rm),\ \tau(0)=1,\ {\rm supp}(\tau) \subset \{|u|\leq 2\},\ {\rm supp}(\tau') \subset \{1\leq |u|\leq 2\}.
\]
We assume in this paper that $f^{(j)}$ is compactly supported for $j\geq1$. When $f$ is  used in the construction of a unitary operator, then $f$ is compactly supported as well. In the construction of idempotent operators (projectors), $f$ is typically bounded but not compactly supported.

Let $H_1$ and $H_2$ be two operators. Then the above formula states that
\begin{equation}\label{eq:deltaf}
  f(H_1)-f(H_2) = -\frac 1\pi\dint_{\Cm}  \pdr{\tilde f}{\bar z}(z) \Big( (z-H_1)^{-1} - (z-H_2)^{-1} \Big) dx dy,
\end{equation}
which provides a convenient formula to propagate fluctuations in the resolvent to fluctuations in spectrally defined operators. We use this formula to obtain specific results in spectral theory below; see \cite{cycon2009schrodinger} for a more comprehensive presentation.  Many proofs below are based on the following classical relations
\begin{equation}\label{eq:AB}
	(z-(A+B))^{-1}-(z-A)^{-1}=(z-(A+B))^{-1}B(z-A)^{-1},
\end{equation}
where $(A,B)$ is of the form $(H,V)$ or $(H_V,-V)$ or $(H,z- \tilde z)$ with $H_V=H+V$. For $B$ bounded (to simplify), we therefore obtain the existence of bounded operators $C_{1,2}$ such that 
\begin{equation}\label{eq:AB2}
  (z-(A+B))^{-1}-(z-A)^{-1}= C_1 (z-A)^{-1} = (z-A)^{-1} C_2.
\end{equation}
\begin{proposition}
\label{prop:HSF}
 Let $f$ be as above with $n=1$. Let $H$ and $H_V=H+V$ be unbounded operators on $L^2(\Rm^2;\Cm^m)$ with $V$ bounded. 
 
 Assume that 
 \begin{equation}\label{eq:compactres}
 	(z-H_V)^{-1} - (z-H)^{-1} \quad \mbox{ is compact for some $z\in\Cm$},
 \end{equation}
 so that it is compact for all $z\in\Cm$ such that $y=\Im z\not=0$, and assume that $f$ is compactly supported and $n=1$.  Then $f(H_V)-f(H)$ is a compact operator.
 
 Assume \eqref{eq:compactres} and 
 \begin{equation}\label{eq:boundres}
 \|(z-H_V)^{-1} V (z-H)^{-1}\| = \|(z-H_V)^{-1} - (z-H)^{-1} \| \lesssim |\Im z|^{-1-\mu},
 \end{equation}
 for some $\mu>0$ and assume that $f$ is bounded and its derivatives are compactly supported.  Then $f(H_V)-f(H)$ is compact. 
\end{proposition}
\begin{proof}
Using \eqref{eq:AB}, we have $(z-H_V)^{-1}-(z-H)^{-1}=(z-H_V)^{-1}V(z-H)^{-1}$. Using \eqref{eq:AB2} with $A=H$ or $H_V$ and $B=z_2-z_1$, we find the existence of bounded operators $C_{1,2}$ such that
\[
	(z_2-H_V)^{-1}-(z_2-H)^{-1} = C_1 \big( (z_1-H_V)^{-1} - (z_1-H)^{-1}\big) C_2,
\]
so that if the above difference is compact for $z$, it is so for all $z$ in the resolvent sets of $H$ and $H_V$, which includes $\Im z\not=0$.

We now need to consider the integrations of terms involving $\sigma_x+i\sigma_y$, which is supported on $\aver{x}\leq |y|\leq 2\aver{x}$ so that $|y|\geq1$, and a term involving $f"(x)y\sigma(y)$, which is of order $O(y)$ for $|y|\ll1$. All above terms are continuous in the variable $z$ (with values in the Hilbert space) since $(z_1-H)^{-1}-(z_2-H)^{-1}=(z_1-H)^{-1}(z_2-z_1)(z_2-H)^{-1}$. From the definition of $\tau$, we observe that the integrand is supported in the domain $|y|\leq 2\aver{x}$, which is bounded when $f$ is compactly supported. By the above smoothness, the integral can be approximated by collocation, i.e., as a sum over a finite number of points $z_j$ with appropriate weights $w_j$, at least after removing the sliver $|y|\leq\delta$, which provides a uniformly vanishing contribution as $\delta\to0$. This shows that $f(H+V)-f(H)$ is the uniform limit of compact operators and is therefore compact. This proves the first statement.
 
 For the second statement, the integration involving $0\leq\aver{x}\leq L$ provides a compact contribution for the same reason as above. We thus need to show that the integration over $L<\aver{x}$
 provides a negligible contribution in operator norm. The only non-vanishing term in \eqref{eq:HSftilde} is that corresponding to $r=0$. We observe that $|\sigma_x+i\sigma_y|\lesssim \aver{x}^{-1}$ and is supported on $L<\aver{x}\leq|y|\leq2\aver{x}$. Using the second hypothesis, we observe that the integration over $L<\aver{x}\leq|y|\leq2\aver{x}$ is bounded in uniform norm by
 \[
 	\dint_{L<\aver{x}\leq|y|\leq2\aver{x}} \aver{x}^{-1}|y|^{-1-\mu} dxdy \lesssim \dint_{L<\aver{x}} \aver{x}^{-1-\mu} dx \lesssim L^{-\mu}.
 \]
 As a uniform limit of compact operators, we obtain that  $f(H+V)-f(H)$ is compact. This proves the result.
\end{proof}
Sufficient conditions on $V$ so the above result applies are given in the:
\begin{lemma}\label{lem:Vcompact}
 Let $H=H_0+H_1$ with $H_0$ a self-adjoint (matrix-valued) differential operator of order at least equal to $1$ with constant coefficients and $H_1$ a bounded Hermitian operator. Let $V$ be a bounded and converging to $0$ at infinity \footnote{In the sense that $|V(x)|$ for $|x|>R$ goes to $0$ uniformly as $R\to\infty$.}. Then \eqref{eq:compactres} and \eqref{eq:boundres} hold.
 
 For the Dirac applications considered in the paper, $H_1$ is the contribution in (direct sums of operators of the form) $H[m(x)]$ linear in $m(x)$ and $H_0$ is the remaining differential component of the operator. 
\end{lemma}
\begin{proof}
 Since $V$ is bounded, we can choose $\mu=1$ in \eqref{eq:boundres}. 
 
From \eqref{eq:AB}, we have $(z-H_V)^{-1}-(z-H)^{-1}=(z-H_V)^{-1} V (z-H)^{-1}$ and so the result follows (up to change of $z$) from showing that $(z-H)^{-1}V$ is compact. Using the same formula once more (or \eqref{eq:AB2}), we deduce that 
 \[
 	(z-H)^{-1}V = (z-H)^{-1}H_1(z-H_0)^{-1} V + (z-H_0)^{-1}V= \Big((z-H)^{-1}H_1+I\Big) (z-H_0)^{-1}V.
 \]
 So, again, we are left with showing that $(z-H_0)^{-1}V$ is compact. This is of the form $p(i\nabla)q(x)$ with both functions $p$ and $q$ bounded and converging to $0$ at infinity. The product is therefore compact; see, e.g., \cite[Lemma 4.6]{thaller2013dirac} and \cite[Theorem 3.8.8]{simon2015operator}. \footnote{The proof of \cite[Theorem 3.8.8]{simon2015operator} showing that $(z-H_0)^{-1}V$ is uniformly approximated by Hilbert-Schmidt operators easily applies to (non-local) operators $V$ of the form $v(x')\phi_1\otimes\phi_2$ with $v(x')$ bounded and converging to $0$ at infinity and $\phi_{1,2}=\phi_{1,2}(x_d)\in L^2(\Rm^d)$. Here, $x=(x',x_d)$. Such operators were used in \cite{B-EdgeStates-2018} for a quantitative analysis of the effects of random fluctuations on edge modes.}
\end{proof}

Formula \eqref{eq:deltaf} is also useful to obtain the following stability result used in section \ref{sec:2d}:
\begin{proposition}
\label{prop:HSF2}
 Let $f$ and $\tilde f$ be as indicated above with $f$ compactly supported. Let $H_1$ and $H_2$ be unbounded operators on $L^2(\Rm^2)$ such that $f(H_1)$ and $f(H_2)$ are bounded operators. 
 Assume that 
 \begin{equation}\label{eq:stabH}
 	 \| (z-H_1)^{-1} - (z-H_2)^{-1} \| \leq  C_1 \dfrac{1}{|\Im z|},
 \end{equation}
 for a constant $C_2$ independent of $z$. Then there is a constant $C_2$ independent of $H_1$ and $H_2$ such that
 \begin{equation}\label{eq:stabf}
 	 \| f(H_1)- f(H_2) \| \leq  C_2 C_1 .
 \end{equation}
 \end{proposition}
\begin{proof}
 We deduce from \eqref{eq:deltaf} that
 \[
   \| f(H_1)- f(H_2) \| \leq C_1\dint_{\Cm}  \Big|\pdr{\tilde f}{\bar z}(z)  \Big|  \dfrac{1}{|y|} dx dy.
 \]
 The integration involves compactly supported functions in $x$ and integration over $|y|<2\aver{x}$. From \eqref{eq:HSftilde}, all terms but the one corresponding to $r=0$ are bounded by a constant times $|y|$. The corresponding integrals of bounded functions over a bounded domain are therefore bounded. For the case $r=0$, we observe that the domain of integration is included in the support of $\sigma_x+i\sigma_y$, which is $\aver{x}\leq|y|\leq2\aver{x}$, where $|y|^{-1}\leq 1$. The integral is then also clearly bounded,  which yields $C_2$.
 \end{proof}
Finally, we will use the propagation formula to obtain that the fluctuations are trace-class for appropriate perturbations $V$. This will be useful in the expression of the topological index as a trace representing a current along the edge.
\begin{proposition}\label{prop:HSF3}
Let $f$ and $\tilde f$ be constructed as above with $f$ compactly supported and $n=3$. Let $V=V_1V_2$ with $V_j$ bounded operators and assume that 
\begin{equation}\label{eq:HSbound}
  \|(z-H)^{-1}V_j \|_{HS}\leq \dfrac{C}{|\Im z|},\quad j=1,2,
\end{equation}
where $HS$ denotes the Hilbert Schmidt norm. Then for $H_V=H+V$, we have that $f(H_V)-f(H)$ is trace-class.
\end{proposition}
\begin{proof}
 We first observe that
 \[
 	(z-H_V)^{-1}-(z-H)^{-1} = [(z-H_V)^{-1}V-I](z-H)^{-1}V_1V_2(z-H)^{-1}.
 \]
 From the hypothesis, we observe that the right-hand side is trace-class with a norm in ${\cal I}_1$ (see Appendix) bounded by $C|y|^{-3}$ (since $(z-H_V)^{-1}$ is bounded by $|y|^{-1}$). Now, in the definition of $\partial_{\bar z}\tilde f$, we have terms involving $\sigma_x+i\sigma_y$, which are supported on $\aver{x}\leq |y|$ so that $|y|^{-3}\leq1$ and the integral involves a bounded term integrated over a compact domain and is hence bounded in the ${\cal I}_1$ sense. The last term is $f^{(4)}(iy)^3\sigma(y)$, which compensates for the $|y|^{-3}$ term and also provides the integral of a bounded term over a bounded domain. This concludes the proof of the proposition.
\end{proof}
\subsection{Bulk index}

We now apply the above propagation of perturbation to the bulk Hamiltonian $H$ given as the direct sum of $J$ operators of the form $H_\eta[m]$ with $0<|\eta|<\frac{1}{2|m|}$ and define the perturbed Hamiltonian $H_V:= H+V$.

In the absence of perturbation, $H$ admits a spectral gap since $H^2\geq m^2$. However, after perturbation, the gap may be filled by discrete spectrum.  When $0$ is an eigenvalue, then $\sgn{H+V}$ is somewhat ambiguous and certainly affected by changes in the location of that eigenvalue. Moreover, the $\sgn{}$ function is not sufficiently smooth to allow us to apply the above perturbation result. For these reasons, we introduce the smooth function $0\leq P_\delta\leq 1$ such that $P_\delta(E)=1$ for $E<-\delta$ and $P_\delta(E)=0$ for $u\geq\delta$; for instance $P_\delta=1-\chi_\delta$ for $\chi_\delta$ as introduced in the preceding section.

For the unperturbed bulk operator,  $P_\delta(H)=P(H)$ is a projector thanks to the spectral gap of the bulk Hamiltonian $H$. However, this equality no longer holds for $H_V$ and $P_\delta(H_V)$ may not be a projector. We thus need to modify the definition of the invariant accordingly. We have
\begin{theorem}
\label{thm:perturbulk2d}
	Let $H$, $H_V=H+V$ and $P_\delta$ be defined as above. Let $V$ be a perturbation such that \eqref{eq:compactres} and \eqref{eq:boundres} hold, for instance for $V$ bounded and compactly supported. Then $P_\delta(H_V) \Un P_\delta(H_V) + I-P_\delta(H_V)$ is a Fredholm operator in $L^2(\Rm^2)$ and 
	\begin{equation}
\label{eq:indexbulk}
   -\Ind{P_\delta(H_V) \Un P_\delta(H_V) + I-P_\delta(H_V)} = - \Ind{P(H) \Un P(H)+I-P(H)}.
\end{equation}
\end{theorem}
\begin{proof}
  The results in Prop. \ref{prop:HSF} show that $P_\delta(H_V)-P_\delta(H)$ is compact. As a consequence, $P_\delta(H_V) \Un P_\delta(H_V) + I-P_\delta(H_V)$ is a compact perturbation of $P_\delta(H)\Un P_\delta(H) + I-P_\delta(H)=P\Un P+I-P$, which is Fredholm. The indices of both operators are therefore the same, and can be estimated using Theorem \ref{thm:bulk2d}.
\end{proof}
See Lemma \ref{lem:Vcompact} for admissible perturbations $V$. 

\subsection{Interface index}
The perturbations of the interface Hamiltonian are analyzed in a similar manner. It is even somewhat simpler as the function $f$ we use in the Helffer-Sj\"ostrand formula is compactly supported. Moreover, since the interface invariant is constructed to measure continuous spectrum that participates to transport, no change in the definition of the invariant is necessary to handle the perturbation.
We have:
\begin{theorem}
\label{thm:perturbint2d}
 Let $H=H_\eta[m(x)]$ and $H_V=H+V$. Assume that $H$ satisfies the hypotheses of the preceding section and that $V$ is such that \eqref{eq:compactres} holds; for instance $V$ is bounded and compactly supported. Then $\Pr \Uin(H_V)\Pr$ is a Fredholm operator on ${\rm Ran}\Pr$ and 
 \begin{equation}
\label{eq:perturbind2d}
  -\Ind{\Pr \Uin(H_V)\Pr} = -\Ind{\Pr \Uin(H)\Pr}.
\end{equation}
\end{theorem}
\begin{proof}
  The proof is similar to that of the bulk result. Since $\Uin(H)=I+\Win(H)$ and $\Win$ is smooth and compactly supported, we deduce from Proposition \ref{prop:HSF} that $\Win(H_V)-\Win(H)$ is compact, and as a consequence that $\Pr \Uin(H_V)\Pr-\Pr \Uin(H)\Pr$ is compact. The result then follows and an explicit expression for the index follows from Theorem \ref{thm:interface2d} and the generalization in section \ref{sec:directsum}.
\end{proof}
Here again, see Lemma \ref{lem:Vcompact} for admissible perturbations $V$. 
\subsection{Bulk-Interface correspondence}

Consider the most general perturbed systems of Dirac equations
\begin{equation}\label{Hgal}
  H_V = \oplus_{j=1}^J H_{\eta,j}[m_j(x)] + V,\qquad H_{V\pm} = \oplus_{j=1}^J H_{\eta,j}[m_{j\pm}] + V,
\end{equation}
for $V$ a compactly supported, bounded perturbation (a Hermitian matrix-valued operator on $L^2(\Rm^2;\Cm^{2J})$). We then define the approximate projection and unitary operators
\[
  P_{\delta\pm} = P_\delta (H_{V\pm}) = I-\chi_\delta(H_{V\pm}),\qquad \Uin=\Uin(H_V) = e^{i 2\pi \chi_\delta(H_V)}.
\]
Then $P_{\delta\pm}\Un P_{\delta\pm}+I-P_{\delta\pm}$ as well as $\Pr \Uin \Pr$ are Fredholm operators and we have the following {\em bulk-interface} correspondence
\begin{equation}\label{eq:bicorr}
  -\Ind{\Pr \Uin \Pr} = N_+-N_-,\quad N_\pm = -\Ind{P_{\delta\pm}\Un P_{\delta\pm}+I-P_{\delta\pm}}.
\end{equation}
Recall from section \ref{sec:directsum} that $N_\pm$ may be calculated as
\[
  N_\pm = \dsum_{j=1}^J \dfrac12 \Big( \sgn{m_j} + \sgn{\eta}\Big) \sgn{{\rm det(A_j)}}.
\]
Here, every matrix $A_j$ is assumed to be diagonal for the interface index to be defined.

In other words, the number of protected modes contributing to the unusual current described by $-\Ind{\Pr \Uin\Pr}$ is given by the difference of the bulk indices on either part of the interface. 

Note that such correspondences hold in much more general settings, and in particular in cases where the explicit bulk and interface indices cannot be computed explicitly but are still shown to be related algebraically; see, e.g.,
\cite{bourne2018chern,graf2007aspects,Graf2013,prodan2016bulk}.

\subsection{Interface index and physical quantities}
\label{sec:current}

The above bulk-interface correspondence links two topological quantities, bulk indices and interface indices. Such indices have physical relevance, as explained in, e.g., \cite{bernevig2013topological,Fruchart2013779,prodan2016bulk,RevModPhys.83.1057}. In these references, one observes that the bulk index can be related to some form of Hall conductance by means of a Kubo formula. The interface index is related to the number of `topologically protected' modes that propagate along the interface. 

From the physical point of view, the interface index is arguably the most relevant. Indeed, bulk insulators prevent transport in the frequency range of interest by assumption. In our setting, this is reflected by the fact that the index is not defined un-ambiguously as it depends on the sign of the regularization parameter $\eta$.  Whether one labels a given phase by $N\in\Zm$ or by a different classification $N+1$ makes no difference physically. The interface index on the other hand, reflecting a number of asymmetric modes propagating along the interface $x=0$, should not depend on any regularization and this is what we observe. It is changes in topological numbers (across interfaces), and not the absolute numbers themselves that carry the most physical relevance.

In this section, we would like to assign a physically relevant quantity to the interface index. A typical notion to model (asymmetric) flow along the interface $x=0$ is to consider current in the $y$ direction. We have that $J=\dot X=i[H,X]$ is the current operator when $X$ is the position operator.  It is shown in \cite{prodan2016bulk} for general models of randomness that have to be statistically stationary (with statistical laws that are independent of spatial translations) that the current is appropriately quantized. In our setting, this means considering the unperturbed operator $H$ and analyzing the following quantity. 

Let $\varphi(y)$ be a bounded non-negative function with compact support and such that $\int_{\Rm}\varphi^2(y)dy=1$; for instance $\varphi_L(y)$ the indicatrix function of the interval $[L,L+1]$ for $L\in\Rm$. Recall that $\chi_\delta(E)$ goes from $0$ to $1$ within the (bulk) spectral gap. We want to look at $\chi_\delta(E)$ as a smooth version of a function with a derivative $\chi_\delta'(E)$ equal to $\dfrac{1}{\Delta}\chi_{E_0<E<E_0+\Delta}$ where the whole segment $[E_0,E_0+\Delta]$ belongs to the spectral gap $(-\delta,\delta)$. Therefore, $\chi_\delta'(E)$ corresponds to the density of states inside the spectral gap.   It is therefore tempting to define
\[
  C_\varphi = \Tr (\varphi i[H,y] \chi'_\delta(H) \varphi)
\]
as the trace of the operator integrating current along the interface $i[H,y]$ over the domain normalized by $\varphi$ for the density of states $\chi'_\delta(H)$. We can then show, as in \cite{prodan2016bulk} and \cite{kotani2014quantization} in a slightly different context, that the above operator is indeed trace-class and that the trace is in fact equal to $-\Ind{\Pr U\Pr}$.

Since the result does not seem to generalize to the setting with non-stationary random perturbations, we do not write the details of the derivation and instead consider a different physical interpretation. Instead of the current associated to a position operator, we consider the variation associated to the {\em sign} of the position operator. In other words, the Schr\"odinger equation tells us that $\partial_t\frac {X}{|X|}$ projected along the interface is given by $i[H,\Pr]$, where $\Pr$ is projection onto $y>0$, or more generally $\Pr$ is $\Pr_{y_0}=\chi(y-y_0)$ projection onto $y>y_0$.

The amount of signal crossing the line $y=0$ per unit time (with $y_0=0$ to simplify) is thus given by an {\em interface conductivity} 
\[
 \sigma_I= \Tr\  i[\Pr,H] \chi'_\delta(H).
\]
This quantity, which is very similar to the edge conductivity $\sigma_E$ in \cite{elbau2002equality} is a physical observable that is quantized independently of a large class of random perturbations corresponding to trace-class perturbations.
We now show that the above quantity is quantized and given by the interface index. We need the following 
\begin{lemma}
\label{lem:derivtrace}
Let $g(H)$ be a smooth compactly supported function and $P$ be a bounded operator. We assume that for all smooth function $f(H)$, the operators $[P,f(H)]g(H)$ and $[P,f(H)g(H)]$ are trace-class.  Then $[P,H] f'(H) g(H)$ is trace-class and we have the equality
\[
 \Tr [P,f(H)]g(H) = \Tr [P,H] f'(H) g(H).
\]
\end{lemma}
\begin{proof}
 This lemma is essentially \cite[Lemma A.4]{elbau2002equality}. We propose a quick derivation.
 Let $\psi(H)$ be a smooth compactly supported function equal to $1$ on the support of $g(H)$. We then verify by cyclicity of the trace that $\Tr [P,f(H)]g(H)=\Tr [P,f(H)\psi(H)]g(H)$. In other words, we can assume that $f(H)$ is compactly supported as well. Let then $f_p(H)$ be sequence of polynomials so that $(f(H)-f_p(H))\chi(H)$  and $(f'(H)-f'_p(H))\chi(H)$ go to $0$ uniformly as $p\to\infty$. From the assumptions, $[P,f_pg]$ is trace-class and
 \[
  0=\Tr [P,f_pg] = \Tr [P,f_p]g + \Tr [P,g]f_p
 \]
 since all operators involved are trace-class. The same holds with $f_p$ replaced by $f$. Since $[P,g]$ is trace-class, $\Tr [P,g] (f-f_p)$ converges to $0$. By hypothesis, $[P,H] f'_pg$ and $[P,H]f'g$ are trace-class and $\Tr[P,H] f'_pg\to \Tr [P,H]f'g$ as $p\to\infty$ since $[P,H]g$ is trace-class. It remains to show that $\Tr [P,f_p(H)]g(H) = \Tr [P,H] f_p'(H) g(H)$ to conclude. Since $f_p$ is a polynomial, it is enough to prove the result with $f_p(H)=H^n$ for any $n\geq0$. It is clear for $n=0$ and $n=1$. 
 
 We observe that $[P,AB]=A[P,B]+[P,A]B$, i.e., $[P,\cdot]$ acts as a (non-commutative) derivation, and calculate
 \[
  \Tr [P, H^n] g= \Tr (H[P, H^{n-1}]+[P,H]H^{n-1})g= \Tr ([P, H^{n-1}]H+[P,H]H^{n-1})g
\]
by cyclicity of the trace $\Tr AB=\Tr BA$.
The result holds for $n=2$. Assuming it does for $n-1$, we get
\[
\Tr [P,H]((n-1)H^{n-2}H + H^{n-1}) g(H) = \Tr [P,H] (H^n)' g(H).
\]
This proves the result for $f_p$ a polynomial, and as described above, for any $f(H)$ by continuity. 
\end{proof}
We now apply the lemma to obtain the following result.
\begin{proposition}
\label{prop:quantsign}
	Let $H=H[m(x)]$ with $m(x)$ monotone and $H_V=H+V$. Assume that $V$ satisfies \eqref{eq:HSbound}. Then we have
	\begin{equation}\label{eq:quantsign}
	 -\Ind{\Pr \Uin(H_V)\Pr} = \Tr [\Pr,\Uin(H_V)]\Uin^*(H_V) = 2\pi i\Tr [\Pr,H_V] \chi_\delta'(H_V).
	\end{equation}
\end{proposition}
\begin{proof}
	Let us first prove the above result for $V\equiv0$. We know that $[\Pr,\Uin]\Uin^*$ is trace-class and that its trace is also that of $[\Pr,\Win]\Win^*$. The proof that $[\Pr,\Win]$ is trace-class in Proposition \ref{prop:interfacetrace} applies to any smooth function $f(H)$ that is compactly supported instead of $\Win(H)$. The assumptions of Lemma \ref{lem:derivtrace} are thus satisfied for $P=\Pr$. Choosing $f(H)=\Win(H)$ and $g(H)=\Win^*(H)$, we have $[\Pr,\Uin]\Uin^*=[\Pr,\Win](I+\Win^*)$ whose trace is $[\Pr,H]\Win'(I+\Win^*)$ and hence that of $[\Pr,H]\Uin'\Uin^*$, knowing that  $[\Pr,\Win]=[\Pr,\Win]\psi(H)=[\Pr,H]\Win'\psi=[\Pr,H]\Win'$ has vanishing trace, where $\psi(H)$ is defined in the proof of Lemma \ref{lem:derivtrace}. This yields \eqref{eq:quantsign} when $V\equiv0$ for $\Uin(H)$ defined in \eqref{eq:UH}.
	
	Now, by assumption on $V$, we obtain that $f(H)-f(H_V)$ is trace-class for any smooth compactly supported function $f$. This shows that the hypotheses of Lemma \ref{lem:derivtrace} are also satisfied when $H$ is replaced by $H_V$ and still $P=\Pr$. The above proof therefore also applies to $\Uin(H_V)=I+\Win(H_V)$ and the result follows.
\end{proof}

The above result is similar to the edge conductivity defined in \cite{graf2007aspects} in a slightly different context; see e.g., (15) there. It is known that the above trace-class assumption fails to hold for `stronger' random fluctuations; compare (15) to (12) and (20) in that reference. We thus expect \eqref{eq:quantsign} to fail for fluctuations that generate compact but not trace-class perturbations.

The formula \eqref{eq:quantsign} still provides a physically appealing and intuitive picture for the fact that the index is related to the current along the interface. In the presence of random fluctuations, which generate trace-class perturbations, the transport along the interface is quantized as indicated in the preceding formula, where $\chi_\delta'(H_V)$ represents a density of states and $i[\Pr,H_V]$ a rate of change of sign (from negative values of $y$ to positive values of $y$).


%
\section{Generalization to higher dimensions}
\label{sec:nd}
%

We now consider Dirac equations as low-energy models for higher dimensional topological insulators. As for the classical Dirac system of equations, the models involve first-order systems of equations that apply to spinors of dimension $2^\kappa$ for an even space dimension $d=2\kappa$ and of dimension $2^{\kappa+1}$ for an odd space dimension $d=2\kappa+1$. The most physically relevant cases are arguably the two-dimensional case treated in the preceding section and the three-dimensional case, where the spinors are represented as elements in $\Cm^4$. The various components of such spinors can indeed be given the interpretation of modes in the vicinity of singular (Dirac) points, as for instance recalled in \cite{hasan2011three} for Bismuth-antimony alloys and in \cite{slobozhanyuk2017three} for photonic crystals. As in the two-dimensional setting, we assume in all dimensions that the gap-less operator is given by
\[
  \hat H = Ak\cdot \gamma_d,\qquad \gamma_d=(\gamma^1_d,\ldots,\gamma^d_d)^t,
\]
where the Dirac matrices $\gamma_d^j$ are recalled in Appendix \ref{sec:CA} and where $A=I$ to simplify the presentation. The gapped Hamiltonian is then of the form
\[
  \hat H[m] = \hat H + m \gamma_d^{d+1} = h(k)\cdot \Gamma_d,\qquad h(k) = (k_1,\ldots,k_d,m)^t,\quad \Gamma_d=(\gamma_d,\gamma_d^{d+1})^t.
\]
In the physical domain, the unperturbed gapped (elementary) Hamiltonian is therefore given by
\begin{equation}\label{eq:Hmnd}
  H[m] = \dfrac{1}{i} \nabla\cdot\gamma_d + m \gamma_d^{d+1}.
\end{equation}

As in the two-dimensional setting, the projection of the above operator onto its negative spectrum generates a bounded operator whose kernel is too singular to be used in the construction of a Fredholm operator; see Remark \ref{rem:half}. Rather than regularizing the operator itself, we shall regularize the functional calculus we perform on it. As we saw in the setting of two dimensional bulk and one dimensional interface, Chern numbers associated to projectors are defined in even dimensions while winding numbers associated to unitaries are defined in odd dimensions. This structure persists in higher dimension. In even bulk dimension $d=2\kappa$, we define the projector as 
\begin{equation}\label{eq:Peven}
  P = P[H] = \frac I 2 - \frac12 \sgn{H_\eta},\qquad H_\eta = H_\eta[H] = H + \eta \Delta \gamma_d^{d+1}. 
\end{equation}
In odd bulk dimension $d=2\kappa+1$, we define the {\em chiral} matrix as $\gamma_0:=\gamma_d^{d+2}$. In dimension $d=1$, this would be the matrix $\sigma_3$. We verify that $\gamma_0 H[m]+H[m]\gamma_0=0$, in other words, the unperturbed operator $H$ satisfies a {\em chiral symmetry}. It is only for such operators that a non-trivial topology can be assigned in odd spatial dimensions. The object of interest is then
\begin{equation}\label{eq:Fodd}
  F = \sgn{H_\eta} = \dfrac{H_\eta}{|H_\eta|},\qquad H_\eta = H_\eta[H] = H + \eta \Delta \gamma_d^{d+1}. 
\end{equation}
Here, we use the same regularization $H_\eta$ as in the even-dimensional case. We can always write $\gamma_0=\sigma_3\otimes I_{2^\kappa}$ in an appropriate representation, in which case we verify that
\begin{equation}\label{eq:Uodd}
   F = \left(\begin{matrix} 0& U^* \\ U & 0 \end{matrix}\right),
\end{equation}
since $\gamma_0F+F\gamma_0=0$ as well, where $U=U[H]$ is a unitary operator since $F^2=I$.

The topology of $H[m]$ will be assigned using Fredholm operators constructed on the regularized $P[H]$ and $U[H]$ in even and odd dimensions, respectively, following the structure of Fredholm modules recalled in Appendix \ref{sec:FH}.

Let us now consider the interface
 Hamiltonians $H[m(x_d)]$. We assume that $m=m(x_d)$ consists of a smooth compactly supported transition from a value $m_-\not=0$ as $-x_d\gg1$ to a value $m_+\not=0$ as $x_d\gg1$. When $\sgn{m_+ m_-}=-1$, we expect a non-trivial topology in the vicinity of the interface (hyper-surface) $x_d=0$.

As in the two-dimensional setting, we also expect the topological invariants, which will be interpreted as indices of Fredholm operators, to be immune to a large class of random perturbations replacing $H$ by $H_V=H+V$. Showing this is the objective of the rest of the section. We first consider the commutative setting in section \ref{sec:commnd} where the bulk Hamiltonians are classified according to their (commutative as far as multiplication in the Fourier variables is concerned) symbol $\hat H(k)$. The next three sections concern the construction of the bulk indices first in even, and second in odd, dimensions, both for unperturbed and perturbed Hamiltonians. The crucial geometric identities allowing the passage from the non-commutative to commutative representations are recalled in Appendix \ref{sec:connes}.  The analysis of the interface Hamiltonians $H[m(x_d)]$, where $\eta$ plays no fundamental role and is therefore set to $0$ to simplify, is undertaken in section \ref{sec:interfacend}. All results are stated for one block Hamiltonian with $A=I$ knowing that the extension to arbitrary (finite) direct sums and more general cone structures is performed as in section \ref{sec:directsum}.

\subsection{Commutative Topological invariants}
\label{sec:commnd}

We first consider the commutative setting, where Hamiltonians are represented on the Brillouin zone by Fourier multipliers. The Hamiltonians $\hat H(k)$ of the preceding paragraphs are examples of such multipliers. We denote by $\Xm^d$ the space of parameters $k$ (the Brillouin zone), which in this paper is the open space $\Rm^d$ but could be another typical Brillouin zone such as the torus $\Tm^d$. 

We then consider a vector field $h=h(k)$ from $\Xm^d$ to $\Rm^{d+1}$ and a Hamiltonian
\[
   H(k) = h(k)\cdot\Gamma_d = \dsum_{j=1}^{d+1} h_j(k) \gamma_d^j.
\]
An example of $H(k)$ of interest is the regularized operator $H_\eta[H]$ considered above. Note that $|H(k)|=|h(k)|$ by choice of the $\gamma$ matrices. 

In even dimensions, we define the projector
\[
  P(k) =\frac I2 - \frac12 \frac{H(k)}{|H(k)|} = \frac I2 - \frac12 \frac{h(k)\cdot\Gamma_d}{|h(k)|}.
\]
In odd dimensions $d=2\kappa+1$, we recall that the chiral matrix is $\gamma_0=\gamma_d^{d+2}$ and that the operator $H(k)$ is {\em chiral} in the sense that $\gamma_0 H(k) + H(k) \gamma_0=0$. Representing $\gamma_0=\sigma_3\otimes I_{2^\kappa}$,  we define
\[
   F(k) = \dfrac{H(k)}{|H(k)|} = \left(\begin{matrix} 0&U(k)^*\\U(k)&0 \end{matrix}\right), \qquad U(k)=\frac{h_1\gamma^1+\ldots h_d\gamma^d + i h_{d+1}}{|H(k)|},
\]
where $U(k)$ is a unitary matrix, although not necessarily a Hermitian matrix unless $h_{d+1}=0$. The existence of such a unitary matrix is a direct consequence of the chiral symmetry of the Hamiltonian $H(k)$.

The topological invariants for $H(k)$ then appear as Chern numbers for $P(k)$ or winding numbers of $U(k)$, and in both cases can be evaluated as the degree of the map $h(k)$. These notions and their correspondence appear frequently in the literature in various forms. We summarize and for completeness  prove the results we need to calculate these invariants.

The main result of the section relates Chern number and winding numbers to the degree of the vector field $h$. These relations are algebraic and hold whether the latter quantities are indeed integer-valued invariants or not. 
\begin{theorem}
\label{thm:degreefield}
 Let $d$ be even and $H(k)=h(k)\cdot \Gamma_d$ with $h(k)$ a Lipschitz, non-vanishing, vector field from $\Xm^d$ to $\Rm^{d+1}$. Let $P$ be the projector defined by $P(k)=\frac I2 - \frac12\frac{H}{|H|}$, where $|H|=|h(k)|$. Then
 \begin{equation}\label{eq:PL}
 	 \tr \ P (dP)^{d} = \dfrac{-1}{2^{d+1}}  \dfrac{1}{|H|^{d+1}}\tr  \ H (dH)^{d} = \frac{-(2i)^{\frac d2} d!}{2^{d+1}}  \dfrac{1}{|h|^{d+1}} {\rm Det} (L) dk 
 \end{equation}
as equalities of volume forms, where the matrix $L$ is given by $L_{1j}=h_j$ on the first row and $L_{kj}=\partial_{k-1}h_j$ on the remaining rows $2\leq k\leq d+1$; each time $1\leq j\leq d+1$.
As a consequence, the Chern number is given by
\begin{equation}\label{eq:Chnd}
\Ch{d}[P] = \dfrac{i^\frac d2}{(2\pi)^{\frac d2}(\frac d2)!} \dint_{\Xm^d} \tr \ P (dP)^{d}  =\dfrac{\epsilon_d}{A_d} \dint_{\Xm^d} \dfrac{1}{|h|^{d+1}} {\rm Det} (L) dk 
\end{equation}
where $A_d=2\pi^{\frac{d-1}2}/\Gamma(\frac{d+1}2)=\frac{2^{d+1}\pi^{\frac d2}(\frac d2)!}{d!}$ is the volume of the unit sphere $\Sm^d$ in $\Rm^{d+1}$ (for $d$ even) and $\epsilon_d=(-1)^{\frac{d}2+1}$. 

\medskip

Let $d$ be odd. We assume $H=h(k)\cdot\Gamma_d$ is chiral so that $H\gamma_0+\gamma_0 H=0$ for $\gamma_0=\sigma_3\otimes I$. We define $F=\frac{H}{|H|}$. We also define $F=\sigma_-\otimes U + \sigma_+ \otimes U^*$. Then
\begin{equation}\label{eq:FL}
  \tr\ \gamma_0 F (dF)^d =  \dfrac{1}{|H|^{d+1}} \tr\ \gamma_0  H (dH)^{d} =  (2i)^{\frac{d+1}2} \dfrac{d! }{|h|^{d+1}} {\rm Det}(L) dk.
\end{equation}
As a consequence, the winding number in odd dimensions is given by 
 \begin{equation}\label{eq:Wnd}
    W_d[U] = c_d \dint_{\Xm^d} \tr (U^* dU)^d =  \frac{-c_d\epsilon_d}2 \dint_{\Xm^d}\tr \gamma_0 F (dF)^d = \dfrac{-1}{A_d}\dint_{\Xm^d} \dfrac{1}{|h|^{d+1}} {\rm Det}(L) dk,
\end{equation}
where $c_d=\frac{1}{2^d d!!} \big(\frac i\pi)^{\frac{d+1}2}$, $\epsilon_d=i^{d+1}=(-1)^{\frac{d+1}2}$ and $A_d=2\pi \pi^{\frac{d-1}2}/(\frac {d-1}2)!$ is the volume of the sphere $\Sm^d$ in $\Rm^{d+1}$ (for $d$ odd). The matrix $L$ is defined as in the even case.
\end{theorem}
We recall that $\sigma_\pm=\frac12(\sigma_1\pm i\sigma_2)$. In the above expressions, we recognize the formulas associated to the degree of the map $k\mapsto f(k)$, except that the vector field is not necessarily defined on a compact manifold. 

\begin{proof}  We treat the even and odd dimensional cases in turn. Throughout, we use $\gamma^j=\gamma^j_d$ and denote by $\gamma_0$ the chiral matrix; see Appendix \ref{sec:CA}.
\\[2mm]{\bf Even dimensional case.} We write $P=p\cdot\Gamma$ and $dP=dp\cdot\Gamma$. If the same matrix $\gamma^j$ appears twice in $(dP)^{d}$, then $\tr P(dP)^d=0$ from the antisymmetry of the exterior product. Similarly, $\tr(dP)^d=0$ so that we may replace $P$ by $-\frac12 \frac{H}{|H|}$. Let $F=\frac H{|H|}$. We want to estimate $\tr F(dF)^d$ and essentially replace $dF$ by $dH$. We observe that $F^2=I$ so that $FdF+dF F=0$. Using the latter relations, we find
	\[
		FdF = -dF F = -d\dfrac{1}{|H|}|H| - \dfrac{dH}{|H|} F.
	\]
The first contribution generates a term proportional to $\tr\ d|H| (dF)^{d-1}$. But since only $d-1$ different matrices can appear by antisymmetry of the exterior product, the trace vanishes. Therefore,
\[
  \tr\ F(dF)^{\wedge d} = -\dfrac{\tr}{|H|} dH F (dF)^{\wedge (d-1)} = \dfrac{\tr}{|H|^2} (dH)^2 F (dF)^{\wedge (d-2)} =\dfrac{\tr}{|H|^{d+1}} H(dH)^{\wedge d}.
\]
This proves the first equality in the lemma. Let ${\cal S}_n$ be the set of permutations of $n$ objects and $(-1)^\rho$ the signature of the permutation $\rho$. Using the above and the explicit expression for the exterior product, we have  \[
    \fT = \tr \ P (dP)^{d}= \dfrac{-1}{2^{d+1}}  \tr\ \dfrac{1}{|H|^{d+1}}  h\cdot\Gamma \dsum_{\rho\in{\cal S}_d} (-1)^\rho  \prod_{j=1}^d  \partial_{\rho(j)} h\cdot\Gamma \ dk.
  \]
  The above trace involves the product of $d+1$ terms and as such will have a non-trivial trace only when all elements in $\Gamma$ are present once. Therefore, 
  \[
 \fT= \dfrac{-1}{2^{d+1}}   \dfrac{1}{|H|^{d+1}}  \dsum_{\rho\in{\cal S}_d,\ \sigma\in{\cal S}_{d+1}} (-1)^\rho  h_{\sigma(d+1)} \prod_{j=1}^d  \partial_{\rho(j)} h_{\sigma(j)} \ \tr \  \prod_{j=1}^{d+1} \gamma^{\sigma(j)} \ dk.
  \]
  We have that $\tr \  \prod_{j=1}^{d+1} \gamma^{\sigma(j)} = (2i)^{\frac d2}(-1)^\sigma$ so that
  \[
  	\fT = \dfrac{-1}{2^{d+1}}(2i)^{\frac d2} \dfrac{1}{|h|^{d+1}}\dsum_{\rho\in{\cal S}_d,\ \sigma\in{\cal S}_{d+1}} (-1)^{\rho+\sigma}  h_{\sigma(d+1)}\prod_{j=1}^d  \partial_{\rho(j)} h_{\sigma(j)}\ dk.
  \] 
  For each fixed $\sigma(d+1)$, let $h'$ be the vector formed of the remaining $h_{\sigma(j)}$. Then the summation over $\rho$ gives ${\rm Det}(Dh')$, where $Dh'$ is the matrix of derivatives of the vector $h'$. Now, going from one $h'$ to another by permutation changes either both signs or none in ${\rm Det}(Dh')$ and in $(-1)^\sigma$. As a consequence, the summation over $\rho$ and $\sigma$ at $\sigma(d+1)$ fixed gives $d!$ times ${\rm Det}(Dh')$. Taking an explicit example within each such group of such permutations fixing $\sigma(d+1)$ and now summing over them gives
  \[
     \dfrac{-1}{2^{d+1}}(2i)^{\frac d2} \dfrac{d!}{|h|^{d+1}} \dsum_{j=1}^{d+1} (-1)^{j+1} {\rm Det} (h_{j,\sigma(j)}) h_{\sigma(d+1)} dk =  \dfrac{-1}{2^{d+1}}(2i)^{\frac d2}  \dfrac{d!}{|h|^{d+1}} {\rm Det}(L) dk.
  \]
  This proves the result in the even dimensional case.
  \\[2mm]
 {\bf Odd dimensional case.}  We recall that  $H=h\cdot\Gamma$ and $F=\frac{H}{|H|}$ are chiral, i.e., anticommute with $\gamma_0=\sigma_3\otimes I$. We recall from Appendix \ref{sec:CA} that $\Gamma=(\{\sigma_1\otimes \gamma^j\}_{1\leq j\leq d},\sigma_2\otimes I)$ where $\gamma^d=(-i)^{\frac{d-1}2}\gamma^1\ldots\gamma^{d-1}$. We still use that $FdF+dF F=0$ so that $FdF=-dF F = -\frac{dH}{|H|} F-d|H|^{-1} |H|$. 
  
  We have that $\tr\ \gamma_0 F(dF)^{d-1}=0$. The reason is as follows. All matrices in the $d-1$ terms $dF$ must be different by anti-symmetry of the exterior product, which implies that the matrices in all $dF$ term as well as in the $F$ term have to be of the form $\sigma_1\otimes*$. But then the trace product with $\gamma_0\sigma_3\otimes I$ clearly vanishes. As a consequence 
  \[
   \fT = \tr \  \gamma_0 F (dF)^d = \tr\ \gamma_0 \dfrac{-dH }{|H|} F (dF)^{d-1} = \tr\ \gamma_0 \dfrac{(dH)^d}{|H|^d}(-F)
  \]
  after $d$ (odd) iterations. Since $F\gamma_0+\gamma_0 F=0$, we get by cyclicity of the trace that 
  \[
     \fT =\dfrac{1}{|H|^{d+1}} \tr\ \gamma_0 H (dH)^d =  \dfrac{1}{|h|^{d+1}} \tr\ \gamma_0 h\cdot\Gamma (dh\cdot\Gamma)^d.
  \]
  As in the even case, we find that each matrix in $\Gamma$ appears once in order for the trace not to vanish, so that
  \[
    \fT = \dfrac{1}{|h|^{d+1}}\dsum_{\sigma\in{\cal S}_{d+1}}\dsum_{\rho\in{\cal S}_d}(-1)^\rho h_{\sigma(d+1)} \prod_{j=1}^d \partial_{\rho(j)}h_{\sigma(j)} \tr\ \gamma_0\Gamma_{\sigma(d+1)}\prod_{j=1}^d \Gamma_{\sigma(j)}.
  \]
  The last matrix trace is $-\tr\gamma_0\prod_{j=1}^{d+1}\Gamma_{\sigma(j)} = -(-1)^\sigma(2i)(2i)^{\frac{d-1}2}$. Thus,
  \[
   \fT =  -    \dfrac{(2i)^{\frac{d+1}2}}{|h|^{d+1}}\dsum_{\sigma\in{\cal S}_{d+1}}\dsum_{\rho\in{\cal S}_d}(-1)^\rho(-1)^\sigma h_{\sigma(d+1)} \prod_{j=1}^d \partial_{\rho(j)}h_{\sigma(j)} =  (2i)^{\frac{d+1}2} \dfrac{d! }{|h|^{d+1}} {\rm Det}(L) dk,
  \]
  as in the even case, where the change of sign comes from permuting the first column to last ($h_{\sigma}(d+1)$ moves to first line). For $d=2\kappa+1$, we find the volume of the sphere $A_d=2\pi V_{2\kappa}=2\pi \pi^k/k!$. Therefore,  the winding number is given by
  \[
    W_d[U] = c_d \dint_{\Xm^d} \tr (U^* d U)^d = c_d\epsilon_d  (-\frac12) \dint_{\Xm^d}\tr (\gamma_0 F dF)^d,
  \]
  so that
  \[
    W_d[U] =c_d\epsilon_d\frac{-1}2 (2i)^{\frac{d+1}2} d!  \dint_{\Xm^d} \dfrac{1}{|h|^{d+1}} {\rm Det}(L) dk = \dfrac{-1}{A_d}\dint_{\Xm^d} \dfrac{1}{|h|^{d+1}} {\rm Det}(L) dk.
  \]
 This concludes the proof in the odd dimensional case.
\end{proof}

For the Hamiltonians of interest, we now need to relate the above calculations to a degree when $\Xm^d=\Rm^d$ (or when $\Xm^d=\Tm^d$) and $h$ takes one value at infinity so that it can be pulled back to the sphere $\Sm^d$. We then relate invariants to ${\rm deg}(h)$, the degree of the vector field $h$.

\subsection{The bulk even-dimensional case}
\label{sec:higheven}

We come back to the setting of Hamiltonians acting on $d$-dimensional spaces (the non-commutative setting). We first start with the even-dimensional case $d=2\kappa$. We use the notation $\gamma^j$ instead of $\gamma^j_d$.

The Hamiltonian acting on $L^2(\Rm^d)\otimes\Cm^{2^{\kappa}}$ is then given by
\[
  H = \dfrac{1}{i} \dsum_{j=1}^d \partial_{x_j}  \gamma^j + m \gamma^{d+1},
\]
and we recall that $H_\eta=H+\eta\Delta\gamma^{d+1}$ is given as $H$ above with $m$ replaced by  $m_\eta = m+\eta \Delta$. 

Let us look at the unitary matrix $\Un$ that reveals the topology of idempotent functionals of $H_\eta$. Since all $\gamma^j$ for $1\leq j\leq d$ anti-commute with $\gamma^{d+1}$, they can be written as
\[
	\gamma^j = \sigma_- \otimes \check\gamma^j + \sigma_+ \otimes (\check\gamma^j)^*,
\] 
for an appropriate (non-Hermitian) matrix $\check\gamma^j$ in $\Cm^{2^{\kappa-1}}$. From the construction recalled in Appendix \ref{sec:CA}, we verify for concreteness that $\check\gamma^j=\gamma^j_{d-2}$ for $1\leq j\leq d-1$ and $\check\gamma^{d} = i I_{2^{\kappa-1}}$ when $d=2\kappa\geq4$ while  $\check\gamma^1=1$ and $\check\gamma^2=i$ when $d=2$.

Using the notation $\check\gamma=(\check\gamma^1,\ldots,\check\gamma^d)$, we then define the unitary
\begin{equation}\label{eq:Ueven}
  \Un(x) = \dfrac{x\cdot\check\gamma}{|x\cdot\check\gamma|}.
\end{equation}
The above expression is indeed scalar-valued when $d=2$. 
Then, for any projector $P=P(H_\eta)$, we look for the index of the operator $P\otimes I \ I\otimes \Un \  P\otimes I$, which is Fredholm when $H$ is sufficiently regularized. In the above construction, the operators $P$ and $\Un$ act from the matrix point of view on different components of the tensor product. The `left' $I$ above is identity on $\Cm^{2^{\kappa}}$ while the `right' I is identity on $\Cm^{2^{\kappa-1}}$.
These operators are therefore defined on $L^2(\Rm^d)\otimes\Cm^{2^{\kappa}}\otimes \Cm^{2^{\kappa-1}}$.

We define $P=P_-=\chi(H_\eta\leq0)$, the projection onto the negative part of the regularized Hamiltonian $H_\eta$. We also introduce the convenient notation $\tilde P=P\otimes I$ and $\tilde\Un=I\otimes\Un$ and define $\tilde Q=\tilde \Un \tilde P \tilde \Un^*$. 

Then for $\eta\not=0$ (sufficiently small), we have the main result of this section:
\begin{theorem}
\label{thm:bulkeven}
Let $\tilde P$ and $\tilde\Un$ be defined as above. Then $\tilde P\tilde \Un\tilde P$ is Fredholm on Ran$\tilde P$. Moreover, $(\tilde P-\tilde Q)^{d+1}$ is trace-class and
\begin{equation}\label{eq:bulkevenindex}
	-\Ind{\tilde P\tilde \Un\tilde P} = \Tr (\tilde P-\tilde Q)^{d+1} =\Tr\  \tilde \Un \big([\tilde P,\tilde \Un^*][\tilde P,\tilde \Un]\big)^{\frac d2}[\tilde P,\tilde \Un^*]=\frac12 (\sgn{m}+\sgn{\eta}).
\end{equation}
\end{theorem}
\begin{proof}
 The equality of both traces is a simple verification based on $\tilde P-\tilde Q=[\tilde P,\tilde \Un]\tilde \Un^*=-\tilde\Un[\tilde P,\tilde \Un^*]$. 
 
 That $(\tilde P-\tilde Q)^{d+1}$ is trace-class is shown as in dimension $d=2$ using Russo's result in Lemma \ref{lem:russo}, which requires $\eta\not=0$. The proof mimics that of Proposition \ref{prop:1index2d}. We highlight the main differences. We also obtain that the kernel of $p$ decays rapidly at infinity thanks to the spectral gap and that the main contribution at $x=0$ is of the form $p_{-1}(x)=C\frac 1\eta \frac{x\cdot\gamma}{|x|^{d}}$. The estimate on the kernel $k$ of $P-Q$ is thus replaced by an estimate of the form \eqref{eq:bdk} with $\alpha=d-1$ and $\beta=d+1$, say. We apply the result of Lemma \ref{lem:estip} with $p=d+1$ and $n=d$. This shows that $(\tilde P-\tilde Q)$ is in the Schatten class ${\cal I}_{d+1}$ and $(\tilde P-\tilde Q)^{d+1}$ is trace-class.
 
 Let us define ${\cal T}:=\Tr (\tilde P-\tilde Q)^{d+1}$. We calculate the trace as the integral of the operator's kernel along the diagonal. 
 
 The kernel of $\tilde \Un$ is multiplication by $I\otimes \Un(x)$ while the kernel of $P$ is multiplication by $p(x,y)\otimes I$. In the translation invariant setting considered now, this is $p(x-y)\otimes I$. Note that 
 $$
 {\cal T}=-\Tr \ \tilde \Un([\tilde \Un^*,\tilde P][\tilde \Un,\tilde P])^{\frac d2}[\tilde \Un^*,\tilde P].
 $$
 The kernel of $[\tilde \Un,\tilde P]$ is given by $p(x-y) \otimes (\Un(x)-\Un(y))$, where we recall that $p(x-y)$ is an operator on $\Cm^{2^\kappa}$ while $\Un(x)$ is a multiplication operator on $\Cm^{2^{\kappa-1}}$. The trace ${\cal T}$ is therefore given by
 \[ \begin{array}{rlll}
   \dint_{\Rm^{d\times (d+1)}} &-\tr \ \Big( \Un(x_{d+2})(\Un^*(x_{d+2})-\Un^*(x_{d+1}))\ldots (\Un^*(x_{2})-\Un^*(x_1))\Big)
   \\
   \times & \tr \ p(x_{d+2},x_{d+1}) p(x_{d+1},x_{d}) \ldots  p(x_2,x_1) & \prod_{k=1}^{d+1} dx_k,
   \end{array}
 \]
 where we have identified $x_{d+2}$ with $x_1$. Let us introduce $x=x_1$, $y_j=x_{j+1}-x$ so that $x_{j+1}=y_j+x$ for $1\leq j\leq d$. Since $p(x,y)=p(x-y)$, this yields
 \[
  - \dint \tr (\Un(x)(\Un^*(x)-\Un^*(x+y_d))\ldots (\Un^*(x+y_1)-\Un^*(x))) \tr (p(-y_d)p(y_d-y_{d-1})\ldots p(y_1)) dx \prod dy_k.
 \]
 We now use the crucial geometric identity recalled in Lemma \ref{lem:geometric} and stating that
\[
 	\dint_{\Rm^d} \tr \ \Un(x)(\Un^*(x)-\Un^*(x+y_d))\ldots (\Un^*(x+y_1)-\Un^*(x)) dx =\dfrac{(2\pi i)^{\frac d2}}{\frac d2!} {\rm Det}(y_d,\ldots,y_1).
\]
The trace ${\cal T}$ is therefore given by
\[
  -\dfrac{(2\pi i)^{\frac d2}}{\frac d2!}  \dint_{\Rm^{d\times d}} \tr p(-y_d)p(y_d-y_{d-1})\ldots p(y_1)\   {\rm Det}( y_d,\ldots, y_1) \prod dy_k.
\]
Let us describe the above term in the Fourier domain. We recall our convention
\[
  \hat Q(k) = \dint_{\Rm^d} e^{-ik\cdot x} Q(x)dx,\quad Q(x) = \dfrac{1}{(2\pi)^d} \dint_{\Rm^d}  e^{ik\cdot x} \hat Q(k)dk,
\]
so that $(2\pi)^d \int_{\Rm^d}  P(-x) Q(x)dx =  \int_{\Rm^d}  \hat P(k) \hat Q(k) dk$. Thus,
\[
\dint_{\Rm^d} \tr \hat P(k) d\hat P(k) \wedge \ldots d\hat P(k) = \dint_{\Rm^d}  \tr \hat P \dsum_\rho (-1)^\rho \partial_{\rho(1)} \hat P \ldots \partial_{\rho(d)} \hat P(k) dk.
\]
This is, with $\partial_j\hat f$ the Fourier transform of $-ix_j f$,
\[
 \dint_{\Rm^{d\times d}} (2\pi)^d \tr P(-x_1) \dsum_\rho (-1)^\rho [-i(x_1-x_2)_{\rho(1)}] P(x_1-x_2)\ldots (-ix_d)_{\rho(d)} P(x_d) \prod dx_j,
\] 
or
\[
  (2\pi i)^d \dint_{\Rm^{d\times d}} {\rm Det}(x_1,\ldots,x_d) \tr\ P(-x_1)P(x_1-x_2)\ldots P(x_d) \prod dx_j.
\]
We thus obtain that the trace ${\cal T}$ is given by $\epsilon_d\Ch{d}[\hat P]$, the even dimensional Chern number defined in \eqref{eq:Chnd}.
\footnote{This is the same expression as  in \cite[p.24]{prodan2016bulk} and the standard expression for the Chern character as the trace of the exponential of $\frac i{2\pi}$ times curvature $\hat P d\hat P\wedge d\hat P$. From $\hat P^2=\hat P$, we deduce that $(\hat P d\hat P\wedge d\hat P)^\kappa = \hat P (d\hat P)^{\wedge 2\kappa}$, the latter being the form we have above. In dimension $d=2$, we retrieve (\ref{eq:indexTr3}). The factor $\epsilon_d$ comes from our Clifford convention and is the same as $\chi$ in \cite[p.33]{prodan2016bulk}.}
The Fourier transform $\hat P(k)$  of $p(x)$  is also found to be
\[
 \hat P(k) =\dfrac I2 - \frac12 \dfrac{k\cdot\gamma + (m-\eta|k|^2)\gamma_0}{(|k|^2 + (m-\eta|k|^2)^2)^{\frac12}} = \dfrac I2 - \dfrac12 \dfrac{\hat H}{|\hat H|},\qquad \hat H= k\cdot\gamma + (m-\eta|k|^2)\gamma_0,
\]
with $|\hat H|$ scalar-valued. Define $\Gamma=(\gamma,\gamma_0)$ and $h(k)=(k,\phi_\eta(k))$, with $\phi_\eta(k)=m-\eta|k|^2$.

Using the results of Theorem \ref{thm:degreefield}, we obtain that 
\[
\Ch{d}[\hat P] = \dfrac{\epsilon_d}{A_d} \dint_{\Rm^d} \dfrac{1}{|h|^{d+1}} {\rm Det} (L) dk ,
\]
where $L$ is the $(d+1)\times(d+1)$ matrix with $h$ as its first row and $\partial_{k_j}h$ as its $(j+1)$th row for $1\leq j\leq d$ and $A_d$ is the volume of the unit sphere $\Sm^d$ in $\Rm^{d+1}$.

The integral involves an open domain $\Rm^d$ that is not yet amenable to topological simplifications. We consider the orientation preserving stereographic projection $\pi$ from $\Sm^d$ to $\Rm^d$ which maps the north pole to the sphere at $\infty$, and observe that as $|k|\to\infty$, then $\hat P(k)$ depends only on $|k|$. In other words, the value of $\hat P(k)$ is constant at infinity. Therefore, $\pi$ maps $\Sm^d$ to the one point compactification $\Rm^d\cup\{\infty\} \cong \Sm^d$ and the pull-back $\pi^*\hat P$ is a continuous map on $\Sm^d$. As a consequence, we have that
\[
     \Ch{d}[\hat P] =  \dfrac{i^{\frac d2}}{(2\pi)^{\frac d2}(\frac d2)!} \dint_{\Sm^d} \tr \ \pi^* \hat P (d\pi^*\hat P)^{d}  = \dfrac{\epsilon_d}{A_d} \dint_{\Sm^d} \dfrac{1}{|\pi^*h|^{d+1}} {\rm Det} (\pi^*L) \pi^*dk .
\]
The above quantity is the degree of the continuous map from $\Sm^d$ to $\Sm^d$ given by $\pi^*k\to \frac{\pi^*h}{|\pi^*h|}=\frac{(k,\phi_\eta(k))}{(|k|^2+|\phi_\eta(k)|^2)^\frac12}$ with $k=\pi(\theta)$ for $\theta$ parametrizing $\Sm^d$; see \cite[Chapter 13]{DFN-SP-1985}.

To calculate the degree of the above map, we also refer to the method in \cite[Chapter 13]{DFN-SP-1985}, which is more straightforward than the explicit integral calculations we considered in dimension two.
Let us consider $m>0$ and $\eta>0$. Then the point $e_d=(0,\ldots,0,1)$ is a value attained once as $k=0$ while the map converges to $-e_d$ as $|k|\to\infty$. The gradient of the map is given by $\frac{1}{|m|}I$ at that point in the hyperplane tangent to $e_d$. We therefore obtain that the degree of the map is equal to $1$. When $m>0$ and $\eta<0$, we obtain that the map never visits the vicinity of $-e_d$ so that the degree of the map is trivial. Now, for $m<0$ and $\eta$ changing sign, we observe that the above ${\rm Det}(L)$ changes sign compared to the case since one column changes sign. This shows that the degree is $-1$ when $m<0$ and $\eta>0$ while the degree vanishes again when $m$ and $\eta$ are both negative. We could also estimate the gradient of the map at $k=0$ and observe that it is still given by $\frac{1}{|m|}I$. However, at the regular point $-e_d$, the orientation of the tangent hyperplane is reversed and so, signs also need to be reversed.

Note that the argument is independent of the regularization term so long as the map converges to $-e_d$ as $|k|\to\infty$ when $m$ and $\eta$ are positive. For instance, we can replace $\eta|k|^2$ by $\eta|k|^\alpha$ with any $\alpha>1$. For $\alpha\leq1$, as when $\eta=0$, the regularization is not sufficiently strong to force the map to be uniquely defined as $|k|\to\infty$ and the degree is not defined.

This shows that with the above conventions, we have
\[
 {\cal T} = - \Ind{P\Un P} =\epsilon_d \Ch{d}[\hat P] =\dfrac12(\sgn{m}+\sgn{\eta}).
\]
This concludes the proof of the result.
\end{proof}

Once we have this result, we can consider fluctuations $V$ that are relatively compact with respect to $H$. More precisely, assuming that the assumptions in Proposition \ref{prop:HSF} are satisfied, we can construct the following invariant as in dimension $d=2$. Let $\delta<m$ and define $P_\delta(E)$ the smooth approximation of $P(E)=\chi(E\leq0)$ equal to that function outside of $|E|<\delta$.

Let $H_V=H+V$ and $\tilde P_\delta=P_\delta(H_V)\otimes I$. Then we have
\begin{theorem}
\label{thm:bulkindexevenperturb}
 Let $H_V$ and $\tilde P_\delta$ as above. Let $V$ satisfy \eqref{eq:compactres} and \eqref{eq:boundres}. Then $\tilde P_\delta \tilde\Un\tilde P_\delta+I\otimes I-\tilde P_\delta$ is Fredholm and its index is that of the unperturbed operator $\frac12(\sgn{m}+\sgn{\eta})$.
\end{theorem}
The derivation of the theorem is exactly the same as in the two dimensional setting.

\subsection{The bulk odd-dimensional case}
\label{sec:highodd}

We now consider the case $d=2\kappa+1$ with $\kappa\geq1$ \footnote{The case $d=1$ with $\kappa=0$ can be handled similarly for the bulk invariant with some specificities coming from the zero-dimensional `interface' \cite{prodan2016bulk} that we do not consider here.} and $H[m]$ given by \eqref{eq:Hmnd} and acting on $L^2(\Rm^d)\otimes \Cm^{2^{\kappa+1}}$. We still use the notation $\gamma^j=\gamma^j_d$ with now a chirality matrix given by $\gamma_0:=\gamma_d^{d+2}$.  The bulk Hamiltonian satisfies the chiral symmetry $\gamma_0 H[m]+H[m]\gamma_0=0$. We recall the definition of $F=F[H]$ and $U=U[H]$ after appropriate regularization with $\eta\not=0$
\[
  F = \dfrac{H_\eta}{|H_\eta|},\qquad H_\eta = H_\eta[H] = H + \eta \Delta \gamma_d^{d+1},\quad F = \left(\begin{matrix} 0& U^* \\ U & 0 \end{matrix}\right).
\]
The objective is to assign a topological invariant to $H[m]$ by means of $U[H]$ or equivalently $F[H]$. 

This is achieved by pairing $U[H]$ with the following projection operator $\Pr$  in odd dimension
\begin{equation}\label{eq:PUodd}
\Un=\sgn{x\cdot\gamma} =\sgn{\sum_{j=1}^d x_j \gamma^j},\qquad \Pr=\frac12(I+\Un),
\end{equation}
which take values in  matrices on $\Cm^{2^\kappa}$. 

Following Appendix \ref{sec:FH},  we form the operator $I\otimes \Pr\ U[H]\otimes I\ I\otimes \Pr$, which we will prove is Fredholm since  $U[H]$ is appropriately regularized. Here, the first (left) $I$ is identity on $\Cm^{2^{\kappa+1}}$ while the second (right) $I$ is identity on $\Cm^{2^\kappa}$. These operators are defined on  $L^2(\Rm^d)\otimes\Cm^{2^{\kappa+1}}\otimes \Cm^{2^{\kappa}}$. We define $\tilde\Pr=I\otimes \Pr$ and $\tilde U=U[H]\otimes I$ as well as $\tilde\Qr=\tilde U\tilde\Pr\tilde U^*$ and have the following result:
\begin{theorem}
\label{thm:bulkodd}
Let $\tilde U$ and $\tilde\Pr$ be defined as above and let $\eta\not=0$ sufficiently small. Then $\tilde\Pr \tilde U\tilde\Pr$ is Fredholm on Ran$\tilde \Pr$. Moreover, $(\tilde \Pr-\tilde \Qr)^{d+1}$ is trace-class and
\begin{equation}\label{eq:bulkoddindex}
	-\Ind{\tilde\Pr \tilde U\tilde\Pr} = \Tr (\tilde \Pr-\tilde \Qr)^{d} =\Tr\  \tilde U \big([\tilde \Pr,\tilde U^*][\tilde \Pr,\tilde U]\big)^{\frac {d-1}2}[\tilde \Pr,\tilde U^*]=\frac12 (\sgn{m}+\sgn{\eta}).
\end{equation}
\end{theorem}
\begin{proof} 
We want to show that the index ${\cal T}:= -\Ind{\tilde\Pr \tilde U\tilde\Pr}$ is given by 
\[
{\cal T}=\Tr(\tilde\Pr-\tilde\Qr)^d= (-1)^\frac{d-1}2 \Tr\ \tilde U^*([\tilde\Pr,\tilde U][\tilde\Pr,\tilde U^*])^{\frac{d-1}2}[\tilde\Pr,\tilde U] =: (-1)^{\frac{d-1}2} \fT.
\]
Indeed, $\tilde \Pr-\tilde \Qr=[\tilde \Pr,\tilde U]\tilde U^*=-\tilde U[\tilde\Pr,\tilde U^*]$ so that (with $d=2\kappa+1$)
\[
   (\tilde \Pr-\tilde \Qr)^d = (-[\tilde \Pr,\tilde U][\tilde\Pr,\tilde U^*])^\kappa [\tilde\Pr,\tilde U]\tilde U^*,
\]
and the equivalence follows by cyclicity of the trace.

We now apply Russo's result Lemma \ref{lem:russo} using Lemma \ref{lem:estip} as in the even-dimensional case. We observe that the kernel of $(\tilde \Pr-\tilde \Qr)$ involves terms of the form $(\tilde\Pr(x)-\tilde\Pr(y)) \tilde U(x-y)$. From $\tilde\Pr(x)-\tilde\Pr(y)$, we obtain a bound of the form ${\rm min}(1,\frac{|x-y|}{|y|})$. From $\tilde U(x-y)$, we obtain a contribution of the form $\delta(x-y)$, which does not contribute in the commutator and then a leading contribution as in the even dimensional case of the form $\eta^{-1}|x-y|^{1-d}\varphi(x-y)$ with $\varphi(x-y)$ rapidly decaying.

We then apply Lemma \ref{lem:estip} with $\alpha=d-1$ and $\beta=d+1$. 
We cannot use the result with $p=d$ directly as we would like. However, the result applies for each $p>d$ with a bound in $\mI_p$ with $\alpha=d-1$ and $\beta=d+1$ uniform in $p>d$. This shows that $0\leq\lambda_j$ the singular values of $P-Q$ satisfy $\| \lambda \|_p^p \leq C^p$ with $C$ independent of $p>d$. Since $C^p$ converges to $C^d$ and $\sum_{j=1}^N\lambda_j^p$ converges to $\sum_{j=1}^N\lambda_j^d$ for any $N\in\Nm$ as $p\to d$, we deduce that $\lambda\in l^d$.

This shows that $(\Pr-\Qr)^{d}$ is trace-class. Moreover, we deduce from Corollary \ref{cor:trace} that the trace is obtained as the integral of its Schwartz kernel along the diagonal. 
Thus, ${\cal T}$ is given by the integral of the kernel along the diagonal and $\fT$ is equal to
\[ \begin{array}{rlll}
   \dint_{\Rm^{d\times (d+1)}} &\tr \ \Big( U^*(x_{d+2},x_{d+1}) U(x_{d+1},x_{d}) \ldots U(x_2,x_1)\Big)
   \\
   \times & \tr \ (p(x_{d+1})-p(x_{d})) (p(x_{d})-p(x_{d-1})) \ldots  (p(x_2)-p(x_1)) & \prod_{k=1}^{d+1} dx_k .
   \end{array}
 \]
With the same change of variables as in the even case $x_1\equiv x_{d+2}=x$, $y_j=x_{j+1}-x$ for $1\leq j\leq d$, we find
\[ \begin{array}{rlll}
   \dint_{\Rm^{d\times (d+1)}} &\tr \ \Big( U^*(-y_d) U(y_d-y_{d-1}) \ldots U(y_1)\Big)
   \\
   \times & \dfrac{1}{2^{d}}\tr \ (\Un(y_{d}+x)-\Un(y_{d-1}+x)) \ldots  (\Un(y_1+x)-\Un(x)) & \prod_{k=1}^{d+1} dx_k.
   \end{array}
\]
We now use the geometric identity (see Lemma \ref{lem:geometric})
\[
 \dint_{\Rm^d} \tr\ (\Un(y_{d}+x)-\Un(y_{d-1}+x)) \ldots  (\Un(y_1+x)-\Un(x)) dx = \dfrac{2^d}{d!!} (i\pi)^{\frac{d-1}2} {\rm Det}(y_d,\ldots,y_1).
\]
The trace $\fT$ is therefore given by
\[
	\dint_{\Rm^{d\times (d+1)}} \tr \ \Big( U^*(-y_d) U(y_d-y_{d-1}) \ldots U(y_1)\Big) \dfrac{1}{d!!} (i\pi)^{\frac{d-1}2} {\rm Det}(y_d,\ldots,y_1) dy_1\ldots dy_d.
\]
As in the even case, let us start from the expression in the Fourier domain, with $\hat U$ the Fourier transform of $U$:
\[
 \dint_{\Rm^d} \tr\  \hat U^* (d\hat U \wedge d\hat U^*)^{\wedge \frac{d-1}2}\wedge d\hat U =\dint_{\Rm^d} \tr\  \hat U^* \dsum_{\rho\in{{\cal S}_d}}(-1)^\rho \partial_{\rho(1)} \hat U \ldots \partial_{\rho(d)} \hat U dk,
\]
where $\hat U$ and $\hat U^*$ alternate. Again, with $\partial_j\hat f$ the Fourier transform $-ix_j f$, this is
\[ 
 	\dint_{\Rm^{d\times d}} \tr\ U^*(-x_1)\dsum_{\rho} (-1)^\rho [-i(x_1-x_2)_{\rho(1)}] U(x_1-x_2)\ldots (-ix_d)_{\rho(d)} U(x_d) \prod dx_j,
\]
or
\[
(2\pi i)^d \dint_{\Rm^{d\times d}} {\rm Det}(x_1,\ldots,x_d) \tr\ U^*(-x_1) U(x_1-x_2)\ldots U(x_d) \prod dx_j.
\]
As a consequence, moving to the Fourier domain as was done in the even dimension case, we obtain that
\[
  \fT = \dfrac{1}{i^d} \dfrac{1}{d!!} (i\pi)^{\frac{d-1}2} \dfrac{1}{(2\pi)^d}\dint_{\Rm^d} \tr\  \hat U^* (d\hat U \wedge d\hat U^*)^{\wedge \frac{d-1}2}\wedge d\hat U. 
\]
From $d\hat U^* \hat U + \hat U^* d\hat U=0$ and $d\hat U^* \wedge d\hat U=d\hat U^*\hat U\wedge \hat U^* d\hat U$, we have
\[
	\hat U^* (d\hat U \wedge d\hat U^*)^{\wedge \frac{d-1}2}\wedge d\hat U = \hat U^* (-d\hat U \wedge \hat U^*  d\hat U \hat U^*)^{\wedge \frac{d-1}2}\wedge d\hat U= (-1)^{\frac{d-1}2}(\hat U^*d\hat U)^{\wedge d}.
\]
We deduce from the first equality in \eqref{eq:Wnd} that ${\cal T}=(-1)^{\frac{d-1}2}/(i^{d+1})W_d[\hat U] = -W_d[\hat U]$. From Theorem \ref{thm:degreefield} in odd dimensions, we obtain that ${\cal T}$ is given by the degree of $h$. As in the even-dimensional setting, we pull back the map $h$ to the sphere by the stereographic projection. The calculation is then obtained as in the even-dimensional setting, where the degree is found to equal $\pm1$ when $m$ and $\eta$ have that same sign and $0$ otherwise. This concludes the proof of the theorem.
\end{proof}

Let us now consider the perturbed case $H_V=H+V$ for $V$ a perturbation that satisfies the chiral symmetry $V\gamma_0+\gamma_0 V=0$. As in the even dimensional case, we need to regularize the functional calculus to propagate the perturbation $V$. This is done as follows. We recall that $F$ is defined as the sign of $H_\eta$ in the unperturbed case. We replace $F$ by $F_\delta(E)=2\chi_\delta(E)-1$ an approximation, where we choose $\chi_\delta$ such that $F_\delta(-E)=-F_\delta(E)$ is an odd function. As a consequence, we still verify that $F_\delta(H_V)\gamma_0+\gamma_0 F_\delta(H_V)=0$ and hence have the decomposition
\[
  F_\delta(H_V) = \left(\begin{matrix} 0 & U_\delta^*\\ U_\delta & 0 \end{matrix}\right),
\]
for $U_\delta=U_\delta(H_V)$ a bounded operator that is no longer necessarily unitary. However, working on $F_\delta$, we observe that the latter function is constant outside of a compact interval $|E|<\delta$. Therefore, provided that \eqref{eq:compactres} holds, then $F_\delta(H_V) -F_\delta(H)$, and hence $U_\delta(H_V) -U_\delta(H)=U_\delta(H_V)-U[H]$, is compact by Proposition \ref{prop:HSF}. As a consequence, we have the result
\begin{theorem}\label{thm:bulkperturbodd}
	Let $H_V$ be a chiral operator and $\tilde U_\delta=U_\delta\otimes I$ with $U_\delta$ as above. Provided that \eqref{eq:compactres} holds, then $\tilde\Pr \tilde U_\delta \tilde \Pr$ is Fredholm on the range of $\tilde\Pr$ and the index is the same as in the case $V\equiv0$.
\end{theorem}

Note that the perturbation $V$ is required to satisfy the chiral symmetry in order for the operator $U_\delta$ to be properly defined. As in \cite{prodan2016bulk}, we could assume that $V$ has a small non-chiral component in the uniform norm and verify that the index is still defined by continuity. It is however unlikely that the index is stable with respect to large non-chiral fluctuations, as reflected by the fact that materials of class A are topologically trivial in odd dimensions; see aforementioned reference.

\subsection{Interface Hamiltonians}
\label{sec:interfacend}

We now consider the Hamiltonians $H[m(x)]$. The Hamiltonians are defined with $\eta=0$. We saw in the two-dimensional case that $\eta$ did not have any role in the topological properties of interface Hamiltonians. While the construction of bulk topological invariants requires $\eta$-dependent projectors or unitaries, this is not needed for the interface Hamiltonians and we therefore assume $\eta=0$ in the rest of the section. 

Let us consider the above settings with $d=2\kappa$ even or $d=2\kappa-1$ odd with 
\[
  H[m(x)] = \dfrac 1i \nabla'\cdot\gamma'_{d} + \dfrac1i\partial_{x_d}\gamma^d_{d} + m(x_d)\gamma_{d}^{d+1},\qquad \gamma_d=(\gamma_d',\gamma_d^d)^t.
\]
We assume that $m(x_d)$ is continuous and converges to $m_\pm$ as $x\to\pm\infty$ for $|m_\pm|\geq m_0>0$ and that $m(x_d)$ is not equal to either constant only on a compact domain. As before, we define
\begin{equation}\label{eq:epsnd}
  \eps = \dfrac12 ( \sgn{m_+}-\sgn{m_-}).
\end{equation}

As we did in the two dimensional setting, it is convenient to modify the representation of the interface Hamiltonian so that $m(x_d)$ and $\partial_{x_d}$ appear in the same matrix entry. Unlike the two-dimensional case, where $x$ modeled the signed distance to the interface, this role is now assumed by $x_d$. Following the Clifford representation recalled in Appendix \ref{sec:CA}, we perform a change of Dirac matrices that preserves their commutativity relation and rewrite the interface Hamiltonian as
\begin{equation}\label{eq:intH}
	H[m(x_d)] = \dfrac 1i \nabla'\cdot \sigma_3\otimes \gamma' - \dfrac1i \partial_{x_d}\sigma_2\otimes I + m(x_d) \sigma_1\otimes I,
\end{equation}
where $\gamma'=\Gamma_{d-2}$  in even dimension while $\gamma'=\Gamma_{d-1}$ in odd dimension. Here, $\nabla'$ is the usual gradient operator in the first $d-1$ variables. In both cases, we have that $-\frac1i\partial_{x_d}\sigma_2+ m(x_d)\sigma_1=\fa\sigma_++\fa^*\sigma_-$, where $\fa=\partial_{x_d}+m(x_d)$ and where $\sigma_\pm=\sigma_1\pm i\sigma_2$.  The Hamiltonian therefore takes the form
\[
H[m(x_d)] = \dfrac 1i \nabla'\cdot \sigma_3\otimes \gamma'  + (\fa\sigma_++\fa^*\sigma_-) \otimes I.
\]
In the partial Fourier transform variables, we find that 
\begin{equation}\label{eq:partialHnd}
   \hat H(k',x_d) =  k'\cdot \sigma_3\otimes \gamma' + (\fa\sigma_++\fa^*\sigma_-) \otimes I = \left(\begin{matrix}   	  k'\cdot\gamma' & \fa \\ \fa^* & -k'\cdot\gamma'\end{matrix}\right).
\end{equation}
All block matrices are square matrices of size $2^{\kappa-1}$, both in even dimension $d=2\kappa$ and in odd dimension $d=2\kappa-1$.

Let us assume that $m(x_d)$ is monotone to simplify. This implies the existence of only one continuous mode within the bulk gap when $\eps\not=0$. With the above expression, we find
\[
  \hat H^2(k') = {\rm Diag}(|k'|^2+\fa\fa^*,|k'|^2+\fa^*\fa)
\]
each block of size $2^{\kappa-1}$. As in $d=2$ and the proof of Lemma \ref{lem:astara}, the continuous spectrum of $\fa\fa^*$ or $\fa^*\fa$ does not intersect with $[0,m_0^2]$. For each $k'$, the spectrum of $H(k')$ is therefore discrete when restricted to the bulk gap. Since $m$ is monotone, the positive (joint) discrete spectrum of $\fa\fa^*$ and $\fa^*\fa$ is also outside of $[0,m_0^2]$. Therefore, only spectrum associated to the kernels of $\fa$ and $\fa^*$ is allowed within the bulk gap. Let us construct
\[
  k'\cdot\gamma' \psi_1 + \fa\psi_2=E\psi_1,\qquad \fa^*\psi_1-k'\cdot\gamma' \psi_2 = E\psi_2.
\]
For $|E|>|k'|$, then we find non-trivial positive spectrum of either $\fa\fa^*$ or $\fa^*\fa$, which is not within the bulk gap. Therefore, $|E|=|k'|$ so that from the structure of $\hat H^2$, either the kernel of $\fa$ or that of $\fa^*$ is non-trivial. When the kernel of $\fa$ is non-trivial, i.e., $\eps=1$, that of $\fa^*$ is and $\psi_1=0$ with $\fa\psi_2=0$ and $-k'\cdot\gamma' \psi_2 = E\psi_2$. When the kernel of $\fa^*$ is non-trivial instead, i.e., $\eps=-1$, then $\psi_2=0$ and $\fa^*\psi_1=0$ and $k'\cdot\gamma' \psi_1 = E\psi_1$. When $\eps=0$, then there is no spectrum inside the bulk gap. 

Define the interface Hamiltonian
\[
  H_I = \dfrac1i\nabla'\cdot \gamma',\qquad \hat H_I(k') = k'\cdot\gamma',
\]
which is a mass-less Dirac operator in the first $d-1$ variables. Then we observe that $E$ is nothing but a description of the spectrum of $-\eps H_I$; see \eqref{eq:partialHnd}. In other words, let $\varphi_0(x_d)$ be the non-trivial vector above, either $(\psi_1,0)^t$ or $(0,\psi_2)^t$ depending on the sign of $\eps$ and let $\Pi_0=\varphi_0\otimes\varphi_0$ be the associated rank-one projector. Note that $\varphi_0$ does not depend on $k'$ when $\eta=0$. Non-vanishing $\eta$ couples $\varphi_0$ with $k'$ and needs to be handled as in the one-dimensional case, which we do not pursue here.

Then let $f$ be a bounded real-valued function on $\Rm$ taking non-vanishing values only inside the bulk gap. The bounded operator $f(H[m(x)])$ defined spectrally is therefore given by
\[
  f(H[m(x)]) = f(\sigma_3\otimes H_I) \Pi_0 = \sigma_3\otimes f(H_I)  \Pi_0\equiv f(-\eps H_I) \Pi_0,
\]
where we slightly abuse notation and pick the component $f(-\eps H_I)$ of $\sigma_3\otimes f(H_I)$ of dimension $2^{\kappa-1}$ in the non-trivial part of the range of $\Pi_0$. As before, $\Pi_0$ displays the exponential behavior of the solution away from $x_d=0$ (wall or interface modes concentrated in the vicinity of the interface $x_d=0$) while $-\eps H_I$ characterizes the behavior in the transverse variables.

It remains to analyze $f(-\eps H_I)$ for appropriate functions $f$. As always, the parity of the dimension matters. When $d-1$ is odd, then $H_I$ is characterized by the winding number of a unitary operator. When $d-1$ is even, then it is characterized by the Chern number of a projection operator. Note that $H_I$ is a mass-less, un-regularized operator corresponding to $m=\eta=0$. We need to introduce a topological classification adapted to such operators.

\paragraph{Even dimension $d=2\kappa$; odd dimension $d-1=2\kappa-1$.} Since $d-1$ is odd, we need to construct a unitary operator from $H$ or $H_I$ as in the setting $d=2$. 

We recall that $\chi_\delta$ is an approximation of the characteristic function $\chi(E\geq0)$ such that $\chi_\delta(E)=\chi(E\geq0)$ outside of the spectral gap $(-m,m)$ and $0\leq\chi_\delta\leq1$ is smooth. Then we define the unitary
\begin{equation}\label{eq:Uinteven}
  \Uin(H) = e^{i 2\pi \chi_\delta(H[m(x)])}, \qquad \mbox{ and the compactly supported } \quad \Win(H) = \Uin(H)-I.
\end{equation}
We consider the operators $\Un$ and $\Pr$ in dimension $d-1$ already introduced in \eqref{eq:PUodd} and wish to calculate $\Ind{\tilde \Pr \tilde \Uin(H)\tilde \Pr}$ for $\tilde\Pr=I\otimes \Pr$ and $\tilde\Uin=\Uin\otimes I$ as in the bulk Hamiltonian setting.  We first prove the
\begin{theorem}\label{thm:interfaceodd}
	Let $\Uin(H)$ and $\Pr$ be defined as above. Then $\tilde \Pr \tilde \Uin(H)\tilde \Pr$ on the range of $\tilde\Pr$ is Fredholm and $-\Ind{\tilde \Pr \tilde \Uin(H)\tilde \Pr}$ is given by $-\eps$ defined in \eqref{eq:epsnd}.
\end{theorem}
Note that in section \ref{sec:2d}, the direction of transition between the two bulk materials was $x=x_1$, whereas it is now played by the last component $x_d$. This change of orientation is responsible for the change in the sign of the interface index.
\begin{proof}
We now prove this extension of Theorem \ref{thm:interface2d} to arbitrary even (odd interface) dimension.
The above analysis shows that $\Win(H)=\Win(-\eps H_I) \Pi_0$ since both operators are constrained within the spectral gap. It remains to analyze $\Win(-\eps H_I)$, which is given by
\[
  \Win(-\eps H_I) = {\cal F}^{-1} \dint_{|k'|<m_0}^{\oplus} \Win(-\eps \hat H_I) dk' \ {\cal F}.
\]
We can `forget' about $\Pi_0$ and the $x_d$ variable now since $\varphi_0$ is smooth and exponentially decaying at infinity. We need to show that $(\tilde\Pr-\tilde \Qr)^{d-1}$ is trace-class and that the trace is given by the integral of the kernel along the diagonal, where $d\geq4$ since $d=2$ was already treated (and is in fact harder). As in the proof of Theorem \ref{thm:bulkodd}, we wish to apply Lemma \ref{lem:estip} and Lemma \ref{lem:russo} but cannot do so for $p=n=d-1$. Instead, we still use that $(\tilde\Pr-\tilde \Qr)^{d-1}=(\tilde\Pr-\tilde \Qr)^{d+1}+(\tilde\Pr-\tilde \Qr)^{d}[\tilde\Pr\tilde\Qr,\tilde\Qr\tilde\Pr]$ and apply the latter lemma with $p=d$ (and $p=d+1$). The kernel of $\Win$ is compactly supported in $k'$ and hence smooth in $x'$ and smooth in $k'$ so that it decays as rapidly as necessary in $x'$ so that Lemma \ref{lem:estip}  applies as in the proof of Theorem \ref{thm:bulkodd}.

Armed with this, and still following the proof of Theorem \ref{thm:bulkodd} and the necessary geometric identity to pass from the spatial integral to an integral in the Fourier domain, it remains to calculate the index as a winding number. We do this as follows. We recall that $\Uin(H)=e^{i2\pi \chi_\delta(H)}$.

We first observe that $\chi_\delta(x)$ is asymptotically similar to $\frac{1}{\pi} \arctan x + \frac12$ and define
\[
  \Uin_a(H) = e^{i 2\arctan(H)}.
\]
We `forget' about the $e^{i\pi}=-1$ since the indices of $\Uin$ and $-\Uin$ are the same. Now consider a homotopy $\chi(t)=(1-t)\chi_\delta+t(\frac{1}{\pi} \arctan x + \frac12)$ and the corresponding $\Uin_t(H)$ linking $\Uin$ to $-\Uin_a$. We verify as we did in the preceding paragraph that all operators $\tilde\Pr \tilde \Uin_t(-\eps H_I) \tilde\Pr$ are Fredholm and since they are continuous in $t$ (in the uniform norm), their index is clearly the same. So, we now know that the index is given by that associated to $e^{i 2\arctan(H_I)}$. But the latter is an example we already know how to compute explictly. Indeed,
\[
  e^{i2\arctan h} = (\frac{1}{\sqrt{1+h^2}} + i \frac{h}{\sqrt{1+h^2}})^2 = \dfrac{1-h^2}{1+h^2} +2i\dfrac{h}{1+h^2}.
\]
We have $H_I^2=|k'|^2$. Multiplication of the unitary by $i$ does not change the index of interest. Up to multiplication by $i$, we thus recognize that 
\[
  i e^{i2\arctan H_I} = i\Uin_a(H_I) = \Uin(H_R) = \dfrac{h^R_{1}\gamma^1+\ldots h^R_{d-1}\gamma^{d-1} + i h^R_d}{|H_R|} 
\]
where
\[
  h^R=(k', \dfrac{1-|k'|^2}2),\qquad H_R = k'\cdot\Gamma' + \dfrac{1-|k'|^2}2 \gamma^d
\]
In other words, we observe that $i\Uin_a(H_I)$ is $\Uin(H_R)$ for a regularized (bulk) Hamiltonian $H_R$ corresponding to $m=\frac12$ and $\eta=\frac12$. We know that the winding number for $\Uin_a$ is therefore the degree of the map $h^R$ and hence equals one. 
Taking a unitary of $\arctan$ of the reduced interface Hamiltonian $H_I$ amounts to regularizing $H_I$ by both adding a mass term and a regularization $\eta|k'|^2$.  Replacing $H_I$ by $-\eps H_I$ proves the result. 
\end{proof}

Let now assume that $H_V=H+V$ is a perturbed Hamiltonian. Then $\Uin$ is perturbed via $\Win$, which is compactly supported. Therefore, provided that $V$ satisfies \eqref{eq:compactres}, then $\Win(H_V)-\Win(H)$ is compact. We deduce from Proposition \ref{prop:HSF} the
\begin{theorem}\label{thm:interfaceoddV}
 Let the perturbation $V$ satisfy \eqref{eq:compactres}. Then $\tilde \Pr\tilde \Uin(H_V) \tilde\Pr$ on the range of $\tilde\Pr$ is Fredholm and its index is the same as in the case $V\equiv0$.
\end{theorem}

\paragraph{Odd dimension $d=2\kappa+1$; even dimension $d-1=2\kappa$.}

It remains to address the case of an even dimensional interface. We aim to construct a projector $\Pin(H)$ based on $\Pin(H_I)$ that will generate a non-trivial Fredholm operator. Since $H_I$ does not have a spectral gap, the sign function cannot be used and needs to be replaced by a regularized version, as we did in the case of odd dimensional interfaces. This is done introducing a unitary and a projector\footnote{Let $F$ be Hermitian with $F^2=I$ and $U$ be unitary. Then $P=\frac I2+\frac 12 U^*FU$ is an orthogonal projector, i.e., $P^2=P=P^*$.}: 
\begin{equation}\label{eq:interfaceP}
 \Upsilon_\delta(H) = e^{i \pi \chi_\delta(H)},\qquad \Pin(H) = \dfrac I2 + \dfrac 12 \Upsilon_{\delta}^*(H) \gamma_0 \Upsilon_{\delta}(H).
\end{equation}
Here, $\gamma_0$ is the chiral matrix $\sigma_3\otimes I$, where $I=I_{2^\kappa}$ for $H$ the Hamiltonian defined in the whole space, and $I=I_{2^{\kappa-1}}$ for $H$ the interface Hamiltonian $-\eps H_I$.

The complicated form of the operator $\Pin$ stems from our need to construct a projector that is well defined for interface Hamiltonians $H_I$ for which both mass term $m=0$ and regularization $\eta=0$. Introducing a contribution proportional to $\gamma_0$, which emulates the role of a mass term is therefore not entirely surprising. The specific form above also comes from the bulk-edge correspondence relating bulk invariants to edge invariants in one less spatial dimension that arises in very general algebraic constructions \cite{bourne2018chern,prodan2016bulk,schulz2000simultaneous}.

We then define $\Un$ as in the bulk even dimensional case with $\tilde\Un=I\otimes \Un$ and define $\tilde \Pin=\Pin\otimes I$. Then we have:
\begin{theorem}\label{thm:interfaceeven}
	Let $\Pin$ and $\Un$  be defined as above. Then $\tilde \Pin(H) \tilde \Un\tilde \Pin(H)$ on the range of $\tilde \Pin(H)$ is Fredholm and $-\Ind{\tilde \Pin(H) \tilde \Un\tilde \Pin(H)}$ is given by $-\eps$ defined in \eqref{eq:epsnd}.
\end{theorem}

\begin{proof}
Let us start with the full space Hamiltonian $H$ and construct $\tilde \Pin \tilde\Un\tilde\Pin$ as well as $(\tilde\Pin-\tilde \Qin)^{d}$ with $\tilde \Qin=\tilde\Un \tilde \Pin \tilde \Un^*$. We observe that all involved functions of $H$ are compactly supported so that the associated kernels are smooth and sufficiently rapidly decaying. We then invoke Russo's criterion in Lemma \ref{lem:estip} (with $p=d>n=d-1$) to obtain that $\tilde \Pin-\tilde \Qin$ is in ${\cal I}_d$ so that $\tilde\Pin \tilde\Un\tilde\Pin$ is Fredholm. The proof is the same as in the even dimensional bulk setting.

It remains to calculate the index, which, as in the preceding case, amounts to replacing $H$ by $-\eps H_I$, and using geometric identities, to recasting a spatial integral by an integral in the Fourier variables. As in the even-dimensional case, we introduce
\[
  \Upsilon_a(H) = e^{i\arctan (H)},\qquad \Pin_a(H) = \dfrac I2 + \dfrac 12 \Upsilon_a^*(H) \gamma_0 \Upsilon_a(H).
\]
The operators with $\tilde \Pin$ replaced by $\tilde \Pin_a$ and any continuous transformation from $\chi_\delta$ to $\frac{1}{\pi} \arctan x + \frac12$ are still a continuous family of Fredholm operators, which therefore share the same index. The objective is now to calculate the index associated to $\Pin_a(H)$. What we observe from the same direct computation as in the preceding paragraph is that 
\[
  \Pin_a(H) = \dfrac12 + \dfrac12 \dfrac{1-H^2}{1+H^2} \gamma_0 + \dfrac12 \dfrac{2i\gamma_0 H}{1+H^2}
\]
Again, $H_I^2$ is scalar and so the above denominator makes sense as a scalar object. For $H_I=k'\cdot\gamma'$, this is
\[
  \Pin_a(H_I) =  \dfrac12  + \dfrac12 \dfrac{k'\cdot 2i\gamma_0 \gamma'}{1+|k'|^2} + \dfrac12 \dfrac{1-|k'|^2}{1+|k'|^2} \gamma_0.
\]
 We then realize that $i\gamma_0\gamma'$ is a set of matrices satisfying all the commutation relations of the original set $\Gamma'$ so that the index of interest is given by the degree of $h^R=(k',\frac{1-|k'|^2}2)$. Recalling that this is $\Pin_a(-\eps H_I)$ that we wish to compute, we obtain an index equal to $-\eps$. As in the preceding case, the definition of an adapted projector involves regularizing the Hamiltonian $H_I$ by both adding a mass $m=\frac12$ and a regularization $\eta=\frac12$. 
 \end{proof}

 Let us conclude this section with the case $H_V=H+V$ with $V$ a perturbation. Since $\Upsilon_\delta(H)$ changes sign from $H$ very negative to $H$ very positive, the application of Proposition \ref{prop:HSF} requires that $V$ satisfy both \eqref{eq:compactres} and \eqref{eq:boundres}. In that case, there is no difficulty in propagating the perturbation $V$ through the spectral calculus and get that $\Pin(H_V)-\Pin(H)$ is compact so that we have the:
 \begin{theorem}\label{thm:eveninterfaceV}
  	Let $V$ satisfy \eqref{eq:compactres} and \eqref{eq:boundres}. Then $\tilde \Pin(H_V) \tilde \Un\tilde \Pin(H_V)$ on the range of $\tilde \Pin(H_V)$ is Fredholm and its index as in the case $V\equiv0$. 
 \end{theorem}

\section*{Acknowledgments} This work was partially funded by the U.S. National Science Foundation and the U.S. Office of Naval Research.

\appendix
%
\section*{Appendix}
\section{Fredholm operators, modules, trace-class}
\label{sec:FH}

Let $H$ be a Hilbert space and $F$ a bounded linear operator on $H$. The operator is said to be Fredholm if its range is closed and its kernel and co-kernel are finite dimensional; i.e., with $F^*$ the adjoint operator to $F$, we have the direct sum $H={\rm Ran}(F) \oplus {\rm Ker}(F^*)={\rm Ran}(F^*)\oplus {\rm Ker}(F)$ with both kernels of finite dimension. The index of $F$ is then defined as
\[
  \Ind{F} = \mbox{dim(Ker($F$))} - \mbox{dim(Ker($F^*$))} \in \Zm.
\]
One of the main properties of the index is its stable, topological, nature. For $K$ a compact operator on $H$ and $F$ a Fredholm operator, then $\Ind{F+K}=\Ind{F}$. Moreover, the set of Fredholm operators is open  in the uniform topology so that a sufficiently small perturbation of a Fredholm operator yields another Fredholm operator with the same index. As a consequence, let $F_t$ for $0\leq t\leq 1$ be a continuous family of Fredholm operators. Then, $\Ind{F_t}$ is independent of $t$. A useful characterization of Fredholm operators is the Atkinson criterion, stating that $F$ is Fredholm if and only if it is invertible up to compact operators. In other words, $F$ is Fredholm if there exists a bounded operator $G$ such that $I-FG$ and $I-GF$ are compact; see \cite[Chapter 19]{H-III-SP-94}.

\paragraph{Fredholm Modules.}
The objective of this paper is to assign topological invariants to materials such as topological insulators that are robust with respect to heterogeneities. Fredholm indices are therefore natural candidates. One then needs to be able to compute such indices and the above analytic description of the index is not easy to manipulate. Finding algebraic expressions for the index from the kernel of a Fredholm operator has a long tradition, with the index theorems of Atiyah-Singer, Boutet de Montvel, Connes. One of the devices used in these derivation is the notion of Fredholm modules. While we will not delve into the algebraic properties of Fredholm modules, they serve as a motivation for the type of Fredholm operators that appear in this and many other papers on the mathematical description of topological insulators. A related notion is that of index of pairs of projection operators. We review here the relevant material.

For $H$ a Hilbert space and $A$ an algebra of bounded operators on $H$, we assume the existence of a bounded operator $F$ such that $F^2=I$ and such that $[F,f]=Ff-fF$ is compact for all $f\in A$. This defines an odd Fredholm module (used in this paper in odd dimensions). In this case, we define $\Pr=\frac12(I+F)$ and verify that this is an orthogonal projector.

An even Fredholm module (used in this paper in even dimensions) has an additional property. We assume the existence of a matrix $\gamma$, typically written as ${\rm Diag}(I,-I)$ for $I$ an identity matrix, such that $\gamma F + F\gamma=0$. This additional symmetry allows us to write $F$ as 
\[
  F = \left(\begin{matrix} 0& \Un^* \\ \Un & 0 \end{matrix} \right)
\]
where $\Un$ is a unitary operator.

In applications to topological insulators, the above operators, $F$ and $\Pr$ in odd dimension, or $F$ and $\Un$ in even dimension, take the form of matrix-valued multiplication operators in the physical domain. They are here to help display the topological properties of specific operators in the algebra $A$ by an appropriate pairing, which results in an integer-valued index. More precisely, in even dimension, we have that for each $P$ an orthogonal  projector in $A$, then $P\Un P$ is a Fredholm operator on Ran$(P)$ and its index admits an explicit integral representation. In odd dimensions, it is unitary operators $U$ in $A$ that are considered, and then $\Pr U \Pr$ is a Fredholm operator on Ran$(\Pr)$ whose index also admits an explicit integral representation. When both $F$ and $P$ or $U$ are matrix valued, the above product needs to be considered as a tensorized product such that the matrix structures of $F$ and $P$ or $U$ do not interact. In other words, $\Pr U \Pr$ really means $(I_1\otimes \Pr) (U\otimes I_2) (I_1\otimes \Pr)$ with matrices $I_{1,2}$ of appropriate dimensions. In this paper, though not in this appendix, we use the notation $\tilde\Pr$ to mean $I_1\otimes \Pr$, etc... That, for instance, $P\Un P$ is Fredholm on Ran$(P)$ comes from the assumption that $[F,P]$, and hence $[P,\Un]$, is compact. Indeed $P\Un P P \Un^* P=P\Un P\Un^* P=PP\Un\Un^*P + K=P+K$ and hence $P\Un^* P$ is an inverse of $P\Un P$ on Ran$(P)$ up to compact and the result follows.

Such pairings of the form $(P,\Un)$ or $(\Pr,U)$ are then stable with respect to continuous perturbations of $P$ or $U$, which is what we will use in this paper, as well as continuous perturbations of $\Un$ or $\Pr$ via $F$, which is the main tool used in (more) algebraic \cite{bourne2018chern,prodan2016bulk} or (more) analytic \cite{elgart2005equality,graf2007aspects,Graf2013} general descriptions of the bulk-boundary correspondence. The latter results are not used in this paper since we compute bulk and interface indices explicitly and thus do not need to relate them implicitly.

The operator $F$ defining the module is typically of the form $F=\sgn{D}=\frac{D}{|D|}$, where $D$ is a Dirac-type operator. The main difference here with most uses of the notion of Dirac operators is that they are matrix-valued multiplications by the coordinates of the spatial domain. In particular, on an appropriate domain of definition, $D=x\cdot\Gamma$, where $\Gamma$ is an appropriate set of matrices (see Appendix \ref{sec:CA}). The reason for this is that topological insulators are naturally described topologically in the Fourier domain, where the Dirac operators take the more familiar form of first-order differential operators. In the physical domain, where we wish to handle asymmetric spatially varying perturbations, the Dirac operator thus takes the form of multipliers. As a coincidence, the low-energy characterization of topological insulators involves objects that are linear in the Fourier variables (in the Brillouin zone), hence first-order differential operators in the physical domain that also take the form of Dirac operators. We thus have two dual types of Dirac operators; that of the Fredholm module that is differential in the Fourier variables and multiplicative in the physical variables; and that of the low-energy description of the topological insulators, which is differential in the physical variables and multiplicative in the dual Fourier variables. These two types of `Dirac' operators should not be confused.

\paragraph{Trace-class operators.} The above pairing indicates which type of Fredholm operator, $P\Un P$ or $\Pr U\Pr$ should be considered. The theory of Fredholm modules also comes with an explicit recipe to calculate the index as the explicit trace of an appropriate trace-class operator. 

We now recall relevant material on spaces of compact operators. The ideal of compact operators (when $K$ is compact, then so are $BK$ and $KB$ for $B$ bounded) admits smaller ideals characterized by the decay of their singular values. Let $K$ be a compact operator on $H$ and let $\lambda_n$ be its singular values, i.e., the eigenvalues of the operator $(K^*K)^{\frac12}$. Then, we define the Schatten ideal ${\cal I}_n$ as the space of compact operators such that $\sum_j |\lambda_j|^n<\infty$, counting multiplicities. This defines a Schatten norm $\|K\|_{{\cal I}_n} = (\sum_j |\lambda_j|^n)^{\frac 1n}$. Classical properties of the Schatten ideals is that these are indeed ideals (for $K\in{\cal I}_n$ and $B$ bounded, then both $KB$ and $BK$ belong to ${\cal I}_n$) and that they are Banach (complete) spaces for the Schatten norm. While $n\geq1$ is an integer, the above expressions more generally hold for any real number $n>0$. Operators in ${\cal I}_1$ are called trace-class. Operators in ${\cal I}_2$ are called Hilbert-Schmidt.

Let $K$ be a square matrix-valued compact operator with entries $(K_{ij})_{1\leq i,j\leq n}$ compact operators as well. Then we can write $K=\sum_{i,j} K_{ij} \mathds{1}_{ij}$ with $\mathds{1}_{ij}$ the matrix with entries equal to $0$ except at position $(i,j)$. This shows that $\|K\|_{{\cal I}_n}\leq C \sup_{i,j} \|K_{ij}\|_{{\cal I}_n}$. The matrix-valued operators are therefore analyzed for each entry as a sufficient condition to be in the appropriate Schatten class.

From the definition of Schatten classes, we verify that for two operators in classes $n$ and $m$, then the product is in the class labeled by $n+m$. In particular, when $T\in{\cal I}_N$, then $T^N$ is trace-class.

\paragraph{Fedosov formula.} Indices of Fredholm operators admit expressions as traces of appropriate operators. Such a relation is found in what is referred to as the Fedosov formula; see \cite[Appendix]{AVRON1994220} and \cite[Prop. 19.1.14]{H-III-SP-94}. We state it as:
\begin{proposition}\label{prop:fedosov}
  Let $F$ and $G$ be such that $R_1=I-GF$ and $R_2=I-FG$ are compact.  
Then $F$ and $G$ are Fredholm operators. Let us assume moreover that $R_1^n$ and $R_2^n$ are trace-class, then we have that
\[
  \Ind{F} = \Tr\ R_1^n - \Tr\ R_2^n = - \Ind{G}.
\]
Let $P$ be a projector and $U$ a unitary operator on a Hilbert space $H$.  Assume that $[P,U]$ is compact, or equivalently that $[P,U]U^*=U[P,U^*]$ or $[P,U^*]$ is compact. Then $PUP$  and $PU^*P$ are 
Fredholm
on Ran$(P)$, applying the above to $F=PUP$ and $G=PU^*P$ with $R_1=P-PU^*PUP$ and $R_2=P-PUPU^*P$  (since $P_{|{\rm Ran}(P)}=I_{|{\rm Ran}(P)}$). Assuming $R_1^n$ and $R_2^n$ are trace-class, then
\[
 -\Ind{PU^*P_{|{\rm Ran}(P)}} =  \Ind{PUP_{|{\rm Ran}(P)}} = \Tr R_1^n -\Tr R_2^n.
\]
\end{proposition}

The above result is applied several times in this paper as follows. Let $\kappa\in\Nm^*$ and set $d=2\kappa$ in even dimension and $d=2\kappa+1$ in odd dimension. Let $P$ be a projection, $U$ a unitary, and $Q=UPU^*$. It is shown in \cite{avron1994} that when $(P-Q)^{2\kappa+1}$ is trace-class, then so is $(P-Q)^{2\kappa+3}$ and moreover,
\[
  \Tr \ (P-Q)^{2\kappa+3} = \Tr (P-PU P U^*P)^{\kappa+1} - \Tr ( P-PU^* PUP)^{\kappa+1}.
\]  
We thus apply the above Fedosov formula with $n=\kappa+1$ and deduce that 
\[
  -\Ind{PUP_{|{\rm Ran}(P)}} = \Tr\ (P-Q)^{2\kappa+1}  = \Tr \ U\big( [P,U^*][P,U]\big)^\kappa [P,U^*].
\]
We summarize the above results as the following
\begin{proposition}\label{prop:trace}
  Let $d=2\kappa$, $P$ be a projection and $\Un$ the unitary of the Fredholm module. Let $Q=\Un P \Un^*$. Let equivalently $[P,\Un]$ or $P-Q=[P,\Un]\Un^*=\Un[\Un^*,P]$ or $[P,\Un^*]$ be of class ${\cal I}_{2\kappa+1}$. Then, $P\Un P$ is Fredholm and
  \[
  -\Ind{P \Un P} = \Tr\ (P-Q)^{2\kappa+1}  = \Tr \ \Un \big( [P,\Un^*][P,\Un]\big)^\kappa [P,\Un^*].
  \]
  Let $d=2\kappa+1$, $\Pr$ the projection of the Fredholm module and $U$ a unitary operator. Let $\Qr=U\Pr U^*$.  Let equivalently $[\Pr,U]$ or $\Pr-\Qr=[\Pr,U]U^*=U[U^*,\Pr]$ or $[\Pr,U^*]$ be of class ${\cal I}_{2\kappa+1}$. Then, $\Pr U \Pr$ is Fredholm and
  \[
  -\Ind{\Pr U \Pr} = \Tr\ (\Pr-\Qr)^{2\kappa+1}  = \Tr \ U \big( [\Pr,U^*][\Pr,U]\big)^\kappa [\Pr,U^*].
  \]
\end{proposition}
That the two traces above agree is based on straightforward calculations.  That the traces equal the index is based on the above Fedosov formula.

The main technical component of the derivation is therefore to show that $(P-Q)^{2\kappa+1}$ is indeed a trace-class operator, and as a sufficient condition that $P-Q\in{\cal I}_{2\kappa+1}$. We recall the results we need in the paper. 

We now introduce the main tool used here to show that an operator belongs to ${\cal I}_n$ for $n>2$. It is based on a result by Russo \cite{russo1977hausdorff} with a modification provided by Goffeng to remove a square-integrability assumption. The results by Goffeng \cite[Theorems 2.3\&2.4]{GOFFENG2012357} are as follows. 
\begin{lemma}\label{lem:russo}
Assume $T$ a bounded operator on $H=L^2(\Rm^d)$ with integral kernel $t(x,y)$. Let $t^*(x,y)$ be the kernel of the adjoint operator $T^*$. Let $p>2$ and $q=\frac{p}{p-1}$ the conjugate index ($\frac1p+\frac1q=1$). Assume that $\|t\|_{q,p}$ and $\|t^*\|_{q,p}$ are bounded, where
\[
   \|t\|_{q,p} = \Big(\dint_{\Rm^d} \Big(\dint_{\Rm^d} |t(x,y)|^q dx\Big)^{\frac pq} \Big)^{\frac1p}.
\]
Then $T$ belongs to ${\cal I}_p$ and $\|T\|_{{\cal I}_p} \leq (\|t\|_{q,p} \|t^*\|_{q,p})^{\frac12}$.
\end{lemma}
\begin{corollary}\label{cor:trace}
  When $T$, or more generally, a matrix-valued operator $K$ belongs to ${\cal I}_n$ for $n\geq2$, then $T^n$ or $K^n$ are trace-class (belong to ${\cal I}_1$) and their trace is given by the integral of their kernel along the diagonal. This means
  \[
   \Tr K^n  = \dint_{\Rm^d} \dsum_j (K^n)_{jj}(x,x) dx.
  \]
\end{corollary}

The latter result shows that an operator $K^n$ is trace-class for $n\geq2$. The corollary can also often be obtained from Brislawn's result \cite{brislawn1988kernels} in the lemma below. 

When $n=1$, none of the above results applies. First, bounds on and continuity of the kernel of $K$ are not sufficient to show that $K$ is trace-class, as shown by a counter-example by Carleman; see e.g., the recent paper \cite{delgado2017schatten} on the topic. In such a setting, we have to prove that $K$ is trace-class `manually', as we do in section \ref{sec:2d} using an approximation trick that may be found for instance in \cite{LAX-JW-02} and \cite{GK-AMS-69}. Moreover, once we know that $K$ is trace-class, we still need to show that its trace is given as the integral of an appropriate kernel. There, we apply Brislawn's result \cite{brislawn1988kernels} (see also \cite{GK-AMS-69}), which we now recall.

Let $k(x,y)$ be the kernel of $K$ acting on $\Rm^d$ and when the limit exists, define
\[
 \tilde k(x,y) = \lim\limits_{r\to0} \dfrac{1}{|C_r|^2} \dint_{C_r\times C_r} h(x+s,y+t)dsdt
\]
where $C_r=[-r,r]^d$ and $|C_r|=(2r)^n$. Then we have \cite[Theorem 3.1]{brislawn1988kernels}
\begin{lemma}\label{lem:trace}
 Let $K$ be a trace-class operator on $L^2(\Rm^d)$. Then $\tilde k(x,x)$ exists almost everywhere and $\Tr\ K=\int_{\Rm^d} \tilde k(x,x) dx$.
\end{lemma}
As a corollary, when $\tilde k(x,x)$ equals $k(x,x)$ almost everywhere (w.r.t. the Lebesgue measure on the diagonal), then the trace of $K$ is given by the integral of $k(x,x)$ along the diagonal. This is the setting of the operators in the current paper.

This concludes the tour of the tools we need to define and compute as integrals the topological invariants we wish to assign to topological insulators.

\section{Conventions on Clifford algebra representations}
\label{sec:CA}

Let $d=2\kappa$ or $d=2\kappa-1$ for $\kappa\geq1$. For $\kappa=1$, the matrices $\gamma_2^{1,2,3}$ are $\sigma_{1,2,3}$, respectively. We verify that $\sigma_3=(-i)^\kappa\sigma_1\sigma_2$ for $\kappa=1$ and $\sigma_i\sigma_j+\sigma_j\sigma_i=2\delta_{ij}$. These matrices are used in dimension $d=2$ and $d=1$. We need to generalize this to higher dimensions. A standard construction is as follows. 

Let $d=2\kappa$ and assume constructed the (square) matrices $\gamma_d^j$ of size $2^\kappa$ for $1\leq j\leq d+1$ with 
\begin{equation}\label{eq:constCA}
\gamma_d^{d+1}=(-i)^\kappa\gamma_d^1\ldots\gamma_d^d,\qquad \gamma_d^j\gamma_d^k+\gamma_d^k\gamma_d^j=2\delta_{jk}\quad \mbox{ for all } \quad1\leq j,k\leq d+1.
\end{equation}
The latter properties are all the structure that is used in the derivation of the results in this paper. 

The constructions for the dimensions $2\kappa+2$ and $2\kappa+1$ are the same with
\[
  \gamma_{d+2}^j=\gamma_{d+1}^j=\sigma_1\otimes \gamma_d^j,\quad 1\leq j\leq d+1,\qquad \gamma_{d+2}^{d+2}=\gamma_{d+1}^{d+2}=\sigma_2\otimes I_{2^\kappa}
\]
followed by $\gamma_{d+2}^{d+3}=(-i)^{\kappa+1}\gamma_{d+2}^1\ldots\gamma_{d+2}^{d+2} = \sigma_3\otimes I_{2^\kappa}$. The latter matrix is the chiral symmetry matrix $\gamma_{d+1}^{d+3}\equiv \gamma_{d+1}^0$ in odd dimension. So, in even dimension, \eqref{eq:constCA} holds while in odd dimension $d=2\kappa-1$, it is replaced by the similar constraint
\begin{equation}\label{eq:constCAodd}
\gamma_d^{d+2}\equiv \gamma_d^{0}=(-i)^\kappa\gamma_d^1\ldots\gamma_d^{d+1},\qquad \gamma_d^j\gamma_d^k+\gamma_d^k\gamma_d^j=2\delta_{jk}\quad \mbox{ for all } \quad1\leq j,k\leq d+1.
\end{equation}
As a convenient notation, we introduce the following vectors of $\gamma$ matrices
\[
  \Gamma_{d} = (\gamma_d^1,\ldots,\gamma_d^{d+1}),  \qquad \gamma_d=(\gamma_d^1,\ldots,\gamma_d^{d})
\]
both for even and odd dimensions $d$.

The bulk Hamiltonian may then be written as
\[
  H= h_{d} \cdot \Gamma_d = \dfrac1i\nabla\cdot\gamma_d+ m_\eta \gamma_{d}^{d+1},\qquad h_d=(\frac1i \partial_1 , \ldots, \frac1i \partial_d, m_\eta),\quad m_\eta = m+\eta\Delta.
\]
This expression applies to both even and odd dimensions $d$. When $d=2\kappa+1$ is odd, then the above Hamiltonian anti-commutes with the chiral matrix $\gamma_d^{d+2}\equiv \gamma_d^0$. To simplify notation, we also represent the latter matrix as $\gamma_0$.

The above representation for the Clifford algebra is convenient for the bulk Hamiltonian but less so for the interface Hamiltonian. The reason is that we want the (annihilation-type) operator $\fa=\partial_{x_d}+m_\eta(x_d)$ to appear naturally in the spectral analysis of the Hamiltonian. This is done by changing the $\gamma$ matrices used in the representation of the Hamiltonian while not changing their commutation and multiplication properties. 

The changes of matrices we need to perform are different in even and odd dimensions. In even dimensions, we replace $(\sigma_1,\sigma_2,\sigma_3)$ by $(\sigma_3,-\sigma_2,\sigma_1)$ in the constructions of $\gamma_{d+2}$ from the original construction of $\gamma_d$. We observe that \eqref{eq:constCA} is satisfied with $\kappa$ replaced by $\kappa+1$. In this system, the interface Hamiltonian is represented as 
\[
  H[m(x_d)] = \dfrac1i \nabla'\cdot \sigma_3\otimes\Gamma_{d-2}  - \dfrac1i \partial_{x_d} \sigma_2\otimes I + m(x_d) \sigma_1\otimes I.
\]
Here $\nabla'$ corresponds to the usual gradient operator in the first $d-1$ variables. 
In the partial Fourier variables $k'$ with respect to the first $d-1$ variables, this is 
\[
  \hat H[k',m(x_d)] = k'\cdot\sigma_3\otimes\Gamma_{d-2} + \fa \sigma_+\otimes I + \fa^* \sigma_-\otimes I = \left(\begin{matrix} k'\cdot\Gamma_{d-2} & \fa \\ \fa^* &-k'\cdot\Gamma_{d-2}\end{matrix}\right).
\]
Here, the block matrices are of size $2^{\kappa-1}$ for $d=2\kappa$.

In odd dimension, we also wish to have $\partial_{x_d}$ and $m_\eta$ appear in the same matrix entries. This is achieved as follows.  Let $d=2\kappa+1$. The first $d-1$ matrices $\sigma_1\otimes \gamma_{d-1}^j$ are replaced by $\sigma_3\otimes \gamma_{d-1}^j$ as in even dimensions. The last matrix $\sigma_2\otimes I$ is replaced by $-\sigma_2\otimes I$ as well. The matrix $\gamma_{d}^d$ corresponding to differentiation with respect to the last variable is replaced by the matrix $\sigma_1\otimes I$. The chirality matrix $\gamma^0_d$ is then replaced by whatever product is necessary so that \eqref{eq:constCAodd} holds. We verify that up to a constant $i^\kappa$, the chirality matrix is given by $\sigma_3\otimes \gamma^0_{d-2}$.  The interface Hamiltonian then takes a very similar form to that in even dimension:
\[
  H[m(x_d)] = \dfrac1i \nabla'\cdot \sigma_3\otimes\Gamma_{d-1}  - \dfrac1i \partial_{x_d} \sigma_2\otimes I + m(x_d) \sigma_1\otimes I.
\]
In the partial Fourier variables $k'$, this is 
\[
  \hat H[k',m(x_d)] = k'\cdot\sigma_3\otimes\Gamma_{d-1} + \fa \sigma_+\otimes I + \fa^* \sigma_-\otimes I = \left(\begin{matrix} k'\cdot\Gamma_{d-1} & \fa \\ \fa^* &-k'\cdot\Gamma_{d-1}\end{matrix}\right).
\]
Here, the block matrices are of size $2^{\kappa-1}$ for $d=2\kappa-1$.

The spectral analysis of the interface Hamiltonian is therefore the same in even and odd dimensions and corresponds to the analysis of the operator $\fa$.

\section{Some geometric identities}
\label{sec:connes}

\begin{lemma}[Geometric Identities.]
 \label{lem:geometric}
 Let $\Un(x)=\hat x\cdot\gamma$ with $\hat x=\frac{x}{|x|}$.
 In the even dimension case $d=2\kappa$, we have the geometric identity
 \[
 	\dint_{\Rm^d} \tr \ \Un(x)(\Un^*(x)-\Un^*(x+y_d))\ldots (\Un^*(x+y_1)-\Un^*(x)) dx = \dfrac{(2\pi i)^{\frac d2}}{\frac d2!} {\rm Det}(y_d,\ldots,y_1).
\]
In the odd dimensional case $d=2\kappa+1$, we have
\[
 \dint_{\Rm^d} \tr\ (\Un(y_{d}+x)-\Un(y_{d-1}+x)) \ldots  (\Un(y_1+x)-\Un(x)) dx = \dfrac{2^d}{d!!} (i\pi)^{\frac{d-1}2} {\rm Det}(y_d,\ldots,y_1).
\]
\end{lemma}
\begin{proof} The proof is adapted from the results in \cite{prodan2016bulk}.
 We recall that $\Un(x)=\frac{x\cdot\check\gamma}{|x|}$. It is more convenient to use the notation $F(x)=\frac{x\cdot\gamma}{|x|}=\hat x\cdot\gamma$ and realize that $F$ anticommutes with $\gamma_0\equiv\gamma^0_d$ and admits a matrix representation with $\Un(x)$ as lower left entry and $\Un^*(x)$ as upper right entry, with $0$ on the diagonal. We then observe that a product of an even number of matrices $F$ is a block diagonal matrix and that moreover
 \[
   \tr \ \Un_1 \Un_2^* \ldots \Un_{2\kappa-1}\Un_{2\kappa}^* = \tr \ \dfrac{I-\gamma_0}2 F_1\ldots F_{2\kappa}.
 \]
 This term appears in the lower right diagonal entry of the block matrix.
 The above calculation involves $2\kappa=d+2$ terms and so, defining $F_1=F(x)$ and $F_{j}=F(x+y_j)$, $1\leq j\leq d$ , we want to calculate
 \[
 	\tr\ \dfrac{I-\gamma_0}2 F_1(F_1-F_{d+1})(F_{d+1}-F_d) \ldots (F_2-F_1).
 \]
  The component involving identity vanishes. The reason is that most terms involve the product of at most $d$ terms by cancellations $F_j^2=I$ and so are trace-less. It remains a term $F_{d+1}F_d\ldots F_2F_1$ involving $d+1$ terms but it anticommutes with $\gamma_0$ and hence is trace-less. Indeed, $\tr A = \tr A\gamma^2=-\tr \gamma A \gamma = -\tr A\gamma^2 = -\tr A=0$ by cyclicity of the trace. 
  
  Thus, only the contribution with $\gamma_0$ does not vanish. Now we use the fact that
  \[
  	\tr \gamma_0 F_1(F_1-F_{d+1})\ldots(F_2-F_1) = 2 \tr \gamma_0 \dsum_{j=1}^{d+1} (-1)^{j-1} (F_{d+1} \ldots \hat F_j \ldots F_1),
  \]
  where $\hat F_j$ means that this term does not appear in the product. Indeed, in the above product, we observe that we can have at most one cancellation $F_j^2=I$ for otherwise, only $d-2$ terms remain and the trace product with $\gamma_0$ vanishes. Using $F_1^2=I$, we have a term $\gamma_0(F_{d+1}-F_d)\ldots(F_2-F_1)$, where no cancellations are allowed, and which gives one such contribution. The second term, using $F_1\gamma_0+\gamma_0 F_1=0$  is $\gamma_0 F_{d+1}(F_{d+1}-F_d)\ldots(F_2-F_1)F_1$. The term $F_{d+1}^2=I$ gives a contribution $\gamma_0 F_{d}\ldots F_1$, the term $F_d^2=I$ gives $\gamma_0 F_{d+1}(-I)F_{d-1}\ldots F_1$ and so on.
  
  We thus obtain that our object of interest is given by
  \[
    -\dint_{\Rm^d} \tr \gamma_0 \dsum_{j=1}^{d+1} (-1)^{j-1} (F(x+y_{d+1}) \ldots \hat F(x+y_j) \ldots F(x+y_1)) dx,
  \]
  where $y_{d+1}\equiv0$. Now, we observe that in even dimension,
  \[
  		\tr \gamma_0 \prod_{j=1}^d  z_j \cdot \gamma = (2i)^{\frac d2} {\rm Det}(z_1,\ldots z_d).
  \]
  This comes from the property that $\tr (\gamma_0 \gamma_{\rho_1}\ldots \gamma_{\rho_d}) = (2i)^{\frac d2}(-1)^\rho$ for $\rho$ a permutation with signature $(-1)^\rho$.
  
  The striking result, which appears in the proof of \cite[Lemma 6.4.1]{prodan2016bulk} and is equivalent to that statement there, is that then
  \[
       - \dint_{\Rm^d} \tr \gamma_0 \dsum_{j=1}^{d+1} (-1)^{j-1} (F(x+y_{d+1}) \ldots \hat F(x+y_j) \ldots F(x+y_1)) dx = \dfrac{(2\pi i)^{\frac d2}}{\frac d2!} {\rm Det}(y_d,\ldots,y_1).
  \]
  This is with $x=-z-y_1$
  \[
  - \dint_{\Rm^d} \tr \gamma_0 \dsum_{j=1}^{d+1} (-1)^{j-1} (F(y_{d+1}-y_1-x) \ldots \hat F(y_j-y_1-z) \ldots F(-z)) dz,
  \]
  which is $\frac{(2\pi i)^{\frac d2}}{\frac d2!}$ times ${\rm Det}(-y_1,y_d-y_1,\ldots,y_2-y_1)={\rm Det}(y_d,\ldots,y_2,y_1)$.   This proves the result in even dimension. A similar modification to the proof in \cite{prodan2016bulk} applies to the odd-dimensional case as well.
\end{proof}

%
%
%
%

\begin{thebibliography}{10}

\bibitem{avron1989chern}
{\sc J.~Avron, L.~Sadun, J.~Segert, and B.~Simon}, {\em {Chern numbers,
  quaternions, and Berry's phases in Fermi systems}}, Communications in
  mathematical physics, 124 (1989), pp.~595--627.

\bibitem{avron1994}
{\sc J.~E. Avron, R.~Seiler, and B.~Simon}, {\em Charge deficiency, charge
  transport and comparison of dimensions}, Comm. Math. Phys., 159 (1994),
  pp.~399--422.

\bibitem{AVRON1994220}
\leavevmode\vrule height 2pt depth -1.6pt width 23pt, {\em The index of a pair
  of projections}, Journal of Functional Analysis, 120 (1994), pp.~220 -- 237.

\bibitem{bal2019topological}
{\sc G.~Bal}, {\em Topological
  invariants for interface modes}, Communications in Partial Differential
  Equations, 47 (2022), pp.~1636--1679.

\bibitem{B-EdgeStates-2018}
\leavevmode\vrule height 2pt depth -1.6pt width 23pt, {\em Topological
  protection of perturbed edge states}, Communications in Mathematical
  Sciences, 17 (2019), pp.~193--225.

\bibitem{bellissard1986k}
{\sc J.~Bellissard}, {\em {K-theory of C*-Algebras in solid state physics}}, in
  Statistical mechanics and field theory: mathematical aspects, Springer, 1986,
  pp.~99--156.

\bibitem{bellissard1994noncommutative}
{\sc J.~Bellissard, A.~van Elst, and H.~Schulz-Baldes}, {\em {The
  noncommutative geometry of the quantum Hall effect}}, Journal of Mathematical
  Physics, 35 (1994), pp.~5373--5451.

\bibitem{bernevig2013topological}
{\sc B.~A. Bernevig and T.~L. Hughes}, {\em Topological insulators and
  topological superconductors}, Princeton university press, 2013.

\bibitem{Bernevig1757}
{\sc B.~A. Bernevig, T.~L. Hughes, and S.-C. Zhang}, {\em {Quantum Spin Hall
  Effect and Topological Phase Transition in HgTe Quantum Wells}}, Science, 314
  (2006), pp.~1757--1761.

\bibitem{bourne2017non}
{\sc C.~Bourne and E.~Prodan}, {\em {Non-Commutative Chern Numbers for Generic
  Aperiodic Discrete Systems}}, arXiv preprint arXiv:1712.04136,  (2017).

\bibitem{bourne2018chern}
{\sc C.~Bourne and A.~Rennie}, {\em Chern numbers, localisation and the
  bulk-edge correspondence for continuous models of topological phases},
  Mathematical Physics, Analysis and Geometry, 21 (2018), p.~16.

\bibitem{brislawn1988kernels}
{\sc C.~Brislawn}, {\em Kernels of trace class operators}, Proceedings of the
  American Mathematical Society, 104 (1988), pp.~1181--1190.

\bibitem{chiu2016classification}
{\sc C.-K. Chiu, J.~C. Teo, A.~P. Schnyder, and S.~Ryu}, {\em Classification of
  topological quantum matter with symmetries}, Reviews of Modern Physics, 88
  (2016), p.~035005.

\bibitem{cycon2009schrodinger}
{\sc H.~L. Cycon, R.~G. Froese, W.~Kirsch, and B.~Simon}, {\em Schr{\"o}dinger
  operators: With application to quantum mechanics and global geometry},
  Springer, 2009.

\bibitem{davies_1995}
{\sc E.~B. Davies}, {\em Spectral Theory and Differential Operators}, Cambridge
  Studies in Advanced Mathematics, Cambridge University Press, 1995.

\bibitem{delgado2017schatten}
{\sc J.~Delgado and M.~Ruzhansky}, {\em {Schatten-von Neumann classes of
  integral operators}}, arXiv preprint arXiv:1709.06446,  (2017).

\bibitem{DFN-SP-1985}
{\sc B.~A. Dubrovin, A.~T. Fomenko, and S.~P. Novikov}, {\em Modern
  geometry---methods and applications. {P}art {II}: The Geometry and Topology
  of Manifolds}, Springer-Verlag, New York, 1985.

\bibitem{elbau2002equality}
{\sc P.~Elbau and G.~Graf}, {\em Equality of bulk and edge hall conductance
  revisited}, Communications in mathematical physics, 229 (2002), pp.~415--432.

\bibitem{elgart2005equality}
{\sc A.~Elgart, G.~Graf, and J.~Schenker}, {\em Equality of the bulk and edge
  hall conductances in a mobility gap}, Communications in mathematical physics,
  259 (2005), pp.~185--221.

\bibitem{fang2012realizing}
{\sc K.~Fang, Z.~Yu, and S.~Fan}, {\em Realizing effective magnetic field for
  photons by controlling the phase of dynamic modulation}, Nature photonics, 6
  (2012), pp.~782--787.

\bibitem{Fefferman2016}
{\sc C.~L. Fefferman, J.~P. Lee-Thorp, and M.~I. Weinstein}, {\em {Edge States
  in Honeycomb Structures}}, Annals of PDE, 2 (2016), p.~12.

\bibitem{Fleury2014}
{\sc R.~Fleury, D.~L. Sounas, C.~F. Sieck, M.~R. Haberman, and A.~Alu}, {\em
  Sound isolation and giant linear nonreciprocity in a compact acoustic
  circulator}, Science, 343 (2014), pp.~516--519.

\bibitem{Fruchart2013779}
{\sc M.~Fruchart and D.~Carpentier}, {\em An introduction to topological
  insulators}, Comptes Rendus Physique, 14 (2013), pp.~779 -- 815.

\bibitem{PhysRevB.76.045302}
{\sc L.~Fu and C.~L. Kane}, {\em Topological insulators with inversion
  symmetry}, Phys. Rev. B, 76 (2007), p.~045302.

\bibitem{PhysRevLett.98.106803}
{\sc L.~Fu, C.~L. Kane, and E.~J. Mele}, {\em Topological insulators in three
  dimensions}, Phys. Rev. Lett., 98 (2007), p.~106803.

\bibitem{GOFFENG2012357}
{\sc M.~Goffeng}, {\em Analytic formulas for the topological degree of
  non-smooth mappings: The odd-dimensional case}, Advances in Mathematics, 231
  (2012), pp.~357 -- 377.

\bibitem{GK-AMS-69}
{\sc I.~C. Gohberg and M.~G. Krein}, {\em Introduction to the theory of linear
  nonselfadjoint operators}, Translated from the Russian by A. Feinstein.
  Translations of Mathematical Monographs, Vol. 18, American Mathematical
  Society, Providence, R.I., 1969.

\bibitem{graf2007aspects}
{\sc G.~M. Graf}, {\em Aspects of the integer quantum hall effect}, in
  Proceedings of Symposia in Pure Mathematics, vol.~76, Providence, RI;
  American Mathematical Society; 1998, 2007, p.~429.

\bibitem{Graf2013}
{\sc G.~M. Graf and M.~Porta}, {\em Bulk-edge correspondence for
  two-dimensional topological insulators}, Communications in Mathematical
  Physics, 324 (2013), pp.~851--895.

\bibitem{hafezi2011robust}
{\sc M.~Hafezi, E.~A. Demler, M.~D. Lukin, and J.~M. Taylor}, {\em Robust
  optical delay lines with topological protection}, Nature Physics, 7 (2011),
  pp.~907--912.

\bibitem{hafezi2013imaging}
{\sc M.~Hafezi, S.~Mittal, J.~Fan, A.~Migdall, and J.~Taylor}, {\em Imaging
  topological edge states in silicon photonics}, Nature Photonics, 7 (2013),
  pp.~1001--1005.

\bibitem{PhysRevLett.61.2015}
{\sc F.~D.~M. Haldane}, {\em {Model for a Quantum Hall Effect without Landau
  Levels: Condensed-Matter Realization of the "Parity Anomaly"}}, Phys. Rev.
  Lett., 61 (1988), pp.~2015--2018.

\bibitem{PhysRevLett.100.013904}
{\sc F.~D.~M. Haldane and S.~Raghu}, {\em Possible realization of directional
  optical waveguides in photonic crystals with broken time-reversal symmetry},
  Phys. Rev. Lett., 100 (2008), p.~013904.

\bibitem{RevModPhys.82.3045}
{\sc M.~Z. Hasan and C.~L. Kane}, {\em Colloquium: Topological insulators},
  Rev. Mod. Phys., 82 (2010), pp.~3045--3067.

\bibitem{hasan2011three}
{\sc M.~Z. Hasan and J.~E. Moore}, {\em Three-dimensional topological
  insulators}, Annu. Rev. Condens. Matter Phys., 2 (2011), pp.~55--78.

\bibitem{H-III-SP-94}
{\sc L.~V. H{\"o}rmander}, {\em {The Analysis of Linear Partial Differential
  Operators III: Pseudo-Differential Operators}}, Springer Verlag, 1994.

\bibitem{PhysRevLett.95.146802}
{\sc C.~L. Kane and E.~J. Mele}, {\em {${Z}_{2}$ Topological Order and the
  Quantum Spin Hall Effect}}, Phys. Rev. Lett., 95 (2005), p.~146802.

\bibitem{kellendonk2004quantization}
{\sc J.~Kellendonk and H.~Schulz-Baldes}, {\em Quantization of edge currents
  for continuous magnetic operators}, Journal of Functional Analysis, 209
  (2004), pp.~388--413.

\bibitem{khanikaev2013photonic}
{\sc A.~B. Khanikaev, S.~H. Mousavi, W.-K. Tse, M.~Kargarian, A.~H. MacDonald,
  and G.~Shvets}, {\em Photonic topological insulators}, Nature materials, 12
  (2013), pp.~233--239.

\bibitem{kotani2014quantization}
{\sc M.~Kotani, H.~Schulz-Baldes, and C.~Villegas-Blas}, {\em Quantization of
  interface currents}, Journal of Mathematical Physics, 55 (2014), p.~121901.

\bibitem{LAX-JW-02}
{\sc P.~D. Lax}, {\em Functional analysis}, Pure and Applied Mathematics (New
  York), Wiley-Interscience [John Wiley \& Sons], New York, 2002.

\bibitem{lu2010massive}
{\sc H.-Z. Lu, W.-Y. Shan, W.~Yao, Q.~Niu, and S.-Q. Shen}, {\em {Massive Dirac
  fermions and spin physics in an ultrathin film of topological insulator}},
  Physical review B, 81 (2010), p.~115407.

\bibitem{lu2014topological}
{\sc L.~Lu, J.~D. Joannopoulos, and M.~Solja{\v{c}}i{\'c}}, {\em Topological
  photonics}, Nature Photonics, 8 (2014), pp.~821--829.

\bibitem{PhysRevB.83.125109}
{\sc R.~S.~K. Mong and V.~Shivamoggi}, {\em Edge states and the bulk-boundary
  correspondence in dirac hamiltonians}, Phys. Rev. B, 83 (2011), p.~125109.

\bibitem{nash2015topological}
{\sc L.~M. Nash, D.~Kleckner, A.~Read, V.~Vitelli, A.~M. Turner, and W.~T.
  Irvine}, {\em Topological mechanics of gyroscopic metamaterials}, Proceedings
  of the National Academy of Sciences, 112 (2015), pp.~14495--14500.

\bibitem{prodan2016bulk}
{\sc E.~Prodan and H.~Schulz-Baldes}, {\em {Bulk and boundary invariants for
  complex topological insulators: From K-Theory to Physics}}, Springer Verlag,
  Berlin, 2016.

\bibitem{RevModPhys.83.1057}
{\sc X.-L. Qi and S.-C. Zhang}, {\em Topological insulators and
  superconductors}, Rev. Mod. Phys., 83 (2011), pp.~1057--1110.

\bibitem{PhysRevA.78.033834}
{\sc S.~Raghu and F.~D.~M. Haldane}, {\em {Analogs of quantum-Hall-effect edge
  states in photonic crystals}}, Phys. Rev. A, 78 (2008), p.~033834.

\bibitem{rechtsman2013photonic}
{\sc M.~C. Rechtsman, J.~M. Zeuner, Y.~Plotnik, Y.~Lumer, D.~Podolsky,
  F.~Dreisow, S.~Nolte, M.~Segev, and A.~Szameit}, {\em {Photonic Floquet
  topological insulators}}, Nature, 496 (2013), pp.~196--200.

\bibitem{russo1977hausdorff}
{\sc B.~Russo}, {\em {On the Hausdorff-Young theorem for integral operators}},
  Pacific Journal of Mathematics, 68 (1977), pp.~241--253.

\bibitem{schulz2000simultaneous}
{\sc H.~Schulz-Baldes, J.~Kellendonk, and T.~Richter}, {\em Simultaneous
  quantization of edge and bulk hall conductivity}, Journal of Physics A:
  Mathematical and General, 33 (2000), p.~L27.

\bibitem{shen2011topological}
{\sc S.-Q. Shen, W.-Y. Shan, and H.-Z. Lu}, {\em {Topological insulator and the
  Dirac equation}}, in Spin, vol.~1, World Scientific, 2011, pp.~33--44.

\bibitem{simon2015operator}
{\sc B.~Simon}, {\em {Operator Theory. A Comprehensive Course in Analysis, Part
  4}}, American Mathematical Society, 92 (2015).

\bibitem{slobozhanyuk2017three}
{\sc A.~Slobozhanyuk, S.~H. Mousavi, X.~Ni, D.~Smirnova, Y.~S. Kivshar, and
  A.~B. Khanikaev}, {\em Three-dimensional all-dielectric photonic topological
  insulator}, Nature Photonics, 11 (2017), p.~130.

\bibitem{thaller2013dirac}
{\sc B.~Thaller}, {\em {The Dirac equation}}, Springer Science \& Business
  Media, 2013.

\bibitem{wang2009observation}
{\sc Z.~Wang, Y.~Chong, J.~Joannopoulos, and M.~Solja{\v{c}}i{\'c}}, {\em
  Observation of unidirectional backscattering-immune topological
  electromagnetic states}, Nature, 461 (2009), pp.~772--775.

\end{thebibliography}

\end{document}